\documentclass{lmcs}
\pdfoutput=1

\usepackage{lastpage}
\lmcsdoi{15}{2}{18}
\lmcsheading{}{\pageref{LastPage}}{}{}%
{Apr.~24,~2018}{May~24,~2019}{}

\usepackage{amsfonts,amsmath}

\usepackage{mathtools,amssymb}
\usepackage[scr]{rsfso}

\usepackage{xspace,mathpartir,upgreek,accents,stmaryrd}
\usepackage{float,bm,etoolbox,url,tabularx}

\floatstyle{boxed}
\restylefloat{figure}

\usepackage{macros}
\usepackage{hyperref}

\title{Bisimulations for Delimited-Control Operators}

\begin{document}

\title[Bisimulations for Delimited-Control Operators]{Bisimulations for
  Delimited-Control Operators}

\author[D.~Biernacki]{Dariusz Biernacki\rsuper{a}}
\author[S.~Lenglet]{Sergue\"{\i} Lenglet\rsuper{b}}
\author[P.~Polesiuk]{Piotr Polesiuk\rsuper{a}}

\address{\lsuper{a}University of Wroc\l{}aw}
\email{\{dabi,ppolesiuk\}@cs.uni.wroc.pl}

\address{\lsuper{b}Universit\'e de Lorraine}
\email{serguei.lenglet@univ-lorraine.fr}


\thanks{This work was supported by PHC Polonium and by National
  Science Centre, Poland, grant no. 2014/15/B/ST6/00619.}

\begin{abstract}
  We present a comprehensive study of the behavioral theory of an
  untyped $\lambda$-calculus extended with the delimited-control
  operators \shiftId{} and \resetId{}. To that end, we define a
  contextual equivalence for this calculus, that we then aim to
  characterize with coinductively defined relations, called
  \emph{bisimilarities}. We consider different styles of
  bisimilarities (namely applicative, normal-form, and environmental)
  within a unifying framework, and we give several examples to
  illustrate their respective strengths and weaknesses. We also
  discuss how to extend this work to other delimited-control
  operators.

  % This work gathers within a single unifying framework the existing
  % bisimulation theories for delimited continuations along with a
  % number of substantial refinements and extensions.
\end{abstract}

\maketitle

\togglefalse{withoutApp}

\section{Introduction}%
\label{s:intro}
\subsubsection*{Delimited-control operators} Control operators for
delimited continuations enrich a programming language with the ability
to delimit the current continuation, to capture such a delimited
continuation, and to compose delimited continuations. Such operators
have been originally proposed independently by
Felleisen~\cite{Felleisen:POPL88} and by Danvy and
Filinski~\cite{Danvy-Filinski:LFP90}, with numerous variants designed
subsequently~\cite{Hieb-al:LaSC93,Moreau-Queinnec:PLILP94,Gunter-al:FPCA95,Dybvig-al:JFP06}. The
applications of delimited-control operators range from
non-deterministic
programming~\cite{Danvy-Filinski:LFP90,Kiselyov-al:ICFP05}, partial
evaluation~\cite{Lawall-Danvy:LFP94,Danvy:POPL96}, and normalization
by evaluation~\cite{Dybjer-Filinski:APPSEM00} to
concurrency~\cite{Hieb-al:LaSC93}, mobile code~\cite{Sumii:Scheme00},
linguistics~\cite{Shan:CW04}, operating
systems~\cite{Kiselyov-Shan:CONTEXT07}, and probabilistic
programming~\cite{Kiselyov-Shan:DSL09}. Several variants of
delimited-control operators are nowadays available in mainstream
functional languages such as Haskell~\cite{Dybvig-al:JFP06},
OCaml~\cite{Kiselyov:FLOPS10}, Racket~\cite{Flatt-al:ICFP07}, and
Scala~\cite{Rompf-al:ICFP09}.
%, and SML~\cite{Filinski:POPL94}.
%, and even in Coq~\cite{}.

The control operators \shiftId{} and
\resetId{}~\cite{Danvy-Filinski:LFP90} were designed to account for
the traditional model of non-deterministic programming based on
success and failure continuations~\cite{Wand-Vaillancourt:ICFP04}, and
their semantics as well as pragmatics take advantage of an extended
continuation-passing style (CPS), where the continuation of the
computation is represented by the current delimited continuation (the
success continuation) and a metacontinuation (the failure
continuation). The control delimiter \resetId{} resets the current
continuation, whereas the control operator \shiftId{} captures the
current continuation, which then can be either discarded (expressing
failure in a backtracking search) or duplicated (expressing a
backtracking point creation). When a captured continuation is resumed,
the then-current continuation is pushed on the metacontinuation
(representing a list of pending delimited continuations). For this
reason \shiftId{} and \resetId{} are known as static delimited-control
operators, as opposed, e.g., to \controlId{} and
\promptId{}~\cite{Felleisen:POPL88} that are dynamic, in that they
require an actual stack concatenation to compose continuations, and
for this reason go beyond the standard
CPS~\cite{Biernacki-al:SCP06}~\footnote{Expressing \controlId{} and
  \promptId{}, as well as other dynamic control operators, in terms of
  \shiftId{} and \resetId{} and, therefore, in CPS is possible, but it
  requires an involved continuation answer type relying on
  recursion~\cite{DBLP:journals/toplas/BiernackiDM15,Kiselyov:TR05,Shan:HOSC07}. }.

The static delimited-control operators have been surrounded by an
array of CPS-based semantic artifacts that greatly support programming
and reasoning about code, by making it possible to interpret effectful
programs in a purely functional language. As a matter of fact, most of
the applications of delimited-control listed above have been presented
using \shiftId{} and \resetId{}. But the connection with CPS is even
more intimate---in his seminal article~\cite{Filinski:POPL94},
Filinski showed that because the continuation monad can express any
other monad, \shiftId{} and \resetId{} can express any monadic effect
(such as exceptions or non-determinism) in direct style. Furthermore,
iterating the CPS transformation for a language with \shiftId{} and
\resetId{} leads to a CPS hierarchy~\cite{Danvy-Filinski:LFP90} which
in turn allows one to express layered computational effects in direct
style~\cite{Filinski:POPL99}. These results establish a special
position of \shiftId{} and \resetId{} among all the delimited-control
operators considered in the literature, even though from an
operational standpoint, \shiftId{} and \resetId{} can be easily
expressed in terms of dynamic control
operators~\cite{Biernacki-Danvy:JFP06}. Interestingly, the abortive
control operator \callccId{} known from Scheme and SML of New Jersey
requires the presence of mutable state to obtain the expressive power
of \shiftId{} and \resetId{}~\cite{Filinski:POPL94} (and of other
delimited-control operators).

%% The control operators \shiftId{} and
%% \resetId{}~\cite{Danvy-Filinski:LFP90} were designed to account for
%% the traditional model of non-deterministic programming based on
%% success and failure continuations, and their semantics as well as
%% pragmatics take advantage of an extended continuation-passing style
%% (CPS), where the continuation of the computation is represented by the
%% current delimited continuation (the success continuation) and a
%% metacontinuation (the failure continuation). In his seminal
%% articles~\cite{Filinski:POPL94,Filinski:POPL99}, Filinski showed that
%% because the continuation monad can express any other monad, \shiftId{}
%% and \resetId{} can express any monadic effect in direct style (DS),
%% which gives them a special position among all the control operators
%% considered in the literature. In particular, the control operator
%% \callccId{} known from Scheme and SML of New Jersey requires the
%% presence of mutable state to obtain the expressive power of \shiftId{}
%% and \resetId{}.

Relying on the CPS translation to a pure language is helpful and
inspiring when programming with \shiftId{} and \resetId{}, but it is
arguably more convenient to reason directly about the code with
control operators. To facilitate such reasoning, Kameyama et
al.\ devised direct-style axiomatizations for a number of
delimited-control
calculi~\cite{Kameyama-Hasegawa:ICFP03,Kameyama:HOSC07,Kameyama-Tanaka:PPDP10}
that are sound and complete with respect to the corresponding CPS
translations. Numerous other results concerning equational reasoning
in various calculi for delimited
continuations~\cite{Sabry:TR96,Ariola-al:HOSC07,Herbelin-Ghilezan:POPL08,
  Materzok:CSL13} show that it has been a topic of active research.

While the CPS-based equational theories are a natural consequence of
the denotational or translational semantics of control operators such
as \shiftId{} and \resetId{}, they are not strong enough to verify the
equivalences of programs that have unrelated images through the CPS
translation, but that operationally cannot be distinguished (e.g.,
take two different fixed-point combinators). In order to build a
stronger theory of program equivalence for delimited control, we turn
to the operational foundations of \shiftId{} and \resetId{}, and
consider operationally-phrased criteria for program equivalence. The
original operational semantics of \shiftId{} and \resetId{}, in the
form of an abstract machine and a corresponding context-sensitive
reduction semantics, has been derived through
defunctionalization~\cite{Biernacka-al:LMCS05}. Here the concepts of
the delimited continuation and the metacontinuation are materialized
as a stack and a metastack of the machine, and as a context and a
metacontext in the reduction semantics. A direct consequence of this
semantics is that no ``missing reset'' error can occur in the course
of program evaluation---a reset guarding the current delimited
continuation is always present. A \emph{relaxed} version of the
semantics, where a delimiter surrounding the context is not statically
ensured has also been considered in the
literature~\cite{Kameyama-Hasegawa:ICFP03} and in some
implementations~\cite{Filinski:POPL94}. This less-structured approach
sacrifices the direct correspondence with CPS for flexibility and it
scales better to other delimited-control operators.

\subsubsection*{Behavioral equivalences} Because of the complex nature of control
effects, it can be difficult to determine if two programs that use
\shiftId{} and \resetId{} are equivalent (i.e., behave in the same
way) or not. \emph{Contextual equivalence}~\cite{JHMorris:PhD} is
widely considered as the most natural equivalence on terms in
languages based on the $\lambda$-calculus. The intuition behind this
relation is that two programs are equivalent if replacing one by the
other in a bigger program does not change the behavior of this bigger
program. The behavior of a program has to be made formal by defining
the \emph{observable actions} we want to take into account for the
calculus we consider. It can be, \eg inputs and outputs for
communicating systems~\cite{Sangiorgi-Walker:01}, memory reads and
writes, etc. For the plain
$\lambda$-calculus~\cite{Abramsky-Ong:IaC93}, it is usually whether
the term terminates or not. The ``bigger program'' can be seen as a
\emph{context} (a term with a hole) and, therefore, two terms
$\tmzero$ and $\tmone$ are contextually equivalent if we cannot tell
them apart when evaluated within any context $\cctx$, \ie if $\inctx
\cctx \tmzero$ and $\inctx \cctx \tmone$ produce the same observable
actions.

The latter quantification over contexts $\cctx$ makes contextual
equivalence hard to use in practice to prove that two given terms are
equivalent. As a result, one usually looks for more tractable
alternatives to contextual equivalence, such as logical relations
(see, \eg~\cite{Pitts:MSCS00}), axiomatizations (see, \eg~\cite{Lanese-AL:IC11}), or \emph{bisimulations}. A bisimulation
relates two terms $\tmzero$ and $\tmone$ by asking them to mimic each
other in a coinductive way, \eg if $\tmzero$ reduces to a term
$\tmzero'$, then $\tmone$ has to reduce to a term $\tmone'$ so that
$\tmzero'$ and $\tmone'$ are still in the bisimulation, and conversely
for the reductions of $\tmone$. An equivalence on terms, called
\emph{bisimilarity} can be derived from a notion of bisimulation: two
terms are bisimilar if there exists a bisimulation which relates
them. Finding an appropriate notion of bisimulation consists in
finding the conditions on which two terms are related, so that the
resulting notion of bisimilarity is \emph{sound} and \emph{complete}
{\wrt}contextual equivalence, (\ie it is included in and it contains
contextual equivalence, respectively).

Different styles of bisimulations have been proposed for calculi
similar to the $\lambda$-calculus. For example, \emph{applicative}
bisimilarity~\cite{Abramsky-Ong:IaC93} relates terms by reducing them
to values (if possible), and the resulting values have to be
themselves applicative bisimilar when applied to an arbitrary
argument. As we can see, applicative bisimilarity still contains some
quantification over arguments to compare values, but is nevertheless
easier to use than contextual equivalence because of its coinductive
nature---bisimulation relations are constructed incrementally,
following a step-by-step analysis of the possible interactions of the
program with its environment, and also because we do not have to
consider all forms of contexts. When sound, applicative bisimilarity
is usually also complete {\wrt}contextual equivalence, at least for
deterministic languages such as the plain
$\lambda$-calculus~\cite{Abramsky-Ong:IaC93}.

\emph{Environmental
  bisimilarity}~\cite{Sangiorgi-al:LICS07,Sangiorgi-al:TOPLAS11} is
quite similar to applicative bisimilarity, as it compares terms by
reducing them to values, and then requires the resulting values to be
bisimilar when applied to some arguments. However, the arguments are
no longer arbitrary, but built using an \emph{environment}, which
represents the knowledge accumulated so far by an outside observer on
the tested terms. Like applicative bisimilarity, environmental
bisimilarity is usually sound and complete, but it also allows for
\emph{up-to techniques} to simplify its equivalence proofs. The idea
behind up-to techniques is to define relations that are not exactly
bisimulations but are included in bisimulations. Finding an up-to
relation equating two given terms is usually simpler than finding a
regular bisimulation relating these terms. Unlike for environmental
bisimilarity, the definition of useful up-to techniques for
applicative bisimilarity remains an open problem.

In contrast to applicative and environmental bisimilarity,
\emph{normal-form} bisimilarity~\cite{Lassen:LICS05} (also called
\emph{open} bisimilarity in~\cite{Sangiorgi:LICS92}) does not contain
any quantification over arguments or contexts in its definition. The
principle is to reduce the compared terms to normal forms (if
possible), and then to decompose the resulting normal forms into
sub-components that have to be themselves bisimilar. Unlike
applicative or environmental bisimilarity, normal-form bisimilarity is
usually not complete, \ie there exist contextually equivalent terms
that are not normal-form bisimilar. But because of the lack of
quantification over contexts, proving that two terms are normal-form
bisimilar is usually quite simple, and the proofs can be further
simplified with the help of up-to techniques (like with environmental
bisimilarity).

\subsubsection*{This work} In this article, we present a comprehensive study of
the behavioral theory of a $\lambda$-calculus extended with the operators
\shiftId{} and \resetId{}, called $\lamshift$. In previous works, we defined
ap\-plic\-at\-ive~\cite{Biernacki-Lenglet:FOSSACS12},
normal-form~\cite{Biernacki-Lenglet:FLOPS12,Biernacki-al:MFPS17}, and
environmental~\cite{Biernacki-Lenglet:APLAS13,Aristizabal-al:LMCS17}
bisimilarities for this calculus. Here we present these results in a systematic
and uniform way, with examples allowing for comparisons between the different
styles of bisimulation. In particular, we compare bisimilarities to Kameyama and
Hasegawa's direct style axiomatization of
$\lamshift$~\cite{Kameyama-Hasegawa:ICFP03}, and we use these axioms as examples
throughout the paper. We consider two semantics for~$\lamshift$, one that is
faithful to its defining CPS translation, where terms are evaluated within an
outermost \resetId{} (we call it the ``original semantics''), and another one
where this requirement is lifted (we call it the ``relaxed
semantics''). Finally, we discuss how this work can be extended to other
delimited-control operators.

\subsubsection*{Structure of the article} Section~\ref{s:calc} presents the
syntax and semantics of the calculus $\lamshift$ with \shiftId{} and
\resetId{} that we use in this paper. We also recall the definition of
CPS equivalence, a CPS-based equivalence between terms, and its
axiomatization. Section~\ref{s:ctx-eqv} discusses the definition of a
contextual equivalence for $\lamshift$, and its relationship with CPS
equivalence. We look for (at least sound) alternatives of this
contextual equivalence by considering several styles of
bisimilarities: applicative in Section~\ref{s:app}, environmental in
Section~\ref{s:env}, and normal-form in
Section~\ref{s:nf}. Section~\ref{s:extensions} discusses the possible
extensions of our work to other semantics and other calculi with
delimited control, and Section~\ref{s:conclusion} concludes this
paper. In particular, we summarize in Figure~\ref{f:conclusion} the
relationships between all the behavioral equivalences defined in this
paper. We discuss related work---in particular, our own previous
work---in the relevant sections, \eg related work on applicative
bisimilarities for control operators is discussed at the beginning of
Section~\ref{s:app}.

\subsubsection*{Notations and basic definitions} We use the following notations
frequently throughout the paper. We write $\is$ for a defining
equality, \ie $m \is e$ means that $m$ is defined as the expression
$e$. Given a metavariable $m$, we write $\vect m$ for a sequence of
entities denoted by $m$. Given a binary relation $\rel$, we write $m
\rel m'$ for $(m, m') \mathop{\in} \rel$, $\inv\rel$ for its inverse,
defined as $\inv\rel {\mathop{\is}} \{(m', m) \mmid m \rel m'\}$, and
$\rel^*$ for its transitive and reflexive closure, defined as
%% $\rel^* \is \{ (m, m) \} \cup \{ (m_1, m_n) \mmid n \geq 1 \conjun
%% \forall i \in [1,n-1], m_i \rel m_{i+1}\}$.
\[\rel^* {\mathop{\is}} \{ (m, m') \mmid \exists k, m_1,\dots,m_k, k \geq 0 \conjun m
= m_0 \conjun m_k = m' \conjun \forall 0 \leq i < k, m_i \rel
m_{i+1}\}.\] Further, given two binary relations $\rel$ and $\relS$ we
use juxtaposition $\rel\relS$ for their composition, defined as
$\rel\relS {\mathop{\is}} \{(m, m') \mmid \exists m'', m \rel m'' \conjun m'' \relS
m'\}$. Finally, a relation $\rel$ is \emph{compatible} if it is preserved by all
the operators of the language, \eg $\tmzero \rel \tmone$ implies $\lam \varx
\tmzero \rel \lam \varx \tmone$; a relation is a \emph{congruence} if it is a
compatible equivalence relation.

% Local variables:
% mode: latex
% TeX-master: "journal.tex"
% End:

\section{The Calculus}%
\label{s:calc}
In this section, we present the syntax, reduction semantics, and CPS
equivalence for the language $\lamshift$ studied throughout this
article. The operators \textshift and \textreset have been originally
defined and have then been usually studied and implemented with a
call-by-value semantics; e.g., almost all the references we give in
Section~\ref{s:intro} use such a semantics. We therefore choose to
work with call by value in the main devolpments of this article, and
only briefly discuss call by name in Section~\ref{ss:cbn}.

\subsection{Syntax}

The language $\lamshift$ extends the call-by-value $\lambda$-calculus
with the delimited-control operators \shiftId{} and
\resetId{}~\cite{Danvy-Filinski:LFP90}. We assume we have a set of
term variables, ranged over by $\varx$, $y$, $z$, and $\vark$.
We use the metavariable $k$ for \shiftId{}-bound variables representing a
continuation, while $x$, $y$, and $z$ stand for the usual lambda-bound
variables representing any values; we believe such a distinction helps
to understand examples and reduction rules.

The syntax of terms ($\terms$) and values ($\values$) is given by the
following grammars:
\begin{grammar}
  &\textrm{Terms:} \quad & \tm  &\bnfdef  \val \bnfor \app \tm \tm \bnfor
  \shift \vark \tm \bnfor \reset \tm  \\
  &\textrm{Values:} \quad & \val &\bnfdef \varx \bnfor \lam \varx \tm
\end{grammar}
The operator \shiftId{} ($\shift \vark \tm$) is a capture operator,
the extent of which is determined by the delimiter \resetId{}
($\rawreset$). A $\lambda$-abstraction $\lam \varx \tm$ binds $\varx$
in~$\tm$ and a \shiftId{} construct $\shift \vark \tm$ binds $\vark$
in~$\tm$; terms are equated up to $\alpha$-conversion of their bound
variables. The set of free variables of $\tm$ is written $\fv \tm$; a
term $\tm$ is \emph{closed} if $\fv \tm = \emptyset$. The set of
closed terms (values) is noted~$\cterms$ ($\cvalues$, respectively).

We distinguish several kinds of contexts, represented outside-in, as
follows:
\begin{grammar}
  &\textrm{Pure contexts:} \quad & \ctx & \bnfdef \mtctx \bnfor \vctx \val \ctx \bnfor \apctx \ctx
  \tm \\
  &\textrm{Evaluation contexts:} \quad & \rctx & \bnfdef \mtctx \bnfor \vctx \val \rctx \bnfor \apctx \rctx
  \tm \bnfor \resetctx \rctx \\
  &\textrm{Contexts:} & \cctx & \bnfdef \mtctx \bnfor \lam \varx \cctx
  \bnfor \vctx \tm \cctx \bnfor \apctx \cctx \tm \bnfor \shift \vark \cctx \bnfor \resetctx
  \cctx
\end{grammar}
Regular contexts are ranged over by $\cctx$. The pure evaluation contexts
($\pevctxts$) (abbreviated as pure contexts),\footnote{This terminology comes
  from Kameyama (\eg in~\cite{Kameyama-Hasegawa:ICFP03}); note that we use the
  metavariables of~\cite{Biernacka-al:LMCS05} for evaluation contexts, which are
  reversed compared to~\cite{Kameyama-Hasegawa:ICFP03}.} ranged over by $\ctx$,
represent delimited continuations and can be captured by the \shiftId{}
operator. The call-by-value evaluation contexts, ranged over by $\rctx$,
represent arbitrary continuations and encode the chosen reduction
strategy. Filling a context $\cctx$ ($\ctx$,~$\rctx$) with a term $\tm$ produces
a term, written $\inctx \cctx \tm$ ($\inctx \ctx \tm$, $\inctx \rctx \tm$,
respectively); the free variables of $\tm$ may be captured in the process. We
extend the notion of free variables to contexts (with $\fv \mtctx=\emptyset$),
and we say a context $\cctx$ ($\ctx$, $\rctx$) is \emph{closed} if
$\fv \cctx = \emptyset$ ($\fv \ctx = \emptyset$, $\fv \rctx = \emptyset$,
respectively). The set of closed pure contexts is noted $\cpevctxts$.  In any
definitions or proofs, we say a variable is \emph{fresh} if it does not occur
free in the terms or contexts under consideration.

\subsection{Reduction Semantics}%
\label{ss:reduction}

The reduction semantics of $\lamshift$ is defined by the following
rules, where $\subst \tm \varx \val$ is the usual capture-avoiding
substitution of $\val$ for $\varx$ in~$\tm$:
\[
\begin{array}{rlll}
  \inctx \rctx {\app {\lamp \varx \tm} \val} & \redcbv & \inctx
  \rctx {\subst \tm \varx \val} & \quad \RRbeta\\[2mm]
\inctx \rctx {\reset{\inctx \ctx {\shift \vark \tm}}} &
  \redcbv & \inctx \rctx{\reset{\subst \tm \vark
      {\lam \varx {\reset {\inctx \ctx \varx}}}}} \mbox{ with } \varx \notin \fv{\ctx}
  & \quad \RRshift
 \\[2mm]
 \inctx \rctx {\reset \val} & \redcbv & \inctx \rctx \val & \quad \RRreset
\end{array}
\]
The term $\app {\lamp \varx \tm} \val$ is the usual call-by-value
redex for $\beta$-reduction (rule $\RRbeta$). The operator $\shift
\vark \tm$ captures its surrounding context $\ctx$ up to the
dynamically nearest enclosing \resetId{}, and substitutes $\lam \varx
{\reset {\inctx \ctx \varx}}$ for $\vark$ in $\tm$ (rule
$\RRshift$). If a \resetId{} is enclosing a value, then it has no purpose
as a delimiter for a potential capture, and it can be safely removed
(rule $\RRreset$). All these reductions may occur within a metalevel
context $\rctx$. The chosen call-by-value evaluation strategy is
encoded in the grammar of the evaluation contexts. Furthermore, the
reduction relation $\redcbv$ is compatible with evaluation contexts
$\rctx$, i.e., $\inctx \rctx \tm \redcbv \inctx \rctx {\tm'}$ whenever
$\tm \redcbv {\tm'}$. We write $\tm \redcbv$ when there is a $\tm'$
such that $\tm \redcbv {\tm'}$ and we write $\tm \not \redcbv$ when
no such~$\tm'$ exists.

All along the article, we use the terms $i \is \lam \varx \varx$,
$\omega \is \lam \varx {\app \varx \varx}$, and $\Omega \is \omega \iapp \omega$
to build examples, starting with the next one.

\begin{exa}%
  \label{e:reduction}
  We present the sequence of reductions initiated by
  $\reset {\app{\appp {(\shift {\vark_1}{\app i {\appp {\vark_1} i}})}{\shift
        {\vark_2} \omega}} \Omega}$. The term
  $\shift {\vark_1}{\app i {\appp {\vark_1} i}}$ is within the pure context
  $\ctx \is \apctx {(\apctx \mtctx {\shift {\vark_2} \omega})} {\Omega}$,
  enclosed in a delimiter $\rawreset$, so $\ctx$ is captured according to rule
  $\RRshift$:
  \[
    \reset {\app{\appp {(\shift{\vark_1}{\app i {\appp {\vark_1}
            i}})}{\shift {\vark_2} \omega}}{\Omega}}
  \redcbv \reset{\app i {\appp {\lamp \varx {\reset {\app{\appp \varx
              {\shift {\vark_2} \omega}}{\Omega}}}} i}}
  \]
  The role of reset in $\lam \varx {\reset{\inctx \ctx \varx}}$ is
  more clear after reduction of the $\beta_v$-redex $\app {\lamp \varx
    {\reset{\inctx \ctx \varx}}} i$:
  \[
    \reset{\app i {\appp {\lamp \varx {\reset {\app{\appp \varx
              {\shift {\vark_2} \omega}}{\Omega}}}} i}} \redcbv \reset
  {\app i {\reset{\app {\appp i {\shift {\vark_2}
            \omega}}{\Omega}}}}
  \]
  When the captured context $\ctx$ is
  reactivated, it is not simply \emph{concatenated} with the context
  $\vctx i \mtctx$, but \emph{composed} thanks to the reset enclosing
  $\ctx$. (This operation corresponds to continuation composition in
  the CPS semantics of \shiftId{} and \resetId{}, and it is crucially
  different from context concatenation~\cite{Biernacki-al:SCP06}.) As
  a result, the capture triggered by $\shift {\vark_2} \omega$ leaves
  the term $i$ outside the first enclosing reset intact:
  \[
    \reset {\app i {\reset{\app {\appp i {\shift {\vark_2}
            \omega}}{\Omega}}}} \redcbv \reset {\app i {\reset \omega}}
  \]
  Because $\vark_2$ does not occur in $\omega$, the context
  $\apctx {(\vctx i \mtctx)} {\Omega}$ is discarded when captured by
  $\shift {\vark_2} \omega$. Finally, we remove the useless delimiter
  $\reset {\app i {\reset \omega}} \redcbv \reset{\app i \omega}$ with rule
  $\RRreset$, and we then $\beta_v$-reduce and remove the last delimiter
  $\reset{\app i \omega} \redcbv \reset{\omega} \redcbv \omega$. Note that while
  the reduction strategy is call-by-value, some function arguments are not
  evaluated, like the non-terminating term $\Omega$ in this example.
\end{exa}

\begin{exa}[fixed-point combinators]%
  \label{e:combinators}

  We recall the definition of Turing's and Curry's fixed-point
  combinators. Let $\theta \is \lam {\varx y}{\app y {\lamp z {\app
        {\app {\app \varx \varx} y} z}}}$ and $\delta_\varx \is \lam y
  {\app \varx {\lamp z {\app {\app y y} z}}}$; then $\Turing \is \app
  \theta \theta$ is Turing's call-by-value fixed-point combinator, and
  $\Curry \is \lam \varx {\app{\delta_\varx}{\delta_\varx}}$ is
  Curry's call-by-value fixed-point combinator. In~\cite{Danvy-Filinski:DIKU89}, the authors propose variants of these
  combinators using \shiftId{} and \resetId{}. They write Turing's combinator as
  $\reset {\app \theta {\shift \vark {\app \vark \vark}}}$ and Curry's
  combinator as $\lam \varx {\reset {\app {\delta_\varx}{\shift \vark
        {\app \vark \vark}}}}$. For an example, the following
  reduction sequence demonstrates the behavior of the former:
  \[
  \reset{\app \theta {\shift \vark {\app \vark \vark}}}
  \redcbv
  \reset{\app{\lamp{x}{\reset{\app{\theta}{x}}}}{\lamp{x}{\reset{\app{\theta}{x}}}}}
  %% \redcbv
  %% \reset{\reset{\app{\theta}{\lamp{x}{\reset{\app{\theta}{x}}}}}}
  \clocbv
  \lam{y}{\app{y}{\lamp{z}
      {\app{\app{\app{\lamp{x}{\reset{\app{\theta}{x}}}}{\lamp{x}{\reset{\app{\theta}{x}}}}} y} z}}}
  \]
  We use the combinators and their delimited-control variants as examples for
  the equivalence proof techniques we define throughout the paper.
\end{exa}

\begin{rem}
  The context capture can also be written using local reduction
  rules~\cite{Felleisen:POPL88}, where the context is consumed piece
  by piece. We discuss these reduction rules and their consequences on
  the results of this article in Section~\ref{ss:local-rules}.
\end{rem}

There exist terms which are not values and which cannot be reduced any
further; these are called \emph{stuck terms}.
\begin{defi}
  A term $\tm$ is stuck if $\tm$ is not a value and $\tm \not
  \redcbv$.
\end{defi}
\noindent For example, the term $\inctx \ctx {\shift \vark \tm}$ is stuck
because there is no enclosing \resetId{}; the capture of $\ctx$ by the
\shiftId{} operator cannot be triggered. In fact, stuck terms are easy
to characterize.
\begin{prop}%
  \label{l:stuck}
  A term $\tm$ is stuck iff
  \begin{itemize}
  \item $\tm = \inctx \ctx {\shift \vark {\tm'}}$
    for some $\ctx$, $k$, and $\tm'$, or
  \item $\tm = \inctx \rctx {\app \varx \val}$ for some $\rctx$, $\varx$, and
    $\val$.
  \end{itemize}
\end{prop}

\begin{proof}[Sketch]
  The ``if'' part is straightforward. The ``only if'' part is by induction on
  $\tm$; we detail the application case $\tmzero \iapp \tmone$. If
  $\tmzero$ is stuck, we can conclude with the induction hypothesis. Otherwise,
  $\tmzero$ is a value. If $\tmone$ is stuck, then we can conclude with the
  induction hypothesis again. If $\tmone$ is a value $\valone$, then $\tmzero$
  is not a $\lambda$-abstraction, so we have an open-stuck term $\app \varx
  \valone$ for some $\varx$.
\end{proof}

\noindent We call \emph{control-stuck terms} the terms of the form $\inctx \ctx
{\shift \vark \tm}$ and \emph{open-stuck terms} the terms of the form
$\inctx \rctx {\app \varx \val}$.
\begin{defi}
  A term $\tm$ is a normal form, if $\tm$ is a value or a stuck term.
\end{defi}

We call \emph{redexes} (ranged over by $\redex$) terms of the form
$\app{\lamp \varx \tm} \val$, $\reset {\inctx \ctx {\shift \vark
    \tm}}$, and $\reset \val$. Thanks to the following
unique-decomposition property, the reduction relation $\redcbv$ is
deterministic.
\begin{prop}%
  \label{p:unique-decomp}
  For all terms $\tm$, either $\tm$ is a normal form, or there exist a
  unique redex~$\redex$ and a unique context $\rctx$ such that $\tm =
  \inctx \rctx \redex$.
\end{prop}

\begin{proof}[Sketch]
  By induction on $\tm$. For example, in the case $\tm = \reset {\tm'}$, either
  $\tm'$ itself is reducing, and we can conclude with the induction hypothesis,
  or $\tm'$ is a value $\val$, and we have a redex $\reset\val$ in the empty
  context.
\end{proof}

\noindent Finally,
we define the evaluation relation of $\lamshift$ as
follows.
\begin{defi}
  We write $\tm \evalcbv \tm'$ if $\tm \clocbv \tm'$ and $\tm'$ is  a
  normal form.
\end{defi}
\noindent If a term $\tm$ admits an infinite reduction sequence, like $\Omega$,
we say it \emph{diverges}, written $\tm \divcbv$.

In the rest of the paper, we use the following results on the reduction (or
evaluation) of terms: a control stuck term cannot be obtained from a term
of the form $\reset\tm$, and reduction is preserved by substitution.

\begin{prop}%
  \label{p:eval-reset}
  If $\reset{\tm} \evalcbv \tm'$ then $\tm'$ is a value or an open stuck term of
  the form $\reset{\inctx \rctx {\app \varx \val}}$. (If $\tm$ is closed then
  $\tm'$ can only be a closed value.)
\end{prop}

\begin{proof}[Sketch]
  By case analysis on the reduction rules, $\reset \tm \redcbv \tm'$ implies
  $\tm'$ is a value or $\tm' = \reset{\tm''}$ for some $\tm''$. So if
  $\reset{\tm} \evalcbv \tm'$ then $\tm'$ is a value, or a normal form
  $\reset{\tm''}$ for some $\tm''$. By Proposition~\ref{l:stuck}, if $\tm'$ is
  not a value, it is either control-stuck or open-stuck, but a control-stuck
  term cannot have an outermost reset, so $\tm'$ is necessarily open-stuck.
\end{proof}

\begin{prop}%
  \label{l:redcbv-subst}
  If $\tm \redcbv \tm'$, then $\subst \tm \varx \val \redcbv \subst
  {\tm'} \varx \val$.
\end{prop}

\begin{proof}[Sketch]
  By case analysis on the reduction rules.
\end{proof}

\subsection{The original reduction semantics}%
\label{ss:orig-sem}

Let us notice that the reduction semantics we have introduced does not require
terms to be evaluated within a top-level \resetId{}---a requirement that is
commonly relaxed in practical implementations of \shiftId{} and
\resetId{}~\cite{Dybvig-al:JFP06,Filinski:POPL94}, but also in some other
studies of these operators~\cite{Asai-Kameyama:APLAS07,Kameyama:HOSC07}. This is
in contrast to the original reduction semantics for \shiftId{} and
\resetId{}~\cite{Biernacka-al:LMCS05} that has been obtained from the 2-layered
continuation-passing-style (CPS) semantics~\cite{Danvy-Filinski:LFP90},
discussed in Section~\ref{ss:cps-equivalence}. A consequence of the
correspondence with the CPS-based semantics is that terms in the original
reduction semantics are treated as complete programs and are decomposed into
triples consisting of a subterm (a value or a redex), a delimited context, and a
meta-context (a list of delimited contexts), resembling abstract machine
configurations. Such a decomposition imposes the existence of an implicit
top-level \resetId{}, hard-wired in the decomposition, surrounding any term to
be evaluated.

While the relaxed semantics scales better to calculi with multiple
prompts~\cite{Dybvig-al:JFP06}, the original one lends itself to a
generalization to a hierarchy of delimited-control
operators~\cite{Biernacka-al:LMCS05}; see Section~\ref{ss:others} for
more details about these extensions. The two semantics differ in that
the original semantics does not allow for control-stuck
terms. However, it can be easily seen that operationally the
difference is not essential---they are equivalent when it comes to
terms of the form $\reset{\tm}$. In the rest of the article we call
such terms \emph{delimited terms} and we use the relaxed semantics
when analyzing their behavior.

The top-level \resetId{} requirement, imposed by the original semantics, does
not lend itself naturally to the notion of applicative bisimulation that we
propose for the relaxed semantics in Section~\ref{s:app}. We show, however, that
the requirement can be successfully treated in the framework of environmental
and normal-form bisimulations, presented in Sections~\ref{ss:env-LTS-original}
and~\ref{ss:nf-original}.

\subsection{CPS Equivalence}%
\label{ss:cps-equivalence}

The operators \shiftId{} and \resetId{} have been originally defined
by a translation into continuation-passing
style~\cite{Danvy-Filinski:LFP90} that we present in
Figure~\ref{f:cps-translation}. Translated terms expect two
continuations: the delimited continuation representing the rest of the
computation up to the dynamically nearest enclosing delimiter, and the
metacontinuation representing the rest of the computation beyond this
delimiter. In the first three equations the metacontinuation $k_2$
could be $\eta$-reduced, yielding Plotkin's familiar CBV CPS
translation~\cite{Plotkin:TCS75}. In the equation for \resetId{}, the
current delimited continuation $k_1$ is moved to the metacontinuation
and the delimited term receives the initial delimited continuation. In
the equation for \shiftId{}, the current continuation is captured (and
reinitialized) as a lambda abstraction that when applied pushes the
then-current delimited continuation on the metacontinuation, and
applies the captured continuation to the argument. A CPS-transformed
program is run with the initial delimited continuation $\kinit$ and
the identity metacontinuation.

\begin{figure}
  \[
  \begin{array}{rcl}
    \cps{\varx} &=& \lam {\vark_1 \vark_2}{\app {\app {\vark_1} \varx} \vark_2}
    \\[1mm]
    \cps{\lam \varx \tm} &=& \lam {\vark_1 \vark_2}{\app{\app {\vark_1}{\lamp
          \varx {\cps \tm}}}{\vark_2}}
          \\[1mm]
    \cps{\app \tmzero \tmone} &=& \lam {\vark_1 \vark_2}{\app{\app {\cps \tmzero}{\lamp
          {\varx_0 \vark'_2}{\app{\app {\cps \tmone}{\lamp{\varx_1
                  \vark''_2}{\app{\app{\app
                    {\varx_0}{\varx_1}}{\vark_1}}{\vark''_2}}}}{\vark'_2}}}}{\vark_2}} \\[1mm]
    \cps{\reset \tm} &=& \lam {\vark_1 \vark_2}{\app{\app {\cps \tm}
        \kinit}{\lamp \varx {\app {\app {\vark_1} \varx}{\vark_2}}}} \\[1mm]
    \cps{\shift \vark \tm} &=& \lam {\vark_1 \vark_2}{\app{\app {\subst {\cps
            \tm} \vark {\lamp{\varx_1 \vark'_1 \vark'_2}{\app{\app
                {\vark_1}{\varx_1}}{\lamp {\varx_2}{\app{\app
                    {\vark'_1}{\varx_2}}{\vark'_2}}}}}} \kinit}{\vark_2}}\\[1mm]
    & & \mbox{with } \kinit = \lam {\varx \vark_2}{\app {\vark_2} \varx}
  \end{array}
  \]
\caption{Definitional CPS translation of $\lamshift$}%
\label{f:cps-translation}
\end{figure}

\begin{figure}
  \[
  \begin{array}{rcll}
    \app{\lamp \varx \tm} \val &=& \subst \tm \varx \val & \quad \AXbeta
    \\[1mm]
    \app {\lamp \varx {\inctx \ctx \varx}} \tm &=& \inctx \ctx \tm \mbox{ if } x
    \notin \fv \ctx & \quad \AXbetaomega
    \\[1mm]
    \reset {\inctx \ctx {\shift \vark \tm}} &=& \reset {\subst \tm \vark {\lam
        \varx {\reset {\inctx \ctx \varx}}}} \mbox{ if } x
    \notin \fv \ctx & \quad \AXresetshift
    \\[1mm]
    \reset {\app {\lamp \varx \tmzero}{\reset \tmone}} &=& \app {\lamp \varx
      {\reset \tmzero}}{\reset \tmone} & \quad \AXresetlift
    \\[1mm]
    \reset \val &=& \val & \quad \AXresetval
    \\[1mm]
    \shift \vark {\reset \tm} &=& \shift \vark \tm & \quad \AXshiftreset
    \\[1mm]
    \lam \varx {\app \val \varx} &=& \val \mbox{ if } \varx \notin \fv \val &
    \quad \AXetav
    \\[1mm]
    \shift \vark {\app \vark \tm} &=& \tm \mbox{ if } \vark \notin \fv
    \tm & \quad \AXshiftelim
  \end{array}
  \]
  \caption{Kameyama and Hasegawa's axiomatization of $\lamshift$}%
  \label{f:axioms}
\end{figure}

% For example, CPS translating the term
% $\app{x}{\reset{\app{y}{\shift{k}{\app{z}{(\app{k}{x'})}}}}}$ and
% $\beta$-reducing the administrative
% redexes~\cite{Danvy-Filinski:MSCS92-one} to avoid clutter, we obtain
% \[
% \lam{k_1 k_2}{\app{\app{\app{(\lam{x_2 k_1' k_2'}{\app{\app{\app{y}{x_2}}{\kinit}}
%           {(\lam{x_3}{\app{\app{k_1'}{x_3}}{k_2'}})}})}
%       {x'}}{(\lam{x_1 k_2'}{\app{\app{\app{z}{x_1}}{\kinit}}{k_2'}})}}
%   {(\lam{x_0}{\app{\app{\app{x}{x_0}}{k_1}}{k_2}})}}
% \]
% where the computations are sequentialized according to the evaluation
% strategy in the source calculus.

The CPS translation for $\lamshift$ induces the following notion
of equivalence on $\lamshift$ terms.

\begin{defi}
  Two terms $\tm$ and $\tm'$ are CPS equivalent, written $\tm \cpsequiv \tm'$, if
  their CPS translations are $\beta\eta$-convertible, where
  $\beta\eta$-convertibility is the smallest congruence containing the relations
  $\red{\beta}$ and $\red{\eta}$:
  \[
  \begin{array}{rcll}
    \app{(\lam{\varx}{\tm})}{\tm'} & \red{\beta} & \subst{\tm}{\varx}{\tm'} &
    \\
    \lam{\varx}{\app{\tm}{\varx}} & \red{\eta} & \tm & \mbox{ if } \varx \notin \fv{\tm}
  \end{array}
  \]
\end{defi}
For example, the reduction rules $\tm \redcbv \tm'$ given in Section~\ref{ss:reduction} are sound {\wrt}CPS
because CPS translating $\tm$
and $\tm'$ yields $\beta \eta$-convertible terms in the
$\lambda$-calculus.
The CPS equivalence has been characterized in terms of direct-style
equations by Kameyama and Hasegawa, who developed a sound and complete
axiomatization of \shiftId{} and
\resetId{}~\cite{Kameyama-Hasegawa:ICFP03}: two terms are
CPS equivalent iff one can derive their equality using the equations
of Figure~\ref{f:axioms}.

The axiomatization is a source of examples for the bisimulation techniques that
we study in Sections~\ref{s:app},~\ref{s:env}, and~\ref{s:nf}, and it allows us
to relate the notion of CPS equivalence to the notions of contextual equivalence
that we introduce in Section~\ref{s:ctx-eqv}. In particular, we show that all
but one axiom are validated by the bisimilarities for the relaxed semantics, and
that all the axioms are validated by the equivalences of the original
semantics. The discriminating axiom that confirms the discrepancy between the
two semantics is $\AXshiftelim$---the only equation that hinges on the existence
of the top-level \resetId{}.

It might be possible to consider alternative CPS translations for
\shiftId{} and \resetId{}, e.g., as given
in~\cite{Materzok-Biernacki:ICFP11}, that correspond to the relaxed
semantics. Such CPS translations require a recursive structure of
continuations, which makes it hard to reason about the image of the
translations, and, moreover, the operational correspondence between
the relaxed semantics and such CPS translations is not as tight as
between the original semantics and the original CPS translation
considered in this section. Devising a respective axiomatization to be
validated by the bisimilarity theories presented in this work is a
research path beyond the scope of the present article.

% Local variables:
% mode: latex
% TeX-master: "journal.tex"
% End:

\section{Contextual Equivalence}%
\label{s:ctx-eqv}
Studying the behavioral theory of a calculus usually starts by the definition of
a Morris-style contextual equivalence~\cite{JHMorris:PhD}. As usual, the idea is
to express that two terms are equivalent if and only if they cannot be
distinguished when put in an arbitrary context. The question is then which
behaviors to observe in $\lamshift$ for each of the two semantics considered in
this paper.

\subsection{Definition for the Relaxed Semantics}%
\label{ss:ctx-relax}

We first discuss the definition of contextual equivalence for closed terms,
before extending it to open terms. As in the regular $\lambda$-calculus, we
could observe only if a term reduces to a value or not, leading to the following
relation.

\begin{defi}
  Let $\tmzero$, $\tmone$ be closed terms. We write $\tmzero \ctxequivone
  \tmone$ if for all closed~$\cctx$, $\inctx \cctx \tmzero \evalcbv \valzero$
  for some $\valzero$ implies $\inctx \cctx \tmone \evalcbv \valone$ for some
  $\valone$, and conversely for $\inctx \cctx \tmone$.
\end{defi}

But in $\lamshift$, the evaluation of closed terms generates not only
values, but also control stuck terms. Taking this into account, a more
fine-grained definition of contextual equivalence would be as follows.

\begin{defi}%
  \label{d:ctxequivtwo}
  Let $\tmzero$, $\tmone$ be closed terms. We write $\tmzero \ctxequivtwo
  \tmone$ if for all closed~$\cctx$,
  \begin{itemize}
  \item $\inctx \cctx \tmzero \evalcbv \valzero$ for some $\valzero$ iff $\inctx
    \cctx \tmone \evalcbv \valone$ for some $\valone$;
  \item $\inctx \cctx \tmzero \evalcbv \tmzero'$ for some control stuck term
    $\tmzero'$ iff $\inctx \cctx \tmone \evalcbv \tmone'$ for some control stuck
    term $\tmone'$.
  \end{itemize}
\end{defi}

\noindent This definition can actually be formulated in a simpler way, where we do not
distinguish cases based on the possible normal forms.

\begin{prop}%
  \label{p:ctx-usual-def}
  We have $\tmzero \ctxequivtwo \tmone$ iff for all closed~$\cctx$, $\inctx
  \cctx \tmzero \evalcbv$ iff $\inctx \cctx \tmone \evalcbv$.
\end{prop}

\begin{proof}
  Suppose that $\inctx \cctx \tmzero \evalcbv$ iff $\inctx \cctx \tmone
  \evalcbv$ holds. We prove that we have $\tmzero \ctxequivtwo \tmone$ (the
  reverse implication is immediate). Assume there exists $\cctx$ such that
  $\inctx \cctx \tmzero \evalcbv \tmzero'$ with $\tmzero'$ control stuck, and
  $\inctx \cctx \tmone \evalcbv \valone$. Then $\app {\inctx \cctx \tmzero}
  \Omega \evalcbv \app {\tmzero'} \Omega$ ($\app {\tmzero'} \Omega$ is control
  stuck), and $\app{\inctx \cctx \tmone} \Omega \clocbv \app \valone \Omega
  \divcbv$. The context $\app \cctx \Omega$ distinguishes $\tmzero$ and
  $\tmone$, hence a contradiction. Therefore, if $\inctx \cctx \tmzero$
  evaluates to a control stuck term, then so does $\inctx \cctx \tmone$, and
  similarly for evaluation to values.
\end{proof}

By the definitions, it is clear that $\ctxequivtwo \mathop{\subseteq}
\ctxequivone$. The inclusion is strict, because of terms such as
$\shift \vark \Omega$, which are control-stuck terms but diverge when
unstuck. Indeed, we have $\shift \vark \Omega \not\ctxequivtwo
\Omega$, because $\shift \vark \Omega$ is a stuck term, but not
$\Omega$ and, therefore, the second item of Definition~\ref{d:ctxequivtwo} is violated. However, they are related by
$\ctxequivone$.

\begin{prop}
  We have $\shift \vark \Omega \ctxequivone \Omega$.
\end{prop}

\begin{proof}
  Let $\cctx$ be such that
  $\inctx \cctx {\shift \vark \Omega} \evalcbv \valzero$ for some
  $\valzero$. Then we prove that $\inctx \cctx \Omega$ reduces to a value as
  well; in fact, $\cctx$ does not evaluate the term that fills its hole. We
  define multi-holes contexts $\hctx$ by the following grammar
  \begin{align*}
  \hctx & \bnfdef \mtctx \bnfor \varx \bnfor  \lam \varx \hctx
  \bnfor \app \hctx \hctx \bnfor \shift \vark \hctx \bnfor \resetctx
  \hctx
  \end{align*}
  and we write $\inctx \hctx \tm$ for the plugging of $\tm$ in all the holes of
  $\hctx$. We show that $(*)$ for all $\hctx$ and $\val$,
  $\inctx \hctx {\shift \vark \Omega} \evalcbv \val$ implies that there exists
  $\hctx'$ such that $\val = \lam \varx {\inctx {\hctx'}{\shift \vark \Omega}}$
  and $\inctx \hctx {\Omega} \evalcbv \lam \varx {\inctx {\hctx'} \Omega}$.

  We proceed by induction on the number of steps $n$ in the evaluation; the
  proof is straightforward if $n=0$. Suppose $n > 0$; then
  $\inctx \hctx {\shift \vark \Omega} \redcbv t' \evalcbv \val$ for some
  $t'$. Suppose a copy of $\shift \vark \Omega$ is in an evaluation context $F$
  in $\inctx \hctx {\shift \vark \Omega}$. The context $F$ cannot be pure,
  because $\inctx \hctx {\shift \vark \Omega}$ reduces, so
  $F = \inctx {F'}{\reset E}$, which implies $t' = \inctx {F'}{\reset \Omega}$;
  this contradicts $t' \evalcbv \val$ (the calculus is
  deterministic). Consequently, the copies of $\shift \vark \Omega$ are not in
  an evaluation context in $\inctx \hctx {\shift \vark \Omega}$, and the
  reduction $\inctx \hctx {\shift \vark \Omega} \redcbv t'$ can be written
  $\inctx \hctx {\shift \vark \Omega} \redcbv \inctx {\hctx''}{\shift \vark
    \Omega}$ for some $\hctx''$, and we also have
  $\inctx \hctx \Omega \redcbv \inctx {\hctx''} \Omega$. We can then conclude by
  applying the induction hypothesis on
  $\inctx {\hctx''}{\shift \vark \Omega} \evalcbv \val$.

  Applying the property $(*)$ with $\hctx = \cctx$, we get that
  $\inctx \cctx {\shift \vark \Omega} \evalcbv \valzero$ implies
  $\inctx \cctx \Omega \evalcbv \val$ for some $\val$. Conversely, if
  $\inctx \cctx {\Omega} \evalcbv \valone$, we can prove that
  $\inctx \cctx {\shift \vark \Omega} \evalcbv \val$ for some $\val$ using the
  same reasoning. Therefore, we have $\shift \vark \Omega \ctxequivone \Omega$.
\end{proof}

We work with $\ctxequivtwo$ as the main contextual equivalence for the relaxed
semantics, since it corresponds to the usual definition of contextual
equivalence in languages similar to the $\lambda$-calculus, where we simply
observe termination~\cite{Abramsky-Ong:IaC93} (see
Proposition~\ref{p:ctx-usual-def}). Henceforth, we simply write $\ctxequiv$ for
$\ctxequivtwo$.

We extend $\ctxequiv$ to open terms using closing substitutions: we say $\subs$
closes $\tm$ if it maps the free variables of $\tm$ to closed values. We define
the \emph{open extension} of a relation as follows.
\begin{defi}%
  \label{d:open-extension}
  Let $\rel$ be a relation on closed terms, and $\tmzero$ and $\tmone$ be open
  terms. We write $\tmzero \open\rel \tmone$ if for every substitution $\subs$
  which closes $\tmzero$ and $\tmone$, $\tmzero \subs \rel \tmone \subs$ holds.
\end{defi}

\begin{rem}
  Contextual equivalence can be defined directly on open terms by requiring that
  the context $\cctx$ binds the free variables of the related terms. We prove
  the resulting relation is equal to $\open\ctxequiv$ in
  Section~\ref{ss:soundness-app}.
\end{rem}

To prove completeness of bisimilarities, we use a variant of
$\ctxequiv$ which takes into account only evaluation contexts to
compare terms.
\begin{defi}%
  \label{d:rctxequiv}
  Let $\tmzero$, $\tmone$ be closed terms. We write $\tmzero \rctxequiv
  \tmone$ if for all closed $\rctx$,
  \begin{itemize}
  \item $\inctx \rctx \tmzero \evalcbv \valzero$ for some $\valzero$ iff $\inctx \rctx \tmone
    \evalcbv \valone$ for some $\valone$;
  \item $\inctx \rctx \tmzero \evalcbv \tmzero'$ for some control stuck term
    $\tmzero'$ iff $\inctx \rctx \tmone \evalcbv \tmone'$ for some control stuck
    term $\tmone'$.
  \end{itemize}
\end{defi}

\noindent The definitions imply $\ctxequiv \mathop{\subseteq} \rctxequiv$.
While proving completeness of applicative bisimilarity in
Section~\ref{s:app}, we also prove $\rctxequiv \mathop{=} \ctxequiv$,
which means that testing with evaluation contexts is as discriminating
as testing with any contexts. Such a simplification result is similar
to Milner's context lemma~\cite{Milner:TCS77}.

The relations $\ctxequivone$ and $\ctxequivtwo$ are not suitable for
the original semantics, because they distinguish terms that should be
equated according to Kameyama and Hasegawa's axiomatization. Indeed,
according to these relations, $\shift \vark {\app \vark \val}$ (where
$\vark \notin \fv \val$) cannot be related to $\val$ (axiom~\AXshiftelim in Figure~\ref{f:axioms}), because a stuck term cannot be
related to a value. In the next section, we discuss a definition of
contextual equivalence for the original semantics.

\subsection{Definition for the Original Semantics}

Terms are evaluated in the original semantics within an enclosing \resetId{}, so
the corresponding contextual equivalence should test terms in contexts of the
form $\reset \cctx$ only. Because delimited terms cannot reduce to stuck terms
(Proposition~\ref{p:eval-reset}), the only possible observable action is
evaluation to values. We, therefore, define contextual equivalence for the
original semantics as follows.
\begin{defi}%
  \label{d:context-p}
  Let $\tmzero$, $\tmone$ be closed terms. We write $\tmzero \ctxequivp \tmone$
  if for all closed $\cctx$, $\reset{\inctx \cctx \tmzero} \evalcbv \valzero$
  for some $\valzero$ iff $\reset{\inctx \cctx \tmone} \evalcbv \valone$ for
  some $\valone$.
\end{defi}
The relation $\ctxequivp$ is defined on all (closed) terms, not just
delimited ones. The resulting relation is less discriminating than $\ctxequiv$,
because $\ctxequivp$ uses contexts of a particular form, while~$\ctxequiv$ tests
with all contexts.
\begin{prop}%
  \label{p:gen-inc-prog}
  We have $\mathord{\ctxequiv} \subsetneq \mathord{\ctxequivp}$.
\end{prop}
As a result, any equivalence between terms we prove for the relaxed semantics
also holds in the original semantics, and any bisimilarity sound
{\wrt}$\ctxequiv$ (like the bisimilarities we define in Sections~\ref{s:app},~\ref{ss:env-LTS}, and~\ref{s:nf}) is also sound {\wrt}$\ctxequivp$. However, to
reach completeness, we have to design a bisimilarity suitable for delimited
terms (see Section~\ref{ss:env-LTS-original}). As for the relaxed semantics, we
extend $\ctxequivp$ to open terms using Definition~\ref{d:open-extension}.

The inclusion of Proposition~\ref{p:gen-inc-prog} is strict because, e.g.,
$\ctxequivp$ verifies the axiom \AXshiftelim, while~$\ctxequiv$ does not. In
fact, we prove in Section~\ref{ss:axioms-env} that $\ctxequivp$
contains the CPS equivalence $\cpsequiv$. The reverse inclusion holds neither
for $\ctxequivp$ nor~$\ctxequiv$: there exist contextually equivalent terms that
are not CPS equivalent.

\begin{prop}
\begin{enumerate}
\item\label{p:ce-cps-one}
We have $\Omega \ctxequivp \app \Omega \Omega$ (respectively $\Omega \ctxequiv
\app \Omega \Omega$), but $\Omega \not\cpsequiv \app \Omega \Omega$.
\item\label{p:ce-cps-two}
We have $\Turing \ctxequivp \Curry$
(respectively $\Turing \ctxequiv \Curry$), but $\Turing \not\cpsequiv
\Curry$.
\end{enumerate}
\end{prop}
\noindent The contextual equivalences $\ctxequiv$ and $\ctxequivp$ put all diverging terms
in one equivalence class, while CPS equivalence is more
discriminating. Furthermore, as is usual with equational theories for
$\lambda$-calculi, CPS equivalence is not strong enough to equate Turing's and
Curry's (call-by-value) fixed-point combinators.

\medskip

As explained in the introduction, contextual equivalence is difficult to prove
in practice for two given terms because of the quantification over contexts. We
look for a suitable replacement (that is, an equivalence that is at least sound
{\wrt}$\ctxequiv$ or $\ctxequivp$) by studying different styles of bisimulation
in the next sections.

% Local variables:
% mode: latex
% TeX-master: "journal.tex"
% End:

\section{Applicative Bisimilarity}%
\label{s:app}
Applicative bisimilarity has been originally defined for the lazy
$\lambda$-calculus~\cite{Abramsky-Ong:IaC93}. The main idea is to reduce
(closed) terms to values, and then compare the resulting $\lambda$-abstractions
by applying them to an arbitrary argument. When sound, applicative bisimilarity
for deterministic languages is usually also complete (see,
\eg~\cite{Gordon-Rees:POPL96,Tiuryn-Wand:CAAP96,Gordon:TCS99}), and soundness is
proved thanks to a systematic technique called Howe's method~\cite{Howe:IaC96,Gordon:TCS99}. However, defining and proving sound the most
powerful up-to techniques, such as bisimulation up to context, remain an open
issue for applicative bisimilarity.

Very few works study applicative bisimilarity in a calculus with
control. Merro~\cite{Merro:AI10} defines an applicative bisimilarity which
characterizes contextual equivalence in the \emph{CPS
  calculus}~\cite{Thielecke:PHD}, a minimal calculus which models the control
features of functional languages with imperative jumps. In the
$\lambda\mu$-calculus, Lassen~\cite{Lassen:99} proposes a sound but not complete
applicative bisimilarity in call-by-name. We improve this
result~\cite{Biernacki-Lenglet:MFPS14} by defining sound and complete
applicative bisimilarities in both call-by-name and call-by-value.

In this section, we define a sound and complete applicative
bisimilarity for the relaxed semantics of $\lamshift$. Our definition
of applicative bisimilarity relies on a labeled transition system,
introduced first (Section~\ref{ss:LTS-app}). We then prove its
soundness and completeness in Section~\ref{ss:soundness-app}, before
showing how it can be used on the $\lamshift$ axiomatization
(Section~\ref{ss:axioms-app}). We cover results that have been
originally presented in~\cite{Biernacki-Lenglet:FOSSACS12}.

\subsection{Applicative Bisimilarity}%
\label{ss:LTS-app}

One possible way to define an applicative bisimilarity is to rely on a
labeled transition system (LTS), where the possible interactions of a
term with its environment are encoded in the labels (see, \eg~\cite{Gordon-Rees:POPL96,Gordon:TCS99}). Using a LTS simplifies the
definition of the bisimilarity and makes it easier to use some
techniques in proofs, such as diagram chasing. In Figure~\ref{f:lts},
we define a LTS $\tmzero \lts\act \tmone$ with three kinds of
transitions, where we assume all the terms to be closed. An
\emph{internal action} $\tm \lts\tau \tm'$ is an evolution from~$\tm$
to~$\tm'$ without any help from the surrounding context; it
corresponds to a reduction step from~$\tm$ to~$\tm'$. The transition
$\valzero \lts \valone \tm$ expresses the fact that $\valzero$ needs
to be applied to another value~$\valone$ to evolve, reducing to
$\tm$. Finally, the transition $\tm \lts \ctx \tm'$ means that $\tm$
is control stuck, and when $\tm$ is put in a context $\ctx$ enclosed
in a \resetId{}, the capture can be triggered, the result of which being
$\tm'$.

\begin{figure}
\begin{mathpar}
  \inferrule*[right=\LTSbeta]
            { }
            {\app {\lamp \varx \tm} \val \ltstau \subst \tm \varx \val}
  \and
  \inferrule*[right=\LTScompl]{\tmzero \ltstau \tmzero'}
             {\app \tmzero \tmone \ltstau \app {\tmzero'} \tmone}
  \and
  \inferrule*[right=\LTScompr]{\tm \ltstau \tm'}
  {\app \val \tm \ltstau \app \val {\tm'}}
  \and
  \inferrule*[right=\LTSreset]{ }
  {\reset \val \ltstau \val}
  \and
  \inferrule*[right=\LTScompreset]{\tm \ltstau \tm'}
  {\reset \tm \ltstau \reset {\tm'}}
  \and
  \inferrule*[right=\LTScaptreset]{\tm \lts\mtctx \tm'}
  {\reset \tm \ltstau \tm'}
  \and
  \inferrule*[right=\LTSval]
             { }
             {\lam \varx \tm \lts\val \subst \tm \varx \val}
  \and
  \inferrule*[right=\LTSshift]
             {\varx \notin \fv \ctx}
             {\shift \vark \tm \lts\ctx \reset {\subst \tm \vark {\lam \varx {\reset
                     {\inctx \ctx \varx}}}}}
             \and
  \inferrule*[right=\LTScaptl]
             {\tmzero \lts{\inctx \ctx {\app \mtctx \tmone}} \tmzero'}
             {\app \tmzero \tmone \lts\ctx \tmzero'}
  \and
  \inferrule*[right=\LTScaptr]
             {\tm \lts{\inctx \ctx{\app \val \mtctx}} \tm'}
             {\app \val \tm \lts\ctx \tm'}
\end{mathpar}
\caption{Labeled Transition System}%
\label{f:lts}
\end{figure}

Most rules for internal actions (Figure~\ref{f:lts}) are straightforward; the
rules $\LTSbeta$ and $\LTSreset$ mimic the corresponding reduction rules, and
the compositional rules $\LTScompr$, $\LTScompl$, and~$\LTScompreset$ allow
internal actions to happen within any evaluation context. The rule
$\LTScaptreset$ for context capture is explained later. Rule $\LTSval$ defines
the only possible transition for values. While both rules $\LTSbeta$ and
$\LTSval$ encode $\beta$-reduction, they are quite different in nature; in the
former, the term $\app{\lamp \varx \tm} \val$ can evolve by itself, without any
help from the surrounding context, while the latter expresses the possibility
for $\lam \varx \tm$ to evolve only if a value $\val$ is provided by the
environment.

The rules for context capture are built following the principles of
complementary semantics developed in~\cite{Lenglet-al:CONCUR09}. The
label of the transition $\tm \lts\ctx \tm'$ contains what the
environment needs to provide (a context $\ctx$, but also an enclosing
\resetId{}, left implicit) for the control stuck term $\tm$ to reduce to
$\tm'$. Hence, the transition $\tm \lts\ctx \tm'$ means that we have
$\reset {\inctx \ctx \tm} \lts\tau \tm'$ by context capture. For
example, in the rule $\LTSshift$, the result of the capture of $\ctx$
by $\shift \vark \tm$ is $\reset{\subst \tm \vark {\lam \varx {\reset
      {\inctx \ctx \varx}}}}$.

In rule $\LTScaptl$, we want to know the result of the capture of
$\ctx$ by the term $\app \tmzero \tmone$, assuming~$\tmzero$ contains
a \shiftId{} ready to perform the capture. Under this hypothesis, the
capture of $\ctx$ by~$\app \tmzero \tmone$ comes from the capture of
$\inctx {\ctx}{\apctx \mtctx \tmone}$ by $\tmzero$. Therefore, as a
premise of the rule $\LTScaptl$, we check that $\tmzero$ is able to
capture $\inctx \ctx {\app \mtctx \tmone}$, and the result $\tmzero'$
of this transition is exactly the result we want for the capture of
$\ctx$ by $\app \tmzero \tmone$. The rule $\LTScaptr$ follows the same
pattern. Finally, a control stuck term $\tm$ enclosed in a \resetId{} is
able to perform an internal action (rule $\LTScaptreset$); we obtain
the result $\tm'$ of the transition $\reset \tm \lts\tau \tm'$ by
letting $\tm$ capture the empty context, \ie by considering the
transition $\tm \lts\mtctx \tm'$.

\begin{exa}
  We illustrate how the LTS handles capture by considering the transition from
  $\reset{\app {\appp i {\shift \vark \omega}} \Omega}$.
  \begin{mathpar}
    \inferrule*[Right=\LTScaptreset]{
      \inferrule*[Right=\LTScaptl]{
        \inferrule*[Right=\LTScaptr]{
          \inferrule*[Right=\LTSshift]{ }
          {\shift \vark \omega \lts{\apctx {(\vctx i \mtctx )} \Omega} \reset \omega}}
        {
          \app i {\shift \vark \omega} \lts{\apctx \mtctx \Omega} \reset \omega}}
      {
        \app{\appp i {\shift \vark \omega}} \Omega \lts \mtctx
        \reset \omega}}
    {\reset{\app {\appp i {\shift \vark \omega}} \Omega} \lts\tau
      \reset \omega}
  \end{mathpar}
  Reading the tree from bottom to top, we see that the rules
  $\LTScaptreset$, $\LTScaptl$, and $\LTScaptr$ build the captured
  context in the label by deconstructing the initial term. Indeed, the
  rule $\LTScaptreset$ removes the outermost \resetId{} and initiates the
  context in the label with $\mtctx$. The rules $\LTScaptl$ and
  $\LTScaptr$ then successively remove the outermost application and
  store it in the context. The process continues until a \shiftId{}
  operator is found; then we know the captured context is completed,
  and the rule $\LTSshift$ computes the result of the capture. This
  result is then simply propagated from top to bottom by the other
  rules.
\end{exa}

The LTS corresponds to the reduction semantics $\redcbv$ and exhibits the
observable terms (values and control stuck terms) of the language. The only
difficulty is in the treatment of control stuck terms. The next lemma makes the
correspondence between $\lts\ctx$ and control stuck terms explicit.

\begin{lem}%
  \label{l:decompose-shift}
  If $\tm \lts\ctx \tm'$, then there exist $\ctx'$, $\vark$, and $\tms$
  such that $\tm = \inctx {\ctx'}{\shift \vark \tms}$ and $\tm' =
  \reset{\subst \tms \vark {\lam \varx {\reset {\inctx {\ctx}{\inctx
            {\ctx'} \varx}}}}}$.
\end{lem}

\noindent The proof is by induction on $\tm \lts\ctx \tm'$. From this
lemma, we can deduce the correspondence between $\lts\tau$ and $\redcbv$, and
between $\lts\act$ (for $\act \neq \tau$) and the observable actions of the
language.

\begin{prop}%
  \label{p:lts-redcbv-main}
  The following hold:
  \begin{itemize}
  \item We have $\ltstau \mathop{=} \redcbv$.
  \item If $\tm \lts\ctx \tm'$, then $\tm$ is a stuck term, and $\reset{\inctx
      \ctx \tm} \ltstau \tm'$.
  \item If $\tm \lts\val \tm'$, then $\tm$ is a value, and $\app \tm \val
    \ltstau \tm'$.
  \end{itemize}
\end{prop}

\noindent
We write $\cloltstau$ for the reflexive and transitive closure of $\ltstau$. We
define the weak delay transition\footnote{A transition where internal steps are
  allowed before, but not after a visible action.} $\ltswk\act$ as $\cloltstau$
if $\act=\tau$ and as $\cloltstau \lts\act$ otherwise. The definition of (weak
delay) bisimilarity is then straightforward.
\begin{defi}%
  \label{d:app-bisim}
  A relation $\rel$ on closed terms is an applicative simulation if
  $\tmzero \rel \tmone$ implies that for all $\tmzero \lts\act
  \tmzero'$, there exists $\tmone'$ such that $\tmone \ltswk\act
  \tmone'$ and $\tmzero' \rel \tmone'$.  A relation $\rel$ on closed
  terms is an applicative bisimulation if $\rel$ and $\inv\rel$ are
  applicative simulations. Applicative bisimilarity $\appbisim$ is the
  largest applicative bisimulation.
\end{defi}
\noindent In words, two terms are equivalent if any transition from one is
matched by a weak transition with the same label from the other. Because the
calculus is deterministic, it is not mandatory to test the internal steps when
proving that two terms are bisimilar.
\begin{prop}%
  \label{p:bisim-eval}
  If $\tm \ltstau \tm'$ (respectively $\tm \evalcbv \tm'$) then $\tm \appbisim
  \tm'$.
\end{prop}
\noindent As a result, applicative bisimulation can be defined in terms of
big-step transitions.

\begin{defi}
  A relation $\rel$ on closed terms is a big-step applicative simulation if
  $\tmzero \rel \tmone$ implies that for all $\tmzero \ltswk\act \tmzero'$ with
  $\act \neq \tau$, there exists $\tmone'$ such that $\tmone \ltswk\act \tmone'$
  and $\tmzero' \rel \tmone'$. A relation $\rel$ on closed terms is a big-step
  applicative bisimulation if $\rel$ and $\inv\rel$ are big-step applicative
  simulations.
\end{defi}

\begin{prop}
  If $\rel$ is a big-step applicative bisimulation, then $\rel \mathop{\subseteq}
  \appbisim$.
\end{prop}

\begin{proof}[Sketch]
  By showing that
  $\{(\tmzero, \tmone) \bnfor (\tmzero,\tmone) \in
  \exists \tmzero', \tmone', \tmzero' \clocbv \tmzero \wedge
  \tmone' \clocbv \tmone \wedge \tmzero' \rel \tmone' \}$ is an applicative
  bisimulation.
\end{proof}

\noindent In this section, we drop the adjective ``applicative'' and refer to
the two kinds of relations simply as ``bisimulation'' and ``big-step
bisimulation'' where it does not cause confusion.

\begin{exa}[double \textreset]%
  \label{ex:reset-reset-app}
  For all closed terms $\tm$, we show that
  $\reset{\reset \tm} \appbisim \reset\tm$ holds by proving that
  $\rel {\mathord{\is}} \{ (\reset\tm, \reset{\reset \tm}) \mid \tm \in \cterms
  \} \cup \{ (\tm, \tm) \mid \tm \in \cterms \}$ is an applicative
  bisimulation. First, $\reset \tm$ cannot be a value or a control-stuck term,
  so we only have to consider $\lts\tau$-transition. By case analysis on the
  reduction rules, we can see that $\reset \tm \lts\tau \tm'$ iff
  $\tm' = \reset {\tm''}$ for some $\tm''$, or $\tm'$ is a value $\val$.

  If $\reset \tm \lts\tau \reset{\tm''}$, then $\reset {\reset \tm} \lts\tau \reset
  {\reset {\tm''}}$, and the resulting terms are in $\rel$. Otherwise, if
  $\reset \tm \lts\tau \val$, then $\reset {\reset \tm} \lts\tau \lts\tau \val$
  and we get identical terms. Conversely, if $\reset{\reset \tm} \lts\tau \tm'$,
  then we can show that either $\tm' = \reset{\reset {\tm''}}$ for some $\tm''$ and
  $\reset \tm \lts \tau \reset{\tm''}$, or $\tm' = \reset\val$ for some $\val$
  and $\reset \tm = \reset \val$. This concludes the proof for the terms
  in the first set of $\rel$, and checking the bisimulation game for identical
  terms (the second set of~$\rel$) is straightforward.
\end{exa}

\begin{exa}[Turing's combinator]%
  \label{ex:fixed-point-app}
  We study here the relationships between Turing's and Curry's fixed-point
  combinator and their respective variants with delimited
  control~\cite{Danvy-Filinski:DIKU89} (see Example~\ref{e:combinators} for the
  definitions). We start with Turing's combinator $\Turing$ and its variant
  $\Turingshift \is \reset {\app \theta {\shift \vark {\app \vark
        \vark}}}$. The two terms can perform the following transitions:
  \begin{align*}
    \Turing & \ltswk\val \app \val {\lamp z {\app {\app {\app \theta
              \theta} \val} z}} \\
    \Turingshift & \ltswk\val \app \val  {\lamp z {\app {\app {\app {\lamp
                \varx {\reset {\app \theta \varx}}}{\lamp \varx {\reset {\app
                    \theta \varx}}}} \val} z}}
  \end{align*}
  Taking $\val = \lam \varx \tm$, we have to study
  $\subst \tm \varx {\lamp z {\app {\app {\app \theta \theta} \val} z}}$, and
  $\subst \tm \varx {\lamp z {\app {\app {\app {\lamp \varx {\reset {\app \theta
                \varx}}}{\lamp \varx {\reset {\app \theta \varx}}}} \val}
      z}}$. A way to proceed is by case analysis on $\tm$, the interesting case
  being $\tm = \inctx \rctx {\app \varx {\val'}}$. If it is possible to conclude
  using applicative bisimulation, the needed candidate relation is much more
  complex than with environmental (Example~\ref{ex:fixed-point-env-utc}) or
  normal-form (Example~\ref{ex:fixed-point}) bisimulations, so we refer to these
  examples for a complete proof.

  In contrast, Curry's combinator $\Curry$ is not bisimilar to its
  delimited-control variant
  $\Curryshift \is \lam \varx {\reset {\app {\delta_\varx}{\shift \vark
        {\app \vark \vark}}}}$. Indeed, after applying these values to an
  argument $\val$, we obtain respectively
  $\app \val {\lamp z {\app {\app {\delta_\val}{\delta_\val}} z}}$ and
  $\reset {\reset {\app \val {\lamp z {\app {\app {\lamp y {\reset{\app
                  {\delta_\val} y}}}{\lamp y {\reset{\app {\delta_\val} y}}}}
          z}}}}$, and these terms are not bisimilar if
  $\val = \lam \varx {\shift \vark \Omega}$, as the first one reduces to a
  control-stuck term while the second one diverges.
\end{exa}

\begin{rem}
  Applicative simulation can be formulated in a more classic, but equivalent,
  way (without labeled transitions), as follows. A relation $\rel$ on closed
  terms is an applicative simulation if $\tmzero \rel \tmone$ implies:
  \begin{itemize}
  \item if $\tmzero \redcbv \tmzero'$, then there exists $\tmone'$ such that
    $\tmone \clocbv \tmone'$ and $\tmzero' \rel \tmone'$;
  \item if $\tmzero$ is a value $\lam \varx {\tmzero'}$, then there exists
    $\tmone'$ such that $\tmone \clocbv \lam \varx {\tmone'}$, and for all
    closed $\val$, we have $\subst {\tmzero'} \varx \val \rel \subst {\tmone'}
    \varx \val$;
  \item if $\tmzero$ is a stuck term $\inctx \ctxzero {\shift \vark
      {\tmzero'}}$, then there exist $\tmone'$ and $\ctxone$ such that $\tmone
    \clocbv \inctx \ctxone {\shift \vark {\tmone'}}$, and for all closed $\ctx$,
    we have $\reset{\subst {\tmzero'} \vark {\lam \varx {\reset{\inctx \ctx
            {\inctx {\ctxzero} \varx}}}}} \rel \reset{\subst {\tmone'} \vark
      {\lam \varx {\reset{\inctx \ctx {\inctx {\ctxone} \varx}}}}}$.
  \end{itemize}
  The correspondence between this formulation and Definition~\ref{d:app-bisim}
  is a direct consequence of Proposition~\ref{p:lts-redcbv-main}.
\end{rem}

\subsection{Soundness and Completeness}%
\label{ss:soundness-app}

To prove the soundness of $\appbisim$ {\wrt}the contextual equivalence
$\ctxequiv$, we show that $\appbisim$ is a congruence using \emph{Howe's
  method}, a well-known congruence proof method initially developed for the
$\lambda$-calculus~\cite{Howe:IaC96,Gordon:TCS99}. The idea of the method is as
follows: first, define the \emph{Howe's closure} of $\appbisim$, written
$\clohbisim$, a relation which contains $\appbisim$ and is compatible by
construction. Then, prove a simulation-like property for $\clohbisim$; from this
result, prove that $\clohbisim$ and $\appbisim$ coincide on closed
terms. Because $\clohbisim$ is compatible, it shows that $\appbisim$ is
compatible as well, and therefore a congruence.

The definition of $\clohbisim$ relies on the notion of
\emph{compatible refinement}; given a relation $\rel$ on open terms,
the compatible refinement $\comp\rel$ relates two terms iff they have
the same outermost operator and their immediate subterms are related
by $\rel$. Formally, it is inductively defined by the following rules:
\begin{mathpar}
  \inferrule{ }
  {\varx \comp\rel \varx}
  \and
  \hspace{-0.1em}\inferrule{\tmzero \rel \tmone}
  {\lam \varx \tmzero \comp\rel \lam \varx \tmone}
  \and
  \inferrule{\tmzero \rel \tmone \quad\quad \tmzero' \rel \tmone'}
  {\app \tmzero {\tmzero'} \comp\rel \app \tmone {\tmone'}}
  \and
  \inferrule{\tmzero \rel \tmone}
  {\shift \vark \tmzero \comp\rel \shift \vark \tmone}
  \and
  \inferrule{\tmzero \rel \tmone}
  {\reset \tmzero \comp\rel \reset\tmone}
\end{mathpar}
Howe's closure $\clohbisim$ is inductively defined as the smallest compatible
relation containing $\open\appbisim$ and closed under right composition with
$\open\appbisim$.

\begin{defi}
  Howe's closure $\clohbisim$ is the smallest relation satisfying:
  \begin{mathpar}
    \inferrule{\tmzero \open\appbisim \tmone}
    {\tmzero \clohbisim \tmone}
    \and
    \inferrule{\tmzero \clohbisim\open\appbisim \tmone}
    {\tmzero \clohbisim \tmone}
    \and
    \inferrule{\tmzero \comp\clohbisim \tmone}
    {\tmzero \clohbisim \tmone}
  \end{mathpar}
\end{defi}
\noindent By construction, $\clohbisim$ is compatible (by the third rule of the
definition), and composing on the right with $\open\appbisim$ gives some
transitivity properties to $\clohbisim$. In particular, we can prove
that~$\clohbisim$ is \emph{substitutive}: if $\tmzero \clohbisim \tmone$ and
$\valzero \clohbisim \valone$, then
$\subst \tmzero \varx {\valzero} \clohbisim \subst \tmone \varx {\valone}$.

Let $\clohbisimc$ be the restriction of $\clohbisim$ to closed
terms. We cannot prove directly that $\clohbisimc$ is a bisimulation,
so we prove a stronger result. Suppose we have $\tmzero \clohbisimc
\tmone$; instead of simply requiring $\tmzero \lts\act \tmzero'$ to be
matched by $\tmone$ with the same label $\act$, we ask $\tmone$ to be
able to respond for any label $\act'$ related to $\act$ by
$\clohbisimc$. We, therefore, extend $\clohbisim$ to all labels, by
adding the relation $\tau \clohbisim \tau$, and by defining $\ctx
\clohbisim \ctx'$ as follows:
\begin{mathpar}
  \inferrule{ }{\mtctx \clohbisim \mtctx} \and
  \inferrule{\ctxzero \clohbisim \ctxone \\ \tmzero \clohbisim \tmone}
  {\apctx \ctxzero \tmzero \clohbisim \apctx \ctxone \tmone}
  \and
  \inferrule{\ctxzero \clohbisim \ctxone \\ \valzero \clohbisim \valone}
  {\vctx \valzero \ctxzero  \clohbisim \vctx \valone \ctxone}
\end{mathpar}

\begin{lem}[Simulation-like property]%
  \label{l:sim-property-main}
  If $\tmzero \clohbisimc \tmone$ and $\tmzero \lts\act \tmzero'$, then for all
  $\act \clohbisimc \act'$, there exists $\tmone'$ such that $\tmone
  \ltswk{\act'} \tmone'$ and $\tmzero' \clohbisimc \tmone'$.
\end{lem}
The main difficulty when applying Howe's method is to prove this simulation-like
property. The proof~\cite{Biernacki-Lenglet:FOSSACS12} is by induction on
$\tmzero \clohbisimc \tmone$, and then by case analysis on the transition
$\tmzero \lts\act \tmzero'$.  Lemma~\ref{l:sim-property-main} allows us to prove
that $\clohbisimc$ is a simulation, by choosing $\act'=\act$. We cannot directly
deduce that $\clohbisimc$ is a bisimulation, however we can prove that its
transitive and reflexive closure $\rtclo{(\clohbisimc)}$ is a bisimulation,
because of the following classical property of the Howe's
closure~\cite{Gordon:TCS99}.
\begin{lem}%
  \label{l:rth-is-sym}
  The relation $\rtclo{(\clohbisim)}$ is symmetric.
\end{lem}

\begin{proof}[Sketch]
  The proof is by induction on the definition of the reflexive and transitive
  closure. The inductive case is straightforward. For the base case, we show
  that $\tmzero \clohbisim \tmone$ implies $\tmone \rtclo{(\clohbisim)} \tmzero$
  by induction on the definition of Howe's closure. Most cases are
  straightforward using the induction hypothesis. The interesting case is when
  $\tmzero \clohbisim \tm \open\appbisim \tmone$ for some $\tm$. By the
  induction hypothesis, we have $\tm \rtclo{(\clohbisim)} \tmzero$. Because
  $\appbisim$ itself is symmetric, we also have $\tmone \open\appbisim \tm$,
  which implies $\tmone \clohbisim \tm$, which when combined with
  $\tm \rtclo{(\clohbisim)} \tmzero$ gives the required result.
\end{proof}

The fact that $\rtclo{(\clohbisimc)}$ is a bisimulation implies that
$\rtclo{(\clohbisimc)} \mathop{\subseteq} \appbisim$. Because $\appbisim
\mathop{\subseteq} \clohbisimc \mathop{\subseteq}
\mathord{\rtclo{(\clohbisimc)}}$ holds by construction, we can deduce $\appbisim
\mathop{=} \clohbisimc$. Since $\clohbisimc$ is compatible, and we can easily
show that $\appbisim$ is transitive and reflexive, we have the following
result.
\begin{thm}%
  \label{t:app-congruence}
  The relation $\appbisim$ is a congruence.
\end{thm}
\noindent Combined with the fact that labels correspond to observable actions
(Proposition~\ref{p:lts-redcbv-main}), Theorem~\ref{t:app-congruence} entails
that $\appbisim$ is sound {\wrt}contextual equivalence.
\begin{cor}%
  \label{t:soundness-main}
  We have $\appbisim \mathop{\subseteq} \ctxequiv$.
\end{cor}

\subsubsection*{Completeness and context lemma}

For the reverse inclusion, we use $\rctxequiv$, the contextual equivalence which
tests with contexts $\rctx$ only (see Definition~\ref{d:rctxequiv}). We can prove
that $\appbisim$ is complete {\wrt}$\rctxequiv$, by showing that $\rctxequiv$
is an applicative bisimulation~\cite{Biernacki-Lenglet:FOSSACS12}.
\begin{thm}%
  \label{t:completeness-main}
  We have $\rctxequiv \mathop\subseteq \appbisim$.
\end{thm}

\begin{proof}[Sketch]
  We show that $\rctxequiv$ is an applicative bisimulation. Let
  $\tmzero \rctxequiv \tmone$. If $\tmzero \lts\tau \tmzero'$, it is easy to
  check that we still have $\tmzero' \rctxequiv \tmone$. If
  $\tmzero \lts\val \tmzero'$, then by Proposition~\ref{p:lts-redcbv-main},
  $\tmzero$ is a value and $\tmzero \iapp \val \redcbv \tmzero'$. Because
  $\tmzero \rctxequiv \tmone$, there exists $\valone$ such that
  $\tmone \clocbv \valone$, therefore $\tmone \ltswk \val \tmone'$ for $\tmone'$
  such that $\valone \iapp \val \redcbv \tmone'$. What is left to prove is that
  $\tmzero' \rctxequiv \tmone'$, i.e., for all $\rctx$,
  $\inctx \rctx {\tmzero'}$ behaves like $\inctx \rctx{\tmone'}$ (i.e., one
  evaluates to respectively a value or stuck term iff the other do so as well).
  But from $\tmzero \rctxequiv \tmone$, we get that $\inctx {\rctx'} \tmzero$
  behaves like $\inctx {\rctx'} \tmone$ for all $\rctx'$, so in particular for
  $\rctx' = \inctx \rctx {\app \mtctx \val}$. In the end,
  $\inctx \rctx {\app \tmzero \val}$ behaves like
  $\inctx \rctx {\app \tmone \val}$, but these terms reduces to respectively
  $\inctx \rctx {\tmzero'}$ and $\inctx \rctx {\tmone'}$, so we can conclude
  from there. The reasoning is the same for $\tmzero \lts\ctx \tmzero'$.
\end{proof}

\noindent As a result, the relations $\ctxequiv$, $\rctxequiv$, and $\appbisim$
coincide, which means that $\appbisim$ is complete {\wrt}$\ctxequiv$

\begin{cor}%
  \label{c:context-lemma}
  We have $\ctxequiv \mathop= \rctxequiv \mathop= \appbisim$.
\end{cor}
\noindent Indeed, we have $\rctxequiv \mathop\subseteq \appbisim$ (Theorem~\ref{t:completeness-main}), $\appbisim \mathop\subseteq \ctxequiv$ (Corollary~\ref{t:soundness-main}), and $\ctxequiv \mathop\subseteq \rctxequiv$ (by
definition).

This equality also allows us to prove that we can formulate the open extension
of $\ctxequiv$ using capturing contexts.

\begin{prop}
  We have $\tmzero \open\ctxequiv \tmone$ iff for all $\cctx$ capturing the
  variables of $\tmzero$ and $\tmone$, the following holds:
  \begin{itemize}
  \item $\inctx \cctx \tmzero \evalcbv \valzero$ iff $\inctx
    \cctx \tmone \evalcbv \valone$;
  \item $\inctx \cctx \tmzero \evalcbv \tmzero'$, where $\tmzero'$ is control
    stuck, iff $\inctx \cctx \tmone \evalcbv \tmone'$, with $\tmone'$
    control stuck as well.
  \end{itemize}
\end{prop}

\begin{proof}
  Suppose $\tmzero \open\ctxequiv \tmone$. Then $\tmzero
  \open\appbisim \tmone$, and because $\open\appbisim$ is a
  congruence, for all $\cctx$ capturing the variables of $\tmzero$ and
  $\tmone$, we have $\inctx \cctx \tmzero \appbisim \inctx \cctx
  \tmone$. We have $\inctx \cctx \tmzero \evalcbv \valzero$ iff
  $\inctx \cctx \tmone \evalcbv \valone$ by bisimilarity definition,
  and similarly with $\inctx \cctx \tmzero \evalcbv \tmzero'$, where
  $\tmzero'$ is control stuck.

  For the reverse implication, suppose that for all $\cctx$ capturing the
  variables of $\tmzero$ and~$\tmone$, the two items of the proposition
  hold. Let $\sigma = \{ \valone/\varx_1 \ldots \val_n/\varx_n \}$ be a
  substitution closing $\tmzero$ and~$\tmone$. Let $\cctx$ be a closed
  context. We want to prove that $\inctx \cctx {\tmzero \sigma} \clocbv \val$
  for some $\val$ iff $\inctx \cctx {\tmone \sigma} \clocbv \val'$ for some
  $\val'$, and similarly for control stuck terms. But the context
  $\cctx' \is \inctx \cctx {\app{\app {\lamp{\varx_1 \ldots \varx_n}
        \mtctx}{\val_1}}{\ldots \val_n}}$ is a context capturing the variables
  of $\tmzero$ and $\tmone$, and we have
  $\inctx {\cctx'} \tmzero \clocbv \inctx \cctx {\tmzero \sigma}$ and
  $\inctx {\cctx'} \tmone \clocbv \inctx \cctx {\tmone \sigma}$. Consequently,
  $\inctx \cctx {\tmzero \sigma} \clocbv \val$ iff
  $\inctx {\cctx'} \tmzero \clocbv \val$ iff
  $\inctx {\cctx'} \tmzero \clocbv \val'$ (first item of the proposition) iff
  $\inctx \cctx {\tmone \sigma} \clocbv \val'$ for some $\val$ and $\val'$. The
  reasoning is the same for control stuck terms.
\end{proof}

The next example is used as a counter-example to show that normal-form
bisimilarity is not complete (Proposition~\ref{p:cex-dupl}): the two terms below
are not normal-form bisimilar, but they can be proved applicative bisimilar
quite easily.

\begin{prop}%
  \label{p:cex-nf-completeness}
  We have
  $\reset {\app \varx i} \open\ctxequiv \lamp y {\reset{\app \varx i}} \iapp \reset {\app
    \varx i}$.
\end{prop}

\begin{proof}
  We prove that
  $\rel {\mathord{\is}} \{(\reset \tm, \lamp y {\reset \tm} \iapp \reset \tm)
  \mid \tm \in \cterms, y \notin \fv \tm \} \mathrel\cup \{(\tm, \tm) \mid \tm
  \in \cterms \}$ is a big-step bisimulation. The term $\reset \tm$ can either
  diverge or reduce to a value (according to Proposition~\ref{p:eval-reset}). If
  it diverges, then both $\reset \tm$ and
  $\lamp y {\reset \tm} \iapp \reset \tm$ diverge, otherwise, they both evaluate
  to the same value $\val$. For all $\val'$, we, therefore, have
  $\reset \tm \ltswk{\val'} \tm'$ iff
  $\lamp y {\reset \tm} \iapp \reset \tm \ltswk{\val'} \tm'$, and
  $\tm' \rel \tm'$ holds, as wished.
\end{proof}

\subsection{Proving the Axioms}%
\label{ss:axioms-app}

We show how to prove Kameyama and Hasegawa's axioms
(Section~\ref{ss:cps-equivalence}) except for $\AXshiftelim$ using applicative
bisimulation. In the following propositions, we assume the terms to be closed,
since the proofs for open terms can be deduced directly from the results for
closed terms. First, note that the \AXbeta, \AXresetshift, and \AXresetval
axioms are direct consequences of Proposition~\ref{p:bisim-eval}.

\begin{prop}[\AXetav axiom]
  If $\varx \notin \fv \val$, then $\lam \varx {\app \val \varx}
  \appbisim \val$.
\end{prop}

\begin{proof}
  We prove that $\rel \is \{(\lam \varx {\app {\lamp y \tm} \varx},
    \lam y \tm) \mmid \tm \in \terms, \fv \tm \subseteq \{ y\} \} \mathrel{\cup}
  \mathord{\appbisim}$ is a bisimulation. To this end, we have to check that
  $\lam \varx {\app {\lamp y \tm} \varx} \lts\valzero \app {\lamp y
    \tm} \valzero$ is matched by $\lam y \tm \lts\valzero \subst \tm y
  \valzero$, \ie that $\app {\lamp y \tm} \valzero \rel \subst \tm y
  \valzero$ holds for all $\valzero$. We have $\app {\lamp y \tm}
  \valzero \lts\tau \subst \tm y \valzero$, and because $\lts\tau
  \mathord{\subseteq} \appbisim \mathord{\subseteq} \rel$, we have the
  required result.
\end{proof}

\begin{prop}[\AXshiftreset axiom]%
  \label{p:reset-reset}
  We have $\shift \vark {\reset \tm} \appbisim \shift \vark \tm$.
\end{prop}

\begin{proof}
  We have $\shift \vark {\reset \tm} \lts\ctx \reset {\reset {\subst \tm \vark
      {\lam \varx {\reset{\inctx \ctx \varx}}}}}$ and $\shift \vark \tm \lts\ctx
  \reset {\subst \tm \vark {\lam \varx {\reset{\inctx \ctx \varx}}}}$ for all
  $\ctx$. We obtain terms of the form $\reset{\reset {\tm'}}$ and
  $\reset{\tm'}$, and we have proved in Example~\ref{ex:reset-reset-app} that
  $\reset{\reset {\tm'}} \appbisim \reset{\tm'}$ holds for all $\tm'$.
\end{proof}

\begin{prop}[\AXresetlift axiom]%
  \label{p:resetlift-app}
  We have $\reset {\app {\lamp \varx \tmzero}{\reset \tmone}} \appbisim \app
  {\lamp \varx {\reset \tmzero}}{\reset \tmone}$.
\end{prop}

\begin{proof}
  A transition $\reset {\app {\lamp \varx \tmzero}{\reset \tmone}} \ltswk\act
  \tm'$ (with $\act \neq \tau$) is possible only if $\reset \tmone$ evaluates to
  some value $\val$ (evaluation to a control stuck terms is not possible
  according to Proposition~\ref{p:eval-reset}). In this case, we have $\reset
  {\app {\lamp \varx \tmzero}{\reset \tmone}} \ltswk\tau \reset {\app {\lamp
      \varx \tmzero} \val} \lts\tau \reset {\subst \tmzero \varx \val}$ and
  $\app {\lamp \varx {\reset \tmzero}}{\reset \tmone} \ltswk\tau \reset {\subst
    \tmzero \varx \val}$. Therefore, we have $\reset {\app {\lamp \varx
      \tmzero}{\reset \tmone}} \ltswk\act \tm'$ (with $\act \neq \tau$) iff
  $\app {\lamp \varx {\reset \tmzero}}{\reset \tmone} \ltswk\act \tm'$. From
  there, it is easy to conclude.
\end{proof}

\begin{prop}[\AXbetaomega axiom]%
  \label{p:omega-app}
  If $\varx \notin \fv \ctx$, then $\app {\lamp \varx {\inctx \ctx \varx}} \tm
  \appbisim \inctx \ctx \tm$.
\end{prop}

\begin{proof}[Sketch] We first give some intuitions on why the proof of this
  result is hard with applicative bisimulation. The difficult case is when $\tm$
  in the initial terms $\app {\lamp \varx {\inctx \ctx \varx}} \tm$ and
  $\inctx \ctx \tm$ is a control stuck term
  $\inctx \ctxzero {\shift \vark {\tm'}}$. Then we have the following
  transitions:
  \begin{align*}
    \app {\lamp \varx {\inctx \ctx \varx}} \tm & \lts\ctxone \reset{\subst
      {\tm'} \vark {\lam y {\reset{\inctx{\ctxone}{\app {\lamp \varx {\inctx \ctx
            \varx}}{\inctx {\ctxzero} y}}}}}} \\
    \inctx \ctx \tm & \lts\ctxone \reset{\subst{\tm'} \vark {\lam y
        {\reset{\inctx{\ctxone}{\inctx \ctx {\inctx {\ctxzero} y}}}}}}
  \end{align*}
  We obtain terms of the form $\reset{\tm'} \subs$ and $\reset{\tm'} \subs'$
  (where $\subs$ and $\subs'$ are the above substitutions). We now have to
  consider the transitions from these terms, and the interesting case is when
  $\reset{\tm'} = \inctx \rctx {\app \vark \val}$.
  \begin{align*}
    \reset{\tm'} \subs & \lts\tau \inctx {\rctx
      \subs}{\reset{\inctx{\ctxone}{\app {\lamp \varx {\inctx \ctx
              \varx}}{\inctx {\ctxzero}{\val \subs}}}}} \is \tmzero \\
    \reset{\tm'} \subs' & \lts\tau \inctx {\rctx
      \subs'}{\reset{\inctx{\ctxone}{\inctx \ctx {\inctx
            {\ctxzero}{\val \subs'}}}}} \is \tmone
  \end{align*}
  We obtain terms that are similar to the initial terms $\app {\lamp \varx
    {\inctx \ctx \varx}} \tm$ and $\inctx \ctx \tm$, except for the extra
  contexts $\rctx$ and $\ctxone$, and the substitutions $\subs$ and
  $\subs'$. Again, the interesting cases are when $\inctx \ctxzero \val$ is
  either a control stuck term, or a term of the form $\inctx {\rctx'}{\app \vark
    {\val'}}$. Looking at these cases, we see that the bisimulation we have to
  define has to relate terms similar to $\tmzero$ and $\tmone$, except with an
  arbitrary number of contexts $\rctx'$ and substitutions similar to $\subs$ and
  $\subs'$.

  Formally, given a sequence of (continuation) variables
  $k_1,\dots,k_n$ and a sequence $\vect {\ctx_k}$ of triples of contexts
  $\ctx_i, \ctx_i', \ctx_i''$ such that
  \begin{equation}
    \label{eqn:fv}
  \fv{\ctx_i}\cup\fv{\ctx_i'}\cup\fv{\ctx_i''}
  \subseteq
  \{k_1,\dots,k_{i-1}\}
  \quad\quad \text{for } 1 \leq i \leq n
  \tag{$\star$}
  \end{equation}
  we define two families of sequences of
  substitutions as follows:
  \[
  \begin{array}{rcl}
    \sigma^{\vect E}_i
    & = &
    \subst{}
          {k_i}
          {\lam{x}{\reset{\inctx{\ctx_i''}{\app{(\lam{y}{\inctx{\ctx_i}{y}})}{\inctx{\ctx_i'}{x}}}}}}
    \\
    \delta^{\vect E}_i
    & = &
    \subst{}
          {k_i}
          {\lam{x}{\reset{\inctx{\ctx_i''}{\inctx{\ctx_i}{\inctx{\ctx_i'}{x}}}}}}
  \end{array}
  \]

  Additionally, given a term $t$, a sequence of pure contexts
  $\vect{E} = E_1,\dots,E_m$ and a sequence of evaluation contexts
  $\vect{F} = F_1,\dots,F_m$, we inductively define two sequences of
  terms, $s_0,\dots,s_m$ and $u_0,\dots,u_m$, as follows:
  \[
  \begin{array}{ll}
    \begin{array}{rcl}
      s^{t,\vect{E},\vect{F}}_0 & = & t
      \\
      s^{t,\vect{E},\vect{F}}_{i+1} & = &
      \inctx{F_{i+1}}{\app{\lamp{x}{\inctx{E_{i+1}}{x}}}{s_i}}
      \\
    \end{array}
    & %\hspace{2cm}
    \begin{array}{rcl}
      u^{t,\vect{E},\vect{F}}_0 & = & t
      \\
      u^{t,\vect{E},\vect{F}}_{i+1} & = &
      \inctx{F_{i+1}}{\inctx{E_{i+1}}{u_i}}
      \\
    \end{array}
  \end{array}
  \]
  Then the following relation $\rel$ is a bisimulation:
  \begin{eqnarray*}
    %% \rel_1 &=&
    %% \{(\inctx{\rctx}{\app{\lamp{\varx}{\inctx{\ctx}{\varx}}}{\tm}}
    %% \sigma^{\mathbf{E}}_n\dots\sigma^{\mathbf{E}}_0,
    %% \inctx{\rctx}{\inctx{\ctx}{\tm}}
    %% \delta^{\mathbf{E}}_n\dots\delta^{\mathbf{E}}_0
    %% \mid \fv{\tm,\rctx}\subseteq\{\vark_0,\ldots,\vark_n\} \}
    %% \\
    \rel &=&
    \{(s^{t,\vect{E},\vect{F}}_i\sigma^{\vect{E_k}}_n\dots\sigma^{\vect{E_k}}_1,
    u^{t,\vect{E},\vect{F}}_i\delta^{\vect{E_k}}_n\dots\delta^{\vect{E_k}}_1)
    \mid
    \\ & & \hspace{3cm} \vect{k}=k_1,\dots,k_n, n \geq 0,
    \\ & & \hspace{3cm} \vect{E_k} \text{ satisfies ($\star$)},
    \\ & & \hspace{3cm} \vect{E}=E_1,\dots,E_m,
    \vect{F}=F_1,\dots,F_m, m \geq 0,
    \\ & & \hspace{3cm}
    \fv{\tm}\cup\fv{\vect{E}}\cup\fv{\vect{F}}\subseteq\{\vark_1,\ldots,\vark_n\},
    \\ & & \hspace{3cm} 0 \leq i \leq m
    \}
  \end{eqnarray*}
  We omit the complete bisimulation proof, as we provide much simpler proofs of
  this result with environmental or normal-form bisimilarities (see
  Propositions~\ref{p:omega-env} and~\ref{p:omega-nf}).
\end{proof}

\subsection{Conclusion}

We define an applicative bisimilarity for the relaxed semantics of $\lamshift$
which extends the $\lambda$-calculus definition with a transition for
control-stuck terms. Soundness can be proved by adapting Howe's method to this
extra transition, and we can also show completeness {\wrt}$\ctxequivp$ as well
as a context lemma. However, we do not know how to extend these results to the
original semantics of $\lamshift$. While we can think of an applicative
bisimilarity for the original semantics by adapting the environmental
bisimilarity we define in Section~\ref{ss:env-LTS-original}, we do not know how
to prove it sound with Howe's technique. Roughly, Howe's technique fails because
it requires the semantics to be preserved by all evaluation contexts, while the
original semantics is preserved only by contexts with an outermost reset.

Another issue is that equivalence proofs with applicative bisimulation can be
difficult, as witnessed by Example~\ref{ex:fixed-point-app} or
Proposition~\ref{p:omega-app}. We believe it is due to the lack of powerful
up-to techniques, in particular the absence of bisimulation up to context, which
reveals to be problematic in a calculus where context capture and manipulation
is part of the semantics. As a result, applicative bisimulation seems suitable
only for simple examples, such as Proposition~\ref{p:cex-nf-completeness}.

% Local variables:
% mode: latex
% TeX-master: "journal.tex"
% End:

\section{Environmental Bisimilarity}%
\label{s:env}
Like applicative bisimilarity, environmental bisimilarity reduces closed terms
to normal forms, which are then compared using some particular contexts (e.g.,
$\lambda$-abstractions are tested by passing them arguments). However, the
testing contexts are not arbitrary, but built from an environment, which
represents the knowledge acquired so far by an outside observer. The idea
originally comes from languages with strict isolation or data
abstraction~\cite{Sumii-Pierce:TCS07,Sumii-Pierce:JACM07,Koutavas-Wand:ESOP06,Koutavas-Wand:POPL06},
where environments are used to handle information hiding.
The term ``environmental bisimulation'' has then been introduced
in~\cite{Sangiorgi-al:LICS07,Sangiorgi-al:TOPLAS11}, and such a
bisimilarity has been since defined in various higher-order languages
(see, e.g.,~\cite{Sato-Sumii:APLAS09,Sumii:TCS10,Pierard-Sumii:LICS12}), including
the $\lambda$-calculus with first-class abortive
continuations~\cite{Yachi-Sumii:APLAS16}. Environmental bisimilarity
usually characterizes contextual equivalence, but is harder to use
than applicative bisimilarity to prove that two given terms are
equivalent. Nonetheless, one can define powerful up-to
techniques~\cite{Sangiorgi-al:TOPLAS11} to simplify the equivalence
proofs and deal with this extra difficulty. Besides, the authors
of~\cite{Koutavas-al:MFPS11} argue that the additional complexity is
necessary to handle more realistic features, like local state or
exceptions.

Recently, the notion of environmental bisimilarity has been cast in a framework
in which soundness proofs for the bisimilarity and its up-to techniques are
factorized~\cite{Madiot-al:CONCUR14,Madiot:PhD}. We extended that framework to
allow for more powerful up-to techniques that are better suited for
delimited-control operators~\cite{Aristizabal-al:LMCS17}. We informally explain
in Section~\ref{ss:env-informal} why we need such an extension in $\lamshift$,
before presenting the extended framework in Section~\ref{ss:diacritical} and the
definition of the bisimilarity itself, first for the relaxed semantics in
Section~\ref{ss:env-LTS} and then the original one in
Section~\ref{ss:env-LTS-original}. We improve the bisimilarities with up-to
techniques (Section~\ref{ss:env-upto}) that we apply to examples
(Section~\ref{ss:env-examples}), and in particular to the Kameyama and Hasegawa
axiomatization (Section~\ref{ss:axioms-env}).

An older work~\cite{Biernacki-Lenglet:APLAS13} gives definitions of
environmental bisimulations that are now completely obsolete. We revisit results
originally published in a previous article~\cite{Aristizabal-al:LMCS17}, where
the focus is more on a multi-prompted calculus. More precisely,
Section~\ref{ss:env-informal} is rewritten for $\lamshift$ from~\cite[Section
4.1]{Aristizabal-al:LMCS17} Section~\ref{ss:diacritical} covers~\cite[Section
4.3]{Aristizabal-al:LMCS17}, and
Sections~\ref{ss:env-LTS},~\ref{ss:env-LTS-original}, and~\ref{ss:env-upto}
provide more details that~\cite[Section 5.2]{Aristizabal-al:LMCS17}. The
examples of Sections~\ref{ss:env-examples} and~\ref{ss:axioms-env} are a
contribution of the present article.

\subsection{Informal Presentation}%
\label{ss:env-informal}

In the original formulation of environmental
bisimulation~\cite{Sangiorgi-al:TOPLAS11}, two terms~$\tmzero$ and
$\tmone$ are compared under some environment $\mathcal E$, which
represents the knowledge of an external observer about $\tmzero$ and
$\tmone$. The definition of the bisimulation enforces some conditions
on~$\tmzero$ and $\tmone$ as well as on $\mathcal E$. In Madiot et
al.'s framework~\cite{Madiot-al:CONCUR14,Madiot:PhD}, the conditions
on~$\tmzero$, $\tmone$, and $\mathcal E$ are expressed using a LTS
between \emph{states} of the form $\state \env \tmzero$ and $\state
\envd \tmone$ as well as between states of the form $\env$ and
$\envd$, where $\env$ and $\envd$ are finite sequences of values
corresponding to the first and second projection of the environment
$\mathcal E$, respectively.  Transitions from states of the form
$\state \env \tmzero$ express conditions on $\tmzero$, while
transitions from states of the form $\env$ explain how we compare
environments. Henceforth, if $m$ ranges over a sequence of entities,
we write $m_i$ for the~$i^{\text{ th}}$ element of the sequence.

For the relaxed semantics of $\lamshift$, one could think of extending
the LTS for the $\lambda$-calculus~\cite{Madiot:PhD} (the first three
rules below) with an extra transition for testing stuck terms.
\begin{mathpar}
  \inferrule{\tmzero \redcbv \tmone}
  {\state \env \tmzero \lts\tau \state \env \tmone}
  \and
  \inferrule{ }
  { \state \env \val \ltsval (\env, \val)}
  \and
   \inferrule{\env_i = \lam x \tm}
  {\env \ltslam i {\vmhc} \state \env {\subst \tm x
    {\inctxe{\vmhc}{\env}}}}
  \and
  \inferrule{\tmzero \mbox{ is control-stuck} \\ \reset {\inctx \emhc {\tmzero; \env}} \redcbv
    \tmone}
  {\state \env \tmzero \ltsstuck \emhc \state \env \tmone}
\end{mathpar}

\noindent We use \emph{multi-hole contexts} $\vmhc$ and $\emhc$ to
build respectively values and pure evaluation contexts from an
environment $\env$; such contexts contain numbered holes $\mtctx_i$ to
be filled with~$\env_i$. For example, $\inctx {(\lam \varx {(\mtctx_1
    \iapp \mtctx_3) \iapp (\varx \iapp \mtctx_3)})} \env = \lam \varx
{(\env_1 \iapp \env_3) \iapp (\varx \iapp \env_3)}$, assuming $\env$
is at least of size 3. Internal steps $\lts\tau$ correspond to
reduction steps. The transition $\ltsval$ turns a state $\state \env
\val$ into a sequence of values; when we are done evaluating a term,
we can add the newly acquired knowledge to the
environment. Environments are tested with the transition $\ltslam i
{\vmhc}$, which means that the $i^{\text{th}}$ element of $\env$ is
tested by applying it to an argument built using
$\vmhc$. Finally,~$\lts\emhc$ tests control-stuck terms by putting
them in a context built from $\emhc$ to trigger the capture, where the
notation $\inctx \emhc {\tmzero; \env}$ means that the hole in the
evaluation position in $\emhc$ is plugged with $\tmzero$ while the
numbered holes are plugged with $\env_i$.

The transitions $\lts\tau$, $\ltsstuck \emhc$, and $\ltslam i \vmhc$ correspond
to the transitions $\lts\tau$, $\lts\ctx$, and $\lts\val$ defining applicative
bisimulation, except the testing arguments are built from
the environment. As a result, plain environmental bisimulation proofs are harder
than applicative ones, as witnessed by the following example.

\begin{exa}[Turing's combinator]%
  \label{ex:fixed-point-env}
  Following Example~\ref{ex:fixed-point-app}, we want to prove that Turing's
  combinator $\Turing$ is bisimilar to its variant
  $\Turingshift \is \reset {\app \theta {\shift \vark {\app \vark \vark}}}$.
  We remind that
  \begin{align*}
    \Turing & \evalcbv \lam y {\app y {\lamp z {\app {\app {\app \theta
              \theta} y} z}}} \is \valzero \mbox{, and} \\
    \Turingshift & \evalcbv \lam y {\app y {\lamp z {\app {\app {\app {\lamp
                \varx {\reset {\app \theta \varx}}}{\lamp \varx {\reset {\app
                    \theta \varx}}}} y} z}}} \is \valone.
  \end{align*}
  Let $\env \is (\valzero)$ and $\envd \is (\valone)$; then
  \begin{align*}
    \env & \ltslam 1 \vmhc \state \env {\app {\inctx \vmhc \env}{\lamp z {\app
        {\app {\app \theta \theta} {\inctx \vmhc \env}} z}}} \mbox{ and} \\
    \envd & \ltslam 1 \vmhc \state \envd {\app {\inctx \vmhc \envd}{\lamp z {\app {\app {\app
            {\lamp \varx {\reset {\app \theta \varx}}}{\lamp \varx {\reset {\app
            \theta \varx}}}} {\inctx \vmhc \envd }} z}}}.
  \end{align*}
  Because we have different terms $\inctx \vmhc \env$ and $\inctx \vmhc \envd$
  and not a single value $\val$, the case analysis suggested in
  Example~\ref{ex:fixed-point-app} becomes much more complex, as we have to take
  into account how $\vmhc$ uses $\env$ or~$\envd$.
\end{exa}

Up-to techniques are what makes environmental bisimulation tractable,
in particular bisimulation up to context, which allows to factor out a
common context: when comparing states of the form $\state \env {\inctx
  \mhc \env}$ and $\state \envd {\inctx \mhc \envd}$, where $\mhc$ is
a multi-hole context, we can forget about $\mhc$ and focus on $\env$
and $\envd$. Similarly for $\state \env {\inctx \fmhc {\tmzero;
    \env}}$ and $\state \envd {\inctx \fmhc {\tmone; \envd}}$, where
$\fmhc$ is a multi-hole evaluation context, we can consider only
$\state \env \tmzero$ and $\state \envd \tmone$; the restriction to
evaluation contexts is necessary for the technique to be sound, as
pointed out by Madiot~\cite[page 111]{Madiot:PhD}.  Bisimulation up to
context is unfortunately not powerful enough to be useful
in~$\lamshift$. Suppose we want to prove a variant of the
$\beta_\Omega$ axiom, $\reset{\app {\lamp x {\inctx \ctx x}} \tm}$
equivalent to~$\reset{\inctx \ctx \tm}$ if $x \notin \fv \ctx$. If
$\tm = \inctx {\ctx'} {(\shift \vark {\app \vark \vark}) \iapp \val}$
for some $\ctx'$ and $\val$, then
\begin{align*}
  \state \emptyset {\reset{(\lam x {\inctx \ctx x}) \iapp \inctx {\ctx'}{(\shift
  \vark {\app \vark \vark}) \iapp \val}}} & \ltswk\tau \state \emptyset {\reset{\reset{\app{(\lam x
                                            {\inctx E x})}{\inctx {\ctx'}
                                            {\reset{\app{(\lam x {\inctx E
                                            x})}{\inctx {\ctx'} {\app v v}} }}}}}} \mbox
                                            { and}\\
  \state \emptyset {\reset{\inctx E {\inctx {\ctx'}{(\shift k {\app k k}) \iapp v}}}}
                                          & \ltswk\tau \state \emptyset
                                            {\reset{\reset{\inctx E {\inctx {\ctx'}
                                            {\reset{\inctx E {\inctx {\ctx'} {\app v v}} }}}}}}.
\end{align*}
The two resulting terms do not share a common evaluation context beyond
$\reset{\reset \mtctx}$, so bisimulation up to context cannot simplify the proof
from there.

Yet we can see that the two resulting terms have the same shape,
except for the contexts $\lamp x {\inctx E x} \iapp \mtctx$ and
$E$. Following this observation, in a previous
work~\cite{Aristizabal-al:LMCS17}, we proposed a more expressive
notion of bisimulation up to context where the common context $\fmhc$
can be built out of \emph{related evaluation contexts}. We do so by
adding to the syntax of multi-hole contexts the constructs
$\appcont{\holecont_i} \mhc$ and $\appcont {\holecont_i} \fmhc$, where
the hole $\holecont_i$ can be filled by an evaluation context $\rctx$
to produce respectively $\inctx \rctx \mhc$ and $\inctx \rctx
\fmhc$. We also include sequences of evaluation contexts $\enve$
or~$\envef$ in the LTS states $\stat \enve \env \tm$ and $\statt \enve
\env$. As a result, if $\enve \is (\lamp x {\inctx \ctx x} \iapp
\mtctx)$, $\envef \is (\ctx)$, and $\mhc \is \reset{\reset {\appcont
    {\holecont_1}{\inctx {\ctx'}{\appcont {\holecont_1}{\inctx
          {\ctx'}{\app \val \val}}}}}}$, then
\begin{align*}
  \inctx \mhc {\enve; \emptyset} & = \reset{\reset{\app{(\lam
                                         x
                                         {\inctx E x})}{\inctx {\ctx'}
                                         {\reset{\app{(\lam x {\inctx E
                                         x})}{\inctx {\ctx'} {\app v v}} }}}}} \mbox
                                         { and}\\
  \inctx \mhc {\envef; \emptyset} & = \reset{\reset{\inctx E {\inctx {\ctx'}
                                          {\reset{\inctx E {\inctx {\ctx'} {\app v v}} }}}}}
\end{align*}
so $\mhc$ can be factored out using our notion of bisimulation up to related
contexts.

Extending the state to include evaluation contexts means that these contexts
have to be tested, by plugging them with an argument built from the environment.
\begin{mathpar}
  \inferrule{ }
  {\statt \enve \env \ltsctxx j {\vmhc} \stat \enve \env {\inctx {\enve_j}
      {\inctxe{\vmhc}{\enve; \env}}}}
\end{mathpar}
However, such a transition is problematic in conjunction with our notion of
bisimulation up to related contexts. Indeed, for all $\rctxzero$ and $\rctxone$,
we have
\begin{align*}
  \statt \rctxzero \emptyset & \ltsctxx 1 \vmhc \stat \rctx \emptyset {\inctx
  \rctxzero {\inctx \vmhc {\rctxzero; \emptyset}}} \mbox{ and}\\
  \statt \rctxone \emptyset & \ltsctxx 1 \vmhc \stat {\rctx'} \emptyset {\inctx
  {\rctxone} {\inctx \vmhc {\rctxone; \emptyset}}}.
\end{align*}
But the two resulting states are bisimilar up to related contexts, since for all
$\rctx$,
$\inctx {\rctx} {\inctxe{\vmhc}{\rctx; \emptyset}} = \inctx{(\appcont
  {\holecont_1} \vmhc)}{\rctx; \emptyset}$. If bisimulation up to related
contexts is a valid up-to technique, it implies that
$\statt \rctxzero \emptyset$ and $\statt \rctxone \emptyset$ are bisimilar for
any $\rctxzero$ and $\rctxone$, which is obviously false (consider
$\rctxzero = \mtctx$ and $\rctxone = \mtctx \iapp \Omega$). To prevent this, we
distinguish \emph{passive} transitions (such as $\ltsctxx i \vmhc$) from the
other ones (called \emph{active}), so that only selected up-to techniques
(referred to as \emph{strong}) can be used after a passive transition. In
contrast, any up-to technique (including bisimulation up to related contexts)
can be used after an active transition. To formalize this idea, we extend Madiot
et al.'s framework to allow such distinctions between transitions and between
up-to techniques. We present the definitions in a general setting in
Section~\ref{ss:diacritical}, before illustrating them with environmental
bisimilarity for $\lamshift$.

\subsection{Diacritical Progress and Up-to Techniques}%
\label{ss:diacritical}

We recall the main definitions and results of the extended framework from our
previous work~\cite{Aristizabal-al:LMCS17}; see this paper for more details.

\subsubsection*{Diacritical progress} Let $\lts\act$ be a LTS defined on states
ranged over by $\sts$ or $\stt$, which contains an internal action labeled
$\tau$. Weak transitions $\ltswk\alpha$ are defined as
$\ltswk\tau \is \rtclo{\lts\tau}$ and
$\ltswk\act \is \ltswk\tau \lts\act \ltswk\tau$ if $\act \neq \tau$.  A (weak)
bisimulation over this LTS can be defined using a notion of \emph{progress}: a
relation $\rel$ progresses towards $\rels$, written $\rel \progress \rels$, if
$\sts \rel \stt$ implies that if $\sts \lts\act \sts'$, there exists $\stt'$
such that $\stt \ltswk \act \stt'$ and $\sts' \rels \stt'$, and conversely if
$\stt \lts\act \stt'$. A bisimulation is then defined as a relation $\rel$
verifying $\rel \progress \rel$, and bisimilarity is the largest bisimulation.

In our extended framework, we suppose that the transitions of the LTS
are partitioned into \emph{passive} and \emph{active} transitions, and
we define \emph{diacritical progress} as follows.

\begin{defi}%[Diacritical progress]
  A relation $\rel$ diacritically progresses to $\rels$, $\relt$ written $\rel
  \pprogress \rels, \relt$, if $\rel \mathop\subseteq \rels$, $\rels
  \mathop\subseteq \relt$,  and $\sts \rel \stt$ implies that
  \begin{itemize}
  \item if $\sts \lts\act \sts'$ and $\lts\act$ is passive, then
    there exists $\stt'$ such that $\stt \ltswk\act \stt'$ and $\sts' \rels \stt'$;
  \item if $\sts \lts\act \sts'$ and $\lts\act$ is active, then
    there exists $\stt'$ such that $\stt \ltswk\act \stt'$ and $\sts' \relt \stt'$;
  \item the converse of the above conditions on $\stt$.
  \end{itemize}
\end{defi}
\noindent A bisimulation is a relation $\rel$ such that
$\rel \pprogress \rel, \rel$, and bisimilarity $\bis$ is the largest
bisimulation. Since a bisimulation $\rel$ progresses towards $\rel$ after both
passive and active transitions, the two notions of progress $\progress$ and
$\pprogress$ in fact generate the same notions of bisimulation and bisimilarity;
the distinction between active and passive transitions is interesting only when
considering up-to techniques.

\subsubsection*{Up-to techniques} The goal of up-to techniques is to simplify
bisimulation proofs: instead of proving that a relation $\rel$ is a
bisimulation, we show that $\rel$ respects some looser constraints which still
imply bisimilarity $\bis$. In our setting, we distinguish the up-to techniques
which can be used after a passive transition (called \emph{strong} up-to
techniques), from the ones which cannot. An up-to technique (resp.\ strong up-to
technique) is a function $f$ such that $\rel \pprogress \rel, f(\rel)$ (resp.\
$\rel \pprogress f(\rel), f(\rel)$) implies $\rel \mathop\subseteq
\bis$. Proving that a given~$f$ is an up-to technique is difficult with this
definition, so following Madiot, Pous, and
Sangiorgi~\cite{Sangiorgi-Pous:11,Madiot-al:CONCUR14}, we rely on a notion
of~\emph{respectfulness}, which gives sufficient conditions for~$f$ to be an
up-to technique, and is easier to establish, as functions built out of
respectful functions using composition and union remain respectful.

We first need some auxiliary notions on notations on functions on relations,
ranged over by $f$, $g$, and $h$ in what follows. We define $f \subseteq g$ and
$f \cup g$ argument-wise, e.g., $(f \cup g)(\rel)=f(\rel) \cup g(\rel)$ for all
$\rel$. We define $f^\omega$ as $\bigcup_{n \in \mathbb N} f^n$. We write
$\rawid$ for the identity function on relations, and $\fid f$ for
$f \mathop\cup \rawid$. Given a set $\setF$ of functions, we also write~$\setF$
for the function defined as $\bigcup_{f \in \setF} f$. We say a function $f$ is
\emph{generated from $\setF$} if $f$ can be built from functions in $\setF$ and
$\rawid$ using union, composition, and $\cdot^\omega$. The largest function
generated from~$\setF$ is ${\fid \setF}^\omega$. A function~$f$ is monotone if
$\rel \mathop\subseteq \rels$ implies $f(\rel) \mathop\subseteq f(\rels)$.  We
write $\finpower\rel$ for the set of finite subsets of~$\rel$, and we say $f$ is
continuous if it can be defined by its image on these finite subsets, i.e., if
$f(\rel) \mathop\subseteq \bigcup_{\rels \in \finpower\rel}f(\rels)$. The up-to
techniques of the present paper are defined by inference rules with a finite
number of premises, so they are trivially
continuous. % Continuous functions are interesting
% because of their properties:

% \begin{lem}
%   \label{l:continuity}
%   If $f$ and $g$ are continuous, then $f \compo g$ and $f \cup g$ are continuous.

%   If $f$ is continuous, then $f$ is monotone, and
%   $f \compo {\fid f}^\omega \subseteq {\fid f}^\omega$.
% \end{lem}

\begin{defi}%[Evolution, strong evolution]
  \label{d:evolution}
  A function $f$ evolves to $g, h$, written $f \fevolve g, h$, if for
  all $\rel \pprogress \rel, \relt$, we have $f(\rel) \pprogress
  g(\rel), h(\relt)$. A function $f$ \emph{strongly} evolves to $g,
  h$, written $f \sevolve g, h$, if for all $\rel \pprogress \rels,
  \relt$, we have $f(\rel) \pprogress g(\rels), h(\relt)$.
\end{defi}
\noindent
Evolution can be seen as a notion of progress for functions on
relations. Note that strong evolution does not put any condition on
how $\rel$ progresses, while regular evolution is more restricted, as
it requires a relation $\rel$ such that $\rel \pprogress \rel, \relt$.

\begin{defi}%[Diacritical compatibility]
  \label{d:dia-comp}
  A set $\setF$ of continuous functions is \emph{diacritically
    respectful} if there exists $\setS$ such that $\setS \subseteq
  \setF$ and
  \begin{itemize}
  \item for all $f \in \setS$, we have $f \sevolve {\fid\setS}^\omega,
    {\fid\setF}^\omega$;
  \item for all $f \in \setF$, we have $f \fevolve {\fid\setS}^\omega
    \compo \fid\setF \compo {\fid\setS}^\omega, {\fid\setF}^\omega$.
  \end{itemize}
\end{defi}
\noindent
In words, a function is in a respectful set $\setF$ if it evolves
towards a combination of functions in $\setF$.  The (possibly empty)
subset $\setS$ intuitively represents the strong up-to techniques of
$\setF$. Any combination of functions can be used after an active
transition. After a passive one, only strong functions can be used, except in
the second case, where we progress from $f(\rel)$, with $f$ not
strong. In that case, it is expected to progress towards a combination
that includes $f$; it is safe to do so, as long as $f$ (or in fact,
any non-strong function in $\setF$) is used at most once. If $\setS_1$
and~$\setS_2$ are subsets of~$\setF$ which verify the conditions of
the definition, then $\setS_1 \cup \setS_2$ also does, so there exists
the largest subset of $\setF$ which satisfies the conditions, written
$\strong \setF$.

\begin{prop}%[Properties of diacritical compatibility]
  \label{p:properties-compatibility-better}
  Let $\setF$ be a diacritically compatible set.
  \begin{itemize}
  \item If
    $\rel \pprogress {\widehat{\mathsf{strong}(\setF)}}^\omega(\rel),
    {\fid\setF}^\omega(\rel)$, then ${\fid\setF}^\omega(\rel)$ is a
    bisimulation.
  \item any function generated from $\setF$ is an up-to technique, and any
    function generated from $\strong\setF$ is a strong up-to technique.
  \item For all $f \in \setF$, we have $f(\bisim) \mathop\subseteq \bisim$.
  \end{itemize}
\end{prop}

\noindent The second point implies that combining functions from a respectful
set using union, composition, or $\cdot^\omega$ produces up-to techniques. In
particular, if $f \in \setF$, then $f$ is an up-to technique, and similarly, if
$f \in \strong\setF$, then $f$ is a strong up-to technique. In practice, proving
that $f$ is in a respectful set~$\setF$ is easier than proving it is an up-to
technique. The last item states that bisimilarity is preserved by respectful
functions, so proving that up to context is respectful implies that bisimilarity
is preserved by contexts.

The first item suggests a more flexible notion of up-to technique, as it shows
that given a respectful set $\setF$, a relation may progress towards different
functions $f$ and $g$, $\rel \pprogress f(\rel), g(\rel)$, and still be included
in the bisimilarity as long as $f$ is generated from $\strong\setF$ and $g$ is
generated from $\setF$. In what follows, we rely on that property in examples
and say that in that case, $\rel$ is a bisimulation up to $\setF$, or $\rel$ is
a bisimulation up to $f_1 \ldots f_n$ if $\setF = \{ f_1 \ldots f_n \}$.

\begin{rem}[respectful vs compatible functions]
  The literature distinguishes between respectful~\cite{Sangiorgi:MSCS98} and
  \emph{compatible}~\cite{Pous:APLAS07} functions: $f$ is respectful if
  $\rel \progress \rels$ and $\rel \mathord\subseteq \rels$ implies
  $f(\rel) \progress f(\rels)$, while $f$ is compatible if
  $\rel \progress \rels$ implies $f(\rel) \progress f(\rels)$. Some interesting
  up-to techniques are not compatible but are respectful thanks to the extra
  inclusion hypothesis. Mimicking~\cite{Madiot-al:CONCUR14,Madiot:PhD}, we use
  the term ``compatible'' instead of ``respectful'' in our previous
  work~\cite{Aristizabal-al:LMCS17} for the definition with
  the extra inclusion hypothesis. We use ``respectful'' in this paper to be
  faithful to the original definitions, and because we use ``compatible'' for
  relations preserved by the operators of the language. Pous~\cite{Pous:LICS16}
  argues that the difference between the two notions is of little importance
  anyway as they generate the same companion function.
\end{rem}

\begin{rem}%
  \label{r:no-lts}
  As a matter of fact, the theory we present in this section does not
  require an underlying notion of LTS\@. In particular,
  Definition~\ref{d:evolution} and~\ref{d:dia-comp} as well as the
  proof of Proposition~\ref{p:properties-compatibility-better} do not
  depend on the notion of diacritical progress being defined in terms
  of a LTS\@. As long as the notion of progress satisfies the following
  (simple) conditions:
  \begin{itemize}
  \item if $R \pprogress S, T$, $S \subseteq S'$, and $T \subseteq
    T'$, then $R \pprogress S', T'$;
  \item if $\forall i \in I. R_i \pprogress S, T$, then $\bigcup_{i
    \in I} R_i \pprogress S, T$,
  \end{itemize}
  the presented theory is valid for such a notion of progress. We
  exploit this fact in Section~\ref{ss:nf-def}, when defining a
  normal-form bisimilarity.
\end{rem}

\subsection{Bisimilarity for the Relaxed Semantics}%
\label{ss:env-LTS}

We define environmental bisimulation using a LTS between states of the
form $\stat \enve \env \tm$ (called \emph{term states}) or $\statt
\enve \env$ (called \emph{environment states}), where we denote by
$\env$ or $\envd$ a sequence of closed values, and by $\enve$ or
$\envef$ a sequence of closed evaluation contexts, and where $t$ is a
closed term.
%A state is closed if it is composed of closed terms.
As explained in Section~\ref{ss:env-informal}, the values are used to
build testing arguments to compare $\lambda$-abstractions, while we
store evaluation contexts to define bisimulation up to related
contexts. We build the testing entities out of $\enve$ and $\env$
using \emph{multi-hole contexts}, defined as follows.
\begin{grammar}
  & \textrm{Contexts:} & \mhc & \bnfdef \vmhc \bnfor \app \mhc \mhc \bnfor
  \reset \mhc \bnfor \shift \vark {\mhc} \bnfor \appcont {\holecont_j} \mhc \\
  & \textrm{Value contexts:} & \vmhc & \bnfdef \var{x} \bnfor \lam{x}{\mhc}
  \bnfor \mtctx_i \\
  & \textrm{Evaluation contexts:} \quad & \fmhc & \bnfdef \mtectx \bnfor
  \argectx{\fmhc}{\mhc} \bnfor \valectx{\vmhc}{\fmhc} \bnfor \reset{\fmhc}
  \bnfor \appcont{\holecont_j} \fmhc
\end{grammar}

We distinguish value holes $\mtctx_i$ from context holes $\holecont_j$. These
holes are indexed, unlike the special hole $\mtctx$ of an evaluation context
$\fmhc$, which is in evaluation position (that is, filling the other holes of
$\fmhc$ gives a regular evaluation context~$F$).  Filling the holes of $\mhc$
and $\vmhc$ with~$\enve$ and $\env$, written respectively
$\inctxmhc \mhc \enve \env$ and $\inctxmhc \vmhc \enve \env$, consists in
replacing any subterm of the form $\appcont{\holecont_j}{\mhc'}$ with
$\inctx{\enve_j}{\mhc'}$ and any occurrence of $\mtctx_i$ with $\env_i$,
assuming that $j$ is smaller or equal than the size of $\enve$ and similarly for
$i$ {\wrt}$\env$.  We write $\inctxfmhc \fmhc \tm \enve \env$ for the same
operation with evaluation contexts, where we assume that $\tm$ is put in
$\mtctx$. We extend the notion of free variables to multi-hole contexts as
expected, and a multi-hole context is said closed if it has no free variables.

\begin{figure}
\begin{mathpar}
  \inferrule{\tm \redcbv \tm'}
  {\stat \enve \env \tm \lts\tau \stat \enve \env {\tm'}}
  \and
  \inferrule{\env_i = \lam x \tm}
  {\statt \enve \env \ltslam i \vmhc \stat \enve \env {\subst \tm x
    {\inctxmhc \vmhc \enve \env}}}
  \and
  \inferrule{ }
  { \stat \enve \env \val \ltsval \statt \enve {\env, \val} }
  \and
  \inferrule{\tm \mbox{ is stuck} \\ \inctxfmhc \fmhc \tm \enve \env \redbis \tm'}
  { \stat \enve \env \tm \ltsstuck \fmhc \stat \enve \env {\tm'}}
  \and
  \inferrule{ }
  { \statt \enve \env \ltsctxx j \vmhc \stat \enve \env {\inctx {\enve_j} {\inctxmhc
      \vmhc \enve \env}}}
  \\
  \inferrule{\enve_j = \ctx}
  { \statt \enve \env \ltspure j \statt \enve \env}
  \and
  \inferrule{\enve_j = \inctx \rctx {\reset \ctx}}
  { \statt \enve \env \ltsnotpure j \statt {\enve, \inctx F {\reset \mtctx},
      \reset \ctx} \env}
\end{mathpar}

\caption{LTS for the relaxed semantics}%
\label{fig:lts-env}
\end{figure}

Figure~\ref{fig:lts-env} presents the LTS $\lts\act$ for the relaxed semantics of
$\lamshift$, where the relation $\tm \redbis \tm'$ is defined as follows: if
$\tm \redcbv \tm'$, then $\tm \redbis \tm'$, and if $\tm$ is a normal form, then
$\tm \redbis \tm$.\footnote{The relation $\redbis$ is not exactly the reflexive
  closure of $\redcbv$, since an expression which is not a normal form
  \emph{must} reduce.} The multi-hole contexts $\vmhc$ and $\fmhc$ used in the
transition $\ltsstuck \fmhc$, $\ltslam i \vmhc$, and $\ltsctxx j \vmhc$ are
supposed to be closed. The internal transition $\lts\tau$ corresponds to
reduction. The transition $\ltslam i \vmhc$ tests the
$\lambda$-abstraction~$\env_i$ by passing it an argument built with $\vmhc$. The
transition can be fired for any~$i$ smaller than the size of $\env$, ensuring
that all the values in $\env$ are tested. The transition~$\ltsval$ turns a term
state $\stat \enve \env \val$ into an environment state
$\statt \enve {\env, \val}$ since $\val$ cannot reduce further.

The transition $\ltsstuck \fmhc$ compares stuck terms by putting them in an
evaluation context $\fmhc$ to trigger the capture, like the corresponding
transition in applicative bisimulation. However, characterizing the
capture-triggering contexts is more difficult than in Section~\ref{s:app},
because of holes $\holecont_i$. Indeed, a context $\appcont {\holecont_i} \emhc$
may also provoke a capture if $\enve_i$ is an impure context: for example, we
have
$\stat {\reset \mtctx} \emptyset {\shift \vark \tm} \ltsstuck {\appcont
  {\holecont_1} \mtctx} \stat {\reset \mtctx} \emptyset {\reset{\subst \tm \vark
    {\lam x {\reset x}}}}$. Instead of looking for a precise characterization,
we simply test with all context $\fmhc$, and then trigger the capture using
$\redbis$ only when possible, i.e., when $\inctxmhc \fmhc \enve \env$ contains a
surrounding reset. A uninteresting transition
$\stat \enve \env \tm \ltsstuck \fmhc \stat \enve \env {\inctxfmhc \fmhc \tm
  \enve \env}$ where no capture happens will then be easily dealt with up-to
techniques (see Example~\ref{ex:stuck-utctx}).

The remaining transitions $\ltsctxx j {\vmhc}$, $\ltspure j$, and
$\ltsnotpure j$ deal with the evaluation contexts in~$\enve$, and are therefore
applied only in conjunction with bisimulation up to related contexts, as~$\enve$
is not empty only in that case. The transition $\ltsctxx j {\vmhc}$ tests the
evaluation context $\enve_j$ by passing it a value built from $\vmhc$, the same
way $\lambda$-abstractions are tested with $\ltslam i \vmhc$. Testing an
evaluation context with a value is simpler than with any term, however it does
not account for all the possible interactions of a term with an evaluation
context. Indeed, a stuck term is able to distinguish a pure context from an
impure one, and it can extract from $\inctx F {\reset E}$ the context up to the
first enclosing reset $\reset E$. We use $\ltspure j$ and $\ltsnotpure j$ to
mimic these behaviors. The transition $\ltspure j$ simply states that $\enve_j$
is pure; in a bisimulation $\rel$ such that
$\statt \enve \env \rel \statt \envef \envd$ and
$\statt \enve \env \ltspure j \statt \enve \env$, $\statt \envef \envd$ has to
match with the same transition, meaning that $\envef_j$ must also be pure. Similarly,
$\ltsnotpure j$ decomposes $\enve_j = \inctx F {\reset E}$ into
$\inctx F {\reset \mtctx}$ and~$\reset E$.  Because the transition leaves a
\textreset inside $F$, applying the same transition again to
$\inctx F {\reset \mtctx}$ does not decompose~$F$ further, but simply generates
$\inctx F {\reset \mtctx}$ again (and $\reset \mtctx$). The duplicated contexts
can then be ignored using up-to techniques.

To define environmental bisimulation using diacritical progress
(Section~\ref{ss:diacritical}), we distinguish the transitions
$\ltsctxx j \vmhc$ and $\ltsval$ as passive, while the remaining others are
active. We consider a transition as passive if it can be inverted by an up-to
technique, which is possible if no new information is generated between its
source and target states. For example,
$\stat \enve \env \val \ltsval \statt \enve {\env, \val}$ is passive because we
simply change the nature of the state (from term to environment). In contrast,
the transition $ \statt \enve \env \ltspure j \statt \enve \env$ is active, as
we gain some information: $\env_j$ is a pure context. The transition
$\statt \enve \env \ltsctxx j \vmhc \stat \enve \env \tm$ is passive at it
simply recombines existing information in $\env$ and $\enve$ to build $\env$,
without any reduction step taking place, and thus without generating new
information. Some extra knowledge is produced only when $\stat \enve \env \tm$
evolves (with active transitions), as it then tells us how the tested
context~$\env_j$ actually interacts with the value constructed from
$\vmhc$. Finally, $\ltslam i {\cval \mhc}$ and $\ltsstuck \emhc$ correspond to
reduction steps and are therefore active, and $\ltsnotpure j$ is also active as
it provides some information by telling us how to decompose a context.

\begin{defi}%
  \label{d:env-relaxed}
  A relation $\rel$ on states is an environmental bisimulation if
  $\rel \pprogress \rel, \rel$. Environmental bisimilarity $\envbisim$
  is the largest environmental bisimulation.
\end{defi}

\noindent We extend $\envbisim$ to open terms as follows: if
$\vect \varx=\fv \tmzero \cup \fv\tmone$, then we write
$\tmzero \open\envbisim \tmone$ if
$\statt \emptyset {\lam {\vect \varx} \tmzero} \envbisim \statt \emptyset
{\lam{\vect \varx} \tmone}$. We discuss the soudness and completeness of
$\envbisim$ in Section~\ref{ss:env-upto}, after giving the definition of the
bisimilarity for the original semantics.

\begin{figure}
\begin{mathpar}
  \inferrule{\prg \redcbv \prg'}
  {\stat \enve \env \prg \lts\tau \stat \enve \env {\prg'}}
  \and
  \inferrule{\env_i = \lam x \tm \\ \inctxmhc \fmhc \enve \env \mbox{ is
      delimited} }
  {\statt \enve \env \ltslamo i \vmhc \fmhc \stat \enve \env {\inctxfmhc
        \fmhc {\subst \tm x
    {\inctxmhc \vmhc \enve \env}} \enve \env}}
  \and
  \inferrule{ }
  { \stat \enve \env \val \ltsval \statt \enve {\env, \val} }
  \and
  \inferrule{ \inctxmhc \fmhc \enve \env \mbox{ is delimited}  }
  { \statt \enve \env \ltsctxxo j \vmhc \fmhc \stat \enve \env {
      \inctxfmhc \fmhc {\inctx {\enve_j} {\inctxmhc
        \vmhc \enve \env}} \enve \env}}
  \\
  \inferrule{\enve_j = \ctx}
  { \statt \enve \env \ltspure j \statt \enve \env}
  \and
  \inferrule{\enve_j = \inctx \rctx {\reset \ctx}}
  { \statt \enve \env \ltsnotpure j \statt {\enve, \inctx F {\reset \mtctx},
      \reset \ctx} \env}
\end{mathpar}

\caption{LTS for the original semantics}%
\label{fig:lts-env-original}
\end{figure}

\subsection{Definitions for the Original Semantics}%
\label{ss:env-LTS-original}

In the original semantics, terms are evaluated within a top-level
reset. To follow that principle, the LTS for the original semantics is
defined only on \emph{pure terms}, i.e., terms without effects,
defined as follows.
\begin{grammar}
  & \textrm{Pure terms:} \quad & \prg & \bnfdef \val \bnfor \reset \tm
\end{grammar}
We remind that terms of the form $\reset \tm$ are called delimited; we extend
this notion to contexts as well. A \emph{pure state} is of the form
$\stat \enve \env \prg$, and the LTS operates either on pure or environment
states. The problem is then how to build pure states out of terms that are not
pure. A simple idea would be to relate two impure terms $\tmzero$ and $\tmone$
by comparing $\reset\tmzero$ and~$\reset\tmone$. However, such a solution would
not be sound, as it would relate $\shift \vark {\app \vark y}$ and
$\shift \vark {\app {\lamp z z} y}$, terms that can be distinguished by the
context $\reset {\app \mtctx \Omega}$.

Instead, the transitions $\ltslamo i \vmhc \fmhc$ and $\ltsctxxo j \vmhc \fmhc$
of the LTS for the original semantics (Figure~\ref{fig:lts-env-original}) now
include an extra argument~$\fmhc$ to build pure terms in their resulting
state. Recall that we use any evaluation context and not a delimited context
$\reset \emhc$ as it is possible to build a context of that shape from a context
$\appcont{\holecont_i} \fmhc$ assuming $\enve_i$ contains an enclosing
\textreset. Besides, we can discard the non-interesting parts of a testing
context $\fmhc$ thanks to bisimulation up to context. The other main difference
between the LTS for the original and relaxed semantics is the lack of rule for
testing control-stuck terms, as pure terms cannot become stuck (see
Proposition~\ref{p:eval-reset}).

As in the relaxed semantics, the transitions $\ltsval$ and $\ltsctxxo j \vmhc
\fmhc$ are passive and the others are active; we write $\pprogressp$ for the
notion of progress based on the LTS of Figure~\ref{fig:lts-env-original}.
\begin{defi}%
  \label{d:env-original}
  A relation $\rel$ on states is a pure environmental bisimulation if
  $\rel \pprogressp \rel, \rel$. Pure environmental
  bisimilarity~$\envbisimp$ is the largest pure environmental
  bisimulation.
\end{defi}
\noindent We define $\open \envbisimp$ on open terms as in the relaxed case, and
we extend $\envbisimp$ to any terms as follows: we have
$\tmzero \envbisimp \tmone$ if for all $\ctx$, we have
$\stat \emptyset \emptyset {\reset {\inctx \ctx \tmzero}} \envbisimp \stat
\emptyset \emptyset {\reset {\inctx \ctx \tmone}}$. As a simple example
illustrating the differences between $\envbisim$ and $\envbisimp$, we have the
following result.
\begin{prop}%
  \label{p:omega-stuck}
  We have $\Omega \envbisimp \shift \vark \Omega$.
\end{prop}
\noindent The relation
$\{(\stat \emptyset \emptyset {\reset {\inctx \ctx \Omega}}, \stat \emptyset
\emptyset {\reset {\inctx \ctx {\shift \vark \Omega}}}), (\stat \emptyset
\emptyset {\reset {\inctx \ctx \Omega}}, \stat \emptyset \emptyset {\reset
  {\Omega}})\}$ is a pure
bisimulation. Proposition~\ref{p:omega-stuck} does not hold with $\envbisim$
because $\Omega$ is not stuck.

\subsection{Up-to Techniques}%
\label{ss:env-upto}

\begin{figure}

  \textbf{Techniques for both semantics}
  \begin{mathpar}
    \inferrule{\tmzero \rtclo\redcbv \tmzero' \\ \tmone \rtclo\redcbv \tmone' \\ \stat \enve
      \env {\tmzero'} \rel \stat \envef \envd {\tmone'} }
    {\stat \enve \env \tmzero \utred\rel  \stat \envef \envd \tmone}
    \and
    \inferrule{\stat {\enve, \seq \rctx}{\env, \seq \val} \tmzero \rel \stat
      {\envef, \seq {\rctx'}}{\envd, \seq w} \tmone} {\stat \enve \env \tmzero
      \weak\rel \stat \envef \envd \tmone}
  \end{mathpar}

  \vspace{1em}
  \textbf{Techniques specific to the relaxed semantics}
  \begin{mathpar}
    \inferrule{\statt \enve \env \rel \statt \envef \envd}
    {\stat {\enve, \inctxmhc \fmhcs \enve \env}{\env, \inctxmhc{\mhcvs} \enve
        \env}{\inctxmhc \mhc \enve \env} \utrctxv\rel
      \stat {\envef, \inctxmhc \fmhcs \envef \envd}{\envd,\inctxmhc \mhcvs
        \envef \envd}{\inctxmhc \mhc \envef \envd}}
    \and
    \inferrule{\stat \enve \env \tmzero \rel \stat \envef \envd \tmone}
    {\stat {\enve, \inctxmhc \fmhcs \enve \env}{\env, \inctxmhc \mhcvs \enve
        \env}
      {\inctxfmhc \fmhc \tmzero \enve \env} \utrctx\rel
      \stat {\envef, \inctxmhc \fmhcs \envef \envd}{\envd, \inctxmhc \mhcvs
        \envef \envd} {\inctxfmhc \fmhc \tmone \envef \envd}}
  \end{mathpar}

  \vspace{1em}
  \textbf{Techniques specific to the original semantics}
  \begin{mathpar}
    \inferrule{\statt \enve \env \rel \statt \envef \envd \\ \inctxmhc \mhc \enve \env
        \mbox{ and } \inctxmhc \mhc \envef \envd \mbox{ are delimited} }
    {\stat {\enve, \inctxmhc \fmhcs \enve \env}{\env, \inctxmhc{\mhcvs} \enve
        \env}{\inctxmhc \mhc \enve \env} \utrctxv\rel
      \stat {\envef, \inctxmhc \fmhcs \envef \envd}{\envd,\inctxmhc \mhcvs
        \envef \envd}{\inctxmhc \mhc \envef \envd}}
    \and
    \inferrule{\stat \enve \env \prgzero \rel \stat \envef \envd \prgone
      \\ \inctxmhc \fmhc \enve \env
        \mbox{ and } \inctxmhc \fmhc \envef \envd \mbox{ are delimited}}
    {\stat {\enve, \inctxmhc \fmhcs \enve \env}{\env, \inctxmhc \mhcvs \enve
        \env}
      {\inctxfmhc \fmhc \prgzero \enve \env} \utrctx\rel
      \stat {\envef, \inctxmhc \fmhcs \envef \envd}{\envd, \inctxmhc \mhcvs
        \envef \envd} {\inctxfmhc \fmhc \prgone \envef \envd}}
  \end{mathpar}

  \caption{Up-to techniques in both semantics}%
  \label{fig:env-upto}
\end{figure}

Environmental bisimulation is meant to be used with up-to techniques, as doing
bisimulation proofs with Definition~\ref{d:env-relaxed} or~\ref{d:env-original}
alone is tedious in practice. Figure~\ref{fig:env-upto} lists the up-to
techniques we use for environmental bisimilarity in the two
semantics. Bisimulation up to reduction $\rawutred$ relates terms after some
reduction steps, thus allowing a big-step reasoning even with a small step
bisimulation. Bisimulation up to weakening $\rawweak$, also called bisimulation
up to environment in previous works~\cite{Sangiorgi-al:TOPLAS11}, removes values
and contexts from a state, thus diminishing its testing power, since less values
and contexts means less arguments to build from with multi-hole contexts.

We define two kinds of bisimulations up to related contexts, depending
whether we operate on environment states ($\rawutrctxv$) or on term or
pure states ($\rawutrctx$). As explained before, only evaluation
contexts are allowed for term and pure states, while any context is
valid with environment states. These up-to techniques differ from the
usual bisimulation up to context in the syntax of the multi-hole
contexts, which may include context holes $\holecont_i$. The
definitions for the original semantics differ from the ones for the
relaxed semantics in that only pure terms can be built in the case of
the original semantics.

The definitions of $\rawutrctxv$ and $\rawutrctx$ also allow the sequences of
values and contexts to be extended. This operation opposite to weakening, known
as \emph{strengthening}~\cite{Madiot-al:CONCUR14}, does not change the testing
power of the states, since the added values and contexts are built out of the
existing ones. We inline strengthening in the definitions of bisimulation up to
related contexts for technical reason: a separate notion of bisimulation up to
strengthening would be a regular up-to technique (not strong), like bisimulation
up to related contexts, which entails that these up-to techniques could not be
composed after a passive transition in a respectfulness proof.

The functions we define are indeed up-to techniques, are they form a respectful
set in both semantics~\cite{Aristizabal-al:LMCS17}.
\begin{lem}%
  \label{l:compatible}
  $\setF \is \{ \rawutred, \rawweak, \rawutrctxv, \rawutrctx\}$ is
  diacritically respectful, with $\strong \setF= \{ \rawutred, \rawweak\}$.
\end{lem}
\noindent This lemma and the third property of
Proposition~\ref{p:properties-compatibility-better} directly imply that
$\envbisim$ and $\envbisimp$ are compatible, from which we can deduce that they
are sound {\wrt}respectively $\ctxequiv$ and $\ctxequivp$. We can also prove that
they are complete~\cite{Aristizabal-al:LMCS17}.

\begin{thm}%
  \label{t:env-carac}
  $\tmzero \ctxequiv \tmone$ iff
  $\stat \emptyset \emptyset \tmzero \envbisim \stat \emptyset \emptyset \tmone$, and
  $\tmzero \ctxequivp \tmone$ iff
  $\stat \emptyset \emptyset \tmzero \envbisimp \stat \emptyset \emptyset \tmone$.
\end{thm}

\begin{proof}[Sketch]
  Given two terms $\tmzero$, $\tmone$, we write $\tmzero \eqObs \tmone$ if
  \begin{itemize}
  \item $\tmzero \evalcbv \valzero$ for some $\valzero$ iff
    $\tmone \evalcbv \valone$ for some $\valone$, and
  \item $\tmzero \evalcbv \tmzero'$
    for some control-stuck term $\tmzero'$ iff $\tmzero \evalcbv \tmone'$ for
    some control-stuck term $\tmone'$.
  \end{itemize}
  For the relaxed semantics, we show that
  \begin{multline*}
    \rel \mathord{\is} \{ (\stat \enve \env \tmzero, \stat \envef \envd \tmone)
    \mid \forall \fmhc, \inctxfmhc \fmhc \tmzero \enve \env \eqObs \inctxfmhc
    \fmhc \tmone
    \envef \envd \} \\
    \cup \{ (\statt \enve \env, \statt \envef \envd) \mid \forall \mhc,
    \inctxmhc \mhc \enve \env \eqObs \inctxmhc \mhc \envef \envd \}
   \end{multline*}
   is an environmental bisimulation. We define a similar candidate relation or
   the original semantics with the extra requirement that the contexts $\fmhc$
   and $\mhc$ should be delimited. The proof is then by case analysis on the
   possible transitions. For example, take
   $\statt \enve \env \rel \statt \envef \envd$ such that
   $\statt \enve \env \ltslam i \vmhc \stat \enve \env \tmzero$. Then
   $\statt \envef \envd \ltslam i \vmhc \stat \envef \envd \tmone$ for some
   $\tmone$. Let $\fmhc$ such that $\inctxfmhc \fmhc \tmzero \enve \env$
   evaluates to a value or a control-stuck term. Consider
   $\mhc = \inctx \fmhc {\app {\mtctx_i} \vmhc}$; then
   $\inctxmhc \mhc \enve \env \redcbv \inctxfmhc \fmhc \tmzero \enve \env$ and
   $\inctxmhc \mhc \envef \envd \redcbv \inctxfmhc \fmhc \tmone \envef
   \envd$. Because
   $\inctxmhc \mhc \enve \env \eqObs \inctxmhc \mhc \envef \envd$, we deduce
   $\inctxfmhc \fmhc \tmzero \enve \env \eqObs \inctxfmhc \fmhc \tmone \envef
   \env$.
\end{proof}

\subsection{Examples}%
\label{ss:env-examples}

We illustrate the usefulness of bisimulation up to related contexts, first on
one of our running basic examples.

\begin{exa}[double reset]%
  \label{ex:reset-reset-env}
  The relation
  $\rel {\mathord{\is}} \{ \stat{\emptyset}{\emptyset}{\reset {\reset \tm}}, \stat{\emptyset}{\emptyset}{\reset \tm} \mid \tm \in \cterms \}$ is a bisimulation up to
  context. Indeed, we show in Example~\ref{ex:reset-reset-app} that either
  $\reset \tm \redcbv \reset{\tm'}$ for some $\tm'$ or
  $\reset \tm \redcbv \val$, and $\reset \tm \redcbv \reset{\tm'}$ iff
  $\reset{\reset \tm} \redcbv \reset{\reset {\tm'}}$ and
  $\reset \tm \redcbv \val$
  $\reset{\reset \tm} \redcbv \reset \val \redcbv \val$. After a $\lts\tau$
  step, we either stay in $\rel$, or we get identical terms, i.e., we are in
  $\utrctxv \rel$.
\end{exa}

\noindent We then show how it helps when checking the $\ltsstuck \fmhc$
transition.

\begin{exa}%
  \label{ex:stuck-utctx}
  Let $\sts \is \stat \enve \env \tmzero$ and
  $\stt \is \stat \envef \enve \tmone$ so that $\tmzero$ and $\tmone$ are stuck,
  and $\sts \rel \stt$. If
  $\sts \ltsstuck \fmhc \stat \enve \env {\inctxfmhc \fmhc \tmzero \enve \env}$
  and
  $\stt \ltsstuck \fmhc \stat \envef \envd {\inctxfmhc \fmhc \tmone \envef
    \envd}$, i.e., $\fmhc$ does not trigger the capture in $\tmzero$ and
  $\tmone$, then we can conclude directly since we have
  $\stat \enve \env {\inctxfmhc \fmhc \tmzero \enve \env} \utrctx\rel \stat
  \enve \env {\inctxfmhc \fmhc \tmone \envef \envd}$. Similarly, if
  $\fmhc = \inctx {\fmhc'}{\reset \emhc}$, then
  $\sts \ltsstuck \fmhc \stat \enve \env {\inctxfmhc {\fmhc'}{\reset{\tmzero'}}
    \enve \env}$ and
  $\stt \ltsstuck \fmhc \stat \envef \envd {\inctxfmhc {\fmhc'}
    {\reset{\tmone'}} \envef \envd}$ for some $\tmzero'$ and $\tmone'$, so
  $\rawutrctx$ allows us to forget about $\fmhc'$ and to focus on
  $\stat \enve \env {\tmzero'}$ and $\stat \envef \enve {\tmone'}$.

\end{exa}

\noindent The next example is specific to the original semantics and illustrates
the role of $\holecont_i$. It does not hold in the relaxed semantics, because
the term on the right is control-stuck, but the one on the left may not evaluate
to a control-stuck term if $\tmone$ does not terminate.

\begin{exa}%
  \label{ex:stail-env}
  If $k \notin \fv{\tmone}$, then
  $\stat \emptyset \emptyset {(\lam x {\shift k {\tmzero}}) \iapp \tmone}
  \envbisimp \stat \emptyset \emptyset {\shift k {((\lam x {\tmzero}) \iapp
      \tmone)}}$, as the relation
  \begin{align*}
    \rel \mathord{\is} & \{ (\stat \emptyset \emptyset {\reset {\inctx \ctx
                         {(\lam x {\shift k {\tmzero}}) \iapp \tmone}}},
                         \stat \emptyset \emptyset {\reset {\inctx \ctx {\shift
                         k {((\lam x {\tmzero}) \iapp \tmone)}}}}) \} \\
    \cup & \{ (\statt {\reset {\inctx E {(\lam x {\shift k {\tmzero}}) \iapp \mtctx}}}
           \emptyset, \statt {\reset{\lamp x {\subst {\tmzero} k {\lam y
           {\reset{\inctx E y}}}} \iapp
           \mtctx}} \emptyset)  \}
  \end{align*}
  is a bisimulation up to reduction and related contexts. We start by analyzing
  the behavior of the first pair $(\sts, \stt)$ in $\rel$. If $\tmone$ is a
  value $\valone$, then
  \begin{align*}
    \sts  \lts\tau\lts\tau & \stat \emptyset \emptyset {
    \reset {\subst {\subst \tmzero \varx \valone} \vark {\lam z {\reset {\inctx
              \ctx z}}}}} \mbox{ and} \\
    \stt \lts\tau\lts\tau & \stat \emptyset \emptyset
  {\reset {\subst {\subst \tmzero \vark {\lam z {\reset {\inctx
                \ctx z}}}} \varx \valone}},
  \end{align*}
  but because $\vark \notin \fv \valone$, the resulting states are in fact
  equal, and therefore in $\utrctxv\rel$. If $\tmone \redcbv \tmone'$, then
  \begin{align*}
    \sts  \lts \tau & \stat \emptyset \emptyset {\reset {\inctx \ctx
          {(\lam x {\shift k {\tmzero}}) \iapp \tmone'}}} \mbox{ and}\\
    \stt  \lts \tau \lts\tau & \stat \emptyset \emptyset {\reset{\lamp x
          {\subst {\tmzero} k {\lam y {\reset{\inctx E y}}}}} \iapp \tmone'};
  \end{align*}
  the resulting states are in $\utrctxv \rel$, by considering the common context
  $\appcont {\holecont_1}{\tmone'}$. If $\tmone$ is a control-stuck term
  $\inctx {\ctx'}{\shift {\vark'}{\tmone'}}$, then
  \begin{align*}
    \sts \lts\tau & \reset {\subst {\tmone'}{\vark'}{\lam z {\reset
                     {\inctx \ctx {(\lam x {\shift k {\tmzero}}) \iapp \inctx {\ctx'} z}}}
                     }} \mbox{ and} \\
    \stt \lts\tau\lts\tau & \reset {\subst {\tmone'}{\vark'}{\lam z
                              {\reset{\lamp x {\subst {\tmzero} k {\lam y {\reset{\inctx E y}}}} \iapp
                              \inctx {\ctx'} z}}} }.
  \end{align*}
  Again, the resulting states are in $\utrctxv \rel$, by considering the context
  $\reset {\subst {\tmone'}{\vark'}{\lam z {\reset{\appcont {\holecont_1}{\inctx
            {\ctx'} z}}}}}$. We have covered all the possible cases for
  $(\sts, \stt)$.

  For the second set, if
  \begin{mathpar}
  \statt \enve \emptyset \is \statt {\reset {\inctx E {(\lam x {\shift k
          {\tmzero}}) \iapp \mtctx}}} \emptyset
          \and\mbox{and}\and
  \statt \envef \emptyset \is \statt {\reset{\lam x {\subst {\tmzero} k {\lam y
          {\reset{\inctx E y}}}} \iapp \mtctx}} \emptyset,
  \end{mathpar}
    then
  \begin{align*}
    \statt \enve \emptyset \ltsctxxo 1 \vmhc \fmhc \lts\tau \lts\tau & \stat \enve \emptyset
                                                                       {\reset {\inctxfmhc \fmhc {\reset{\subst {\subst {\tmzero} x {\inctxmhc \vmhc \enve \emptyset}} k {\lam
                                                                       y {\reset {\inctx E y}}}}} \enve \emptyset}} \mbox{ and} \\
    \statt \envef \emptyset \ltsctxxo 1 \vmhc \fmhc \lts\tau & \stat \envef \emptyset
                                                              {\reset {\inctxfmhc \fmhc {\reset{\subst {\subst {\tmzero} x {\inctxmhc \vmhc \envef \emptyset}} k {\lam
                                                              y {\reset {\inctx E y}}}}} \envef \emptyset}} \\
  \end{align*}
  The resulting states are in $\utrctxv\rel$; note that we use the regular up-to
  technique $\rawutrctxv$ only after an active $\lts\tau$ transition, and not
  after the passive $\ltsctxxo 1 \vmhc \fmhc$ transition.
\end{exa}

\noindent The next example shows the limits of bisimulation up to
related contexts for environmental bisimilarity, as we have to define
an infinite candidate relation for a simple example. However, it is
still an improvement over the plain environmental
(Example~\ref{ex:fixed-point-env}) or applicative
(Example~\ref{ex:fixed-point-app}) proofs.

\begin{exa}[Turing's combinator]%
  \label{ex:fixed-point-env-utc}
  Let $\theta' \is \lam \varx {\reset{\app \theta \varx}}$,
  $\valzero \is \lam y {\app y {\lamp z {\app {\app {\app \theta \theta} y}
        z}}}$,
  $\valone \is \lam y {\app y {\lamp z {\app {\app {\app {\theta'}{\theta'}} y}
        z}}}$.  We define $\rel_0$ inductively as follows.
  \begin{mathpar}
    \inferrule{ }
    {\statt \emptyset \valzero \rel_0 \statt \emptyset \valone}
    \and
    \inferrule{\statt \emptyset \env \rel_0 \statt \emptyset \envd}
    {\statt \emptyset {\env, \lam z {\app {\app {\Turing}{\inctxmhc \vmhc
                              \emptyset \env}} z}} \rel_0 \statt \emptyset
                              {\valone, \lam z {\app {\app {\app
                              {\theta'}{\theta'}}{\inctxmhc \vmhc \emptyset \envd}} z}}}
  \end{mathpar}
  Then
  $ \rel {\mathord{\is}} \{ (\stat \emptyset \emptyset \Turing, \stat \emptyset
  \emptyset {\Turingshift}), (\stat \emptyset \env \Turing, \stat \emptyset
  \envd {\app {\theta'}{\theta'}}) \mid \statt \emptyset \env \rel_0 \statt
  \emptyset \envd \} \mathop\cup \rel_0 $ is a bisimulation up to related
  contexts. For the first
  two pairs, we have respectively
  $\stat \emptyset \emptyset \Turing \ltswkval \statt \emptyset \valzero$ and
  $\stat \emptyset \emptyset {\Turingshift} \ltswkval \statt \emptyset
  \valone$, and
  $\stat \emptyset \env \Turing \ltswkval \statt \emptyset {\env, \valzero}$ and
  $\stat \emptyset \envd {\theta' \iapp \theta'} \ltswkval \statt \emptyset
  {\envd, \valone}$ for some
  $\statt \emptyset \env \rel_0 \statt \emptyset \envd$, so we end up in
  $\rel_0$ in both cases.

  Let $\statt \emptyset \env \rel_0 \statt \emptyset \envd$, and suppose we want
  to test $\valzero$ and $\valone$, i.e.,
  \begin{align*}
    \statt \emptyset \env & \ltslam 1 \vmhc \stat \emptyset \env
                                {\inctxmhc \vmhc \emptyset \env \iapp {\lamp
                                z {\Turing \iapp \inctxmhc \vmhc
                                \emptyset \env \iapp z}}} \mbox{ and}\\
    \statt \emptyset \envd & \ltslam 1 \vmhc \stat \emptyset \envd
                                {\inctxmhc \vmhc \emptyset \envd \iapp {\lamp
                                z {\theta' \iapp \theta' \iapp \inctxmhc \vmhc
                                \emptyset \envd \iapp z}}}
  \end{align*}
  At that point, we would like to relate $\Turing$ and $\theta' \iapp
  \theta'$ and conclude using bisimulation up context, however these
  terms are not in an evaluation contexts in the above resulting
  states. Similarly, we cannot isolate $\Turing \iapp \mtctx$ and
  $(\theta' \iapp \theta') \iapp \mtctx$ using $\holecont$ and
  bisimulation up to related contexts, as these contexts are not
  evaluation contexts. Instead, $\rel_0$ has been defined so that the
  resulting states are in $\utrctxv {\rel_0}$.

  Finally, suppose $i > 1$ and
  $\env_i = \lam z {\app {\app {\Turing}{\inctxmhc \vmhc \emptyset {\env'}}} z}$,
  $\envd_i = \lam z {\app {\app {\theta' \iapp \theta'}{\inctxmhc \vmhc
        \emptyset {\envd'}}} z}$ for some~$\vmhc$ and
  $\statt \emptyset {\env'} \rel_0 \statt \emptyset {\envd'}$. Then
  $\statt \emptyset \env \ltslam i {\vmhc'} \stat \emptyset \env {\Turing \iapp
    \inctxmhc \vmhc \emptyset {\env'} \iapp \inctxmhc {\vmhc'} \emptyset \env}$
  and
  $\statt \emptyset \envd \ltslam i {\vmhc'} \stat \emptyset \envd {\theta'
    \iapp \theta' \iapp \inctxmhc \vmhc \emptyset {\envd'} \iapp \inctxmhc
    {\vmhc'} \emptyset \envd}$, and the resulting states are in $\utrctx\rel$.
\end{exa}

\subsection{Proving the Axioms}%
\label{ss:axioms-env}

We prove Kameyama and Hasegawa's axioms (Section~\ref{ss:cps-equivalence}) with
environmental bisimulation. The \AXbeta, \AXresetshift, and \AXresetval axioms
are direct consequences of a more general result.

\begin{prop}%
  \label{p:redcbv-in-envbisim}
  If $\tm \redcbv \tm'$, then $\stat \emptyset \emptyset \tm
  \envbisim \stat \emptyset \emptyset {\tm'}$.
\end{prop}

\begin{proof}
  It is easy to see that $\{ (\stat \emptyset \emptyset \tm,  \stat \emptyset
  \emptyset {\tm'}) \}$ is a bisimulation up to context: after a transition, we
  get identical terms.
\end{proof}

\begin{prop}[\AXetav axiom]
  If $\varx \notin \fv \val$, then
  $\stat \emptyset \emptyset {\lam \varx {\app \val \varx}} \envbisim \stat
  \emptyset \emptyset \val$.
\end{prop}

\begin{proof}
  The relation
  $\rel {\mathord{\is}} \{(\statt \emptyset {\lam \varx {\app \val \varx}}, \statt \emptyset
  \val) \}$ is a bisimulation up to context up to reduction. If
  $\val = \lam{y}{\tm}$, then the $\ltslam 1 \vmhc$ transition produces
  $\stat \emptyset {\lam \varx {\app \val \varx}}{ \app \val {\inctxmhc \vmhc
      \emptyset {\lam \varx {\app \val \varx}}}}$ and
  $\stat \emptyset \val {\subst \tm y {\inctxmhc \vmhc \emptyset \val}}$. Then
  $\stat \emptyset {\lam \varx {\app \val \varx}}{ \app \val {\inctxmhc \vmhc
      \emptyset {\lam \varx {\app \val \varx}}}} \lts\tau \stat \emptyset {\lam
    \varx {\app \val \varx}} {\subst \tm y {\inctxmhc \vmhc \emptyset {\lam
        \varx {\app \val \varx}}}} \utrctxv\rel \stat \emptyset \val {\subst \tm
    y {\inctxmhc \vmhc \emptyset \val}}$, as wished.
\end{proof}

\begin{prop}[\AXshiftreset axiom]
  We have $\stat \emptyset \emptyset {\shift \vark {\reset \tm}} \envbisim \stat
  \emptyset \emptyset {\shift \vark \tm}$.
\end{prop}

\begin{proof}
  The relation
  $\rel {\mathord{\is}}\{(\stat \emptyset \emptyset {\shift \vark {\reset \tm}},
  \stat \emptyset \emptyset {\shift \vark \tm}), (\stat \emptyset \emptyset
  {\reset{\reset \tm}}, \stat \emptyset \emptyset {\reset \tm}) \mmid \tm \in
  \cterms \}$ is a bisimulation up to context. Indeed, if
  $\fmhc = \inctx \rctx {\reset \ctx}$ (the other case being trivial), then
  $\stat \emptyset \emptyset {\shift \vark {\reset \tm}} \ltsstuck \fmhc \stat
  \emptyset \emptyset {\inctx \rctx {\reset{\reset{\subst \tm \vark {\lam \varx
            {\reset{\inctx \ctx \varx}}}}}}}$ and
  $\stat \emptyset \emptyset {\shift \vark \tm} \ltsstuck \fmhc \stat \emptyset
  \emptyset {\inctx \rctx {\reset{\subst \tm \vark {\lam \varx {\reset{\inctx
              \ctx \varx}}}}}}$. We obtain states of the form
  $\stat \emptyset \emptyset {\inctx \rctx {\reset{\reset {\tm'}}}}$ and
  $\stat \emptyset \emptyset {\inctx \rctx {\reset{\tm'}}}$, which are in
  $\utrctx\rel$. Example~\ref{ex:reset-reset-env} concludes in the case of the pair
  $(\stat \emptyset \emptyset {\reset{\reset \tm}}, \stat \emptyset \emptyset
  {\reset \tm})$.
\end{proof}

\begin{prop}[\AXresetlift axiom]%
  \label{p:reset-beta-env}
  We have
  $\stat \emptyset \emptyset {\reset{\app {\lamp \varx \tmzero}{\reset \tmone}}}
  \envbisim \stat \emptyset \emptyset {\app{\lamp \varx {\reset \tmzero}}{\reset
      \tmone}}$.
\end{prop}

\begin{proof}
  The relation
  $ \{ (\stat \emptyset \emptyset {\reset{\app {\lamp \varx \tm}{\reset
        {\tm'}}}}, \stat \emptyset \emptyset {\app{\lamp \varx {\reset
        \tm}}{\reset {\tm'}}}) \mmid \tm' \in \cterms \}$ is a bisimulation up
  to context, with the same reasoning as in Example~\ref{ex:reset-reset-env}.
\end{proof}

\begin{prop}[\AXbetaomega axiom]%
  \label{p:omega-env}
  If $\varx \notin \fv \ctx$, then $\stat \emptyset \emptyset {\app {\lamp \varx {\inctx \ctx
      \varx}} \tm} \envbisim \stat \emptyset \emptyset {\inctx \ctx \tm}$.
\end{prop}

\begin{proof}
  Define $\enve \is (\lamp \varx {\inctx \ctx \varx} \iapp
  \mtctx)$, $\envef \is (E)$, and
  \[ \rel \mathord{\is} \{ (\statt \enve \emptyset, \statt \envef \emptyset),
    (\stat \enve \emptyset {\app {\lamp \varx {\inctx \ctx \varx}}{\inctxmhc
        \vmhc \enve \emptyset}}, \stat \envef \emptyset {\inctx \ctx {\inctxmhc
        \vmhc \enve \emptyset}} \}.\]

  \noindent
  Then
  $\stat \emptyset \emptyset {\app{(\lam x {\inctx E x})} \tm}
  \weak{\utrctxv\rel} \stat \emptyset \emptyset {\inctx E \tm}$ and $\rel$ is a
  bisimulation up to context, since the sequence
  $\statt \enve \emptyset \ltsctxx 1 \vmhc \stat \enve \emptyset {\app{(\lam x
      {\inctx E x})}{\inctxmhc \vmhc \enve \emptyset}} \lts\tau \stat \enve
  \emptyset {\inctx E {\inctxmhc \vmhc \enve \emptyset}}$ fits
  $\statt \envef \emptyset \ltsctxx 1 \vmhc \stat \envef \emptyset {\inctx E
    {\inctxmhc \vmhc \envef \emptyset}} \ltswk\tau \stat \envef \emptyset
  {\inctx E {\inctxmhc \vmhc \envef \emptyset}}$, where the final states are in
  $\utrctxv \rel$. Notice we use~$\rawutrctxv$ after $\lts\tau$, and not after the
  passive $\ltsctxx 1 \vmhc$ transition. The transition $\ltspure 1$ is easy to
  check.
\end{proof}

\begin{prop}%
  \label{p:skkt-env}
  If $k \notin \fv \tm$, then $\stat \emptyset \emptyset {\shift k {\app k
    \tm}} \envbisimp \stat \emptyset \emptyset \tm$.
\end{prop}

\begin{proof}
  Let
  $\enve \is (\reset{\lam x {\reset {\inctx E x}} \iapp \mtctx}, \reset{\reset
    \mtctx})$, $\envef \is (\reset E, \reset \mtctx)$, and
  \[ \rel \mathord{\is} \{ (\stat \emptyset \emptyset {\reset {\inctx \ctx
        {\shift k {\app k \tm}}}}, \stat \emptyset \emptyset {\reset {\inctx
        \ctx \tm}}), (\statt \enve \emptyset, \statt \envef \emptyset) \mmid x
    \notin \fv E \}.\]
  For the first pair, we have $\stat \emptyset \emptyset
  {\reset {\inctx \ctx {\shift k {\app k \tm}}}} \lts\tau \stat \emptyset
  \emptyset {\reset {\lamp x {\reset {\inctx \ctx \varx}}} \iapp \tm}$ and $\stat
  \emptyset \emptyset {\reset {\inctx \ctx \tm}} \ltswk\tau \stat \emptyset
  \emptyset {\reset {\inctx \ctx \tm}}$ so that the resulting states are in
  $\utrctx \rel$, by considering the context $\appcont {\holecont_1} \tm$.

  Otherwise, the sequence
  $\statt \enve \emptyset \ltsctxxo 1 \vmhc \fmhc \lts\tau \stat \enve \emptyset
  {\reset {\inctxfmhc \fmhc {\reset{\reset{\inctx E {\inctxmhc \vmhc \enve
              \emptyset}}}} \enve \emptyset}}$ is matched by
  $\statt \enve \emptyset \ltswkctxxo 1 \vmhc \fmhc \stat \enve \emptyset
  {\reset {\inctxfmhc \fmhc {\reset{\inctx E {\inctxmhc \vmhc \envef
            \emptyset}}} \envef \emptyset}})$, since the resulting states are in
  $\utrctxv\rel$, and we use up to related contexts after a $\lts\tau$
  transition. Finally,
  $\statt \enve \emptyset \ltsctxxo 2 \vmhc \fmhc \lts\tau \lts\tau \stat \enve
  \emptyset {\reset {\inctxfmhc \fmhc {\inctxmhc \vmhc \enve \emptyset} \enve
      \emptyset}}$ is matched by
  $\statt \envef \emptyset \ltsctxxo 2 \vmhc \fmhc \lts\tau \stat \envef \emptyset
  {\reset {\inctxfmhc \fmhc {\inctxmhc \vmhc \envef \emptyset} \envef
      \emptyset}}$, and the context splitting transitions $\ltsnotpure i$ are
  easy to check for $i \in \{ 1, 2\}$.
\end{proof}

The bisimilarity $\envbisimp$ verifies all the axioms of $\cpsequiv$, it is
therefore complete {\wrt}this relation.

\begin{cor}
  We have $\mathord{\cpsequiv} \subseteq \mathord{\onfbisim}$.
\end{cor}
\noindent As a result, we can use $\cpsequiv$ as a proof technique
for~$\envbisimp$ (and, therefore, for $\ctxequivp$). For instance,
Example~\ref{ex:stail-env} holds directly as it can be derived from the
axioms~\cite{Kameyama-Hasegawa:ICFP03}.

\subsection{Conclusion}

We define environmental bisimilarities that are sound and complete in the
relaxed and original semantics. Plain environmental bisimulation is harder to
use than applicative bisimulation, but it is supposed to be used in conjunction
with up-to techniques. In particular, bisimulation up to related contexts, which
allows to forget about a common context built out of values and evaluation
contexts in the environment, is what makes the proof technique tractable enough
to prove the \AXbetaomega axiom, an axiom which can hardly be proved with
applicative bisimilarity (see Example~\ref{p:omega-app}). However, some
equivalence proofs seem to be still unnecessary complex, as witnessed by
Example~\ref{ex:fixed-point-env-utc}.

Another issue is that the definition of $\envbisimp$ is only a small improvement
over the definition of $\ctxequivp$, as it contains quantifications over
evaluation contexts, either when extending the definition from any terms to pure
terms, or in the transitions $\ltslamo i \vmhc \fmhc$ and
$\ltsctxxo j \vmhc \fmhc$. In practice, these contexts are not too problematic
as many of them can be abstracted away in equivalence proofs thanks to up-to
techniques (see Example~\ref{ex:stail-env} or Proposition~\ref{p:skkt-env}), but
we wonder if it is possible to still have a complete bisimilarity and quantify
over less contexts or to restrict the class of terms on which such a
quantification over contexts is necessary.

% Local variables:
% mode: latex
% TeX-master: "journal.tex"
% End:

\section{Normal-Form Bisimilarity}%
\label{s:nf}
Normal-form bisimilarity~\cite{Lassen:LICS05} (originally defined in~\cite{Sangiorgi:LICS92}, where it was called \emph{open bisimilarity}) equates
(open) terms by reducing them to normal form, and then requiring the sub-terms
of these normal forms to be bisimilar. Unlike applicative and environmental
bisimilarities, normal-form bisimilarity usually does not contain a universal
quantification over testing terms or contexts in its definition, and is
therefore easier to use than the former two. However, it is also usually not
complete w.r.t.\ contextual equivalence, meaning that there exist contextually
equivalent terms that are not normal-form bisimilar.

A notion of normal-form bisimulation has been defined in various calculi,
including the pure $\lambda$-calculus~\cite{Lassen:MFPS99,Lassen:LICS05}, the
$\lambda$-calculus with ambiguous choice~\cite{Lassen:MFPS05}, the
$\lambda\mu$-calculus~\cite{Lassen:LICS06}, and the $\lambda\mu\rho$-calculus~\cite{Stoevring-Lassen:POPL07}, a calculus with control and store, where
normal-form bisimilarity characterizes contextual equivalence. It has also been
defined for typed languages~\cite{Lassen-Levy:CSL07,Lassen-Levy:LICS08}. In a
recent work~\cite{Biernacki-al:MFPS17}, we recast normal-form bisimulation in
the framework of diacritical progress (Section~\ref{ss:diacritical}), to be able
to define up-to techniques which respect $\eta$-expansion; we refer to this work
for more details.

In Section~\ref{ss:nf-def}, we propose a first definition of normal-form
bisimilarity for the relaxed semantics, for which we define up-to techniques in
Section~\ref{ss:upto-nf}. We then refine the definition in
Section~\ref{ss:refined}, to relate more contextually equivalent terms. We turn
to the original semantics in Section~\ref{ss:nf-original}, and we prove the
axioms in Section~\ref{ss:axioms-nf}. The material of Sections~\ref{ss:nf-def}
and~\ref{ss:upto-nf} comes from~\cite{Biernacki-al:MFPS17}, where the proofs can
be found, and supersedes~\cite{Biernacki-Lenglet:FLOPS12}. Refined bisimilarity,
originally defined in~\cite{Biernacki-Lenglet:FLOPS12}, is adapted to the
framework of~\cite{Biernacki-al:MFPS17} in the present article.  Normal-form
bisimilarity for the original semantics is also a contribution of this
article. The proofs for Sections~\ref{ss:refined} and~\ref{ss:nf-original} can be
found in the appendix.

\subsection{Definition}%
\label{ss:nf-def}

The main idea behind the definition of normal-form bisimilarity is that two
terms $\tmzero$ and $\tmone$ are bisimilar if their evaluations lead to matching
normal forms (\eg if~$\tmzero$ evaluates to a control stuck term, then so does
$\tmone$) with bisimilar sub-components. In the
$\lambda$-calculus~\cite{Sangiorgi:LICS92,Lassen:LICS05}, the possible normal
forms are only values and open stuck terms. In the relaxed semantics
of~$\lamshift$, we need to relate also control-stuck terms; we propose here a
first way to deal with these terms, that will be refined in a later
subsection. Deconstructing normal forms leads to comparing contexts as well as
terms. Given a relation $\rel$ on terms, we define in Figure~\ref{f:ext-rnf}
the extensions of $\rel$ to respectively values $\nfv\rel$, other normal forms
$\nf\rel$, and contexts $\nfc\rel$.

\begin{figure}

\begin{mathpar}
  \inferrule{\inctx {\ctx_0} \varx \rel \inctx {\ctx_1} \varx
    \\ \varx \mbox{ fresh}}
  {\ctx_0 \nfc\rel \ctx_1} \and
  \inferrule{\reset{\inctx {\ctx_0} \varx} \rel \reset{\inctx {\ctx_1}
      \varx}
    \\ \inctx {\rctx_0} \varx \rel
    \inctx {\rctx_1} \varx \\ \varx \mbox{ fresh}}
  {\inctx {\rctx_0} {\reset{\ctx_0}}  \nfc\rel \inctx{\rctx_1}{\reset{\ctx_1}}}
  \and
  \inferrule{ \valzero \iapp \varx \rel \valone \iapp \varx \\ \varx \textrm{
      fresh}}
  {\valzero \nfv\rel \valone}
  \and
  \inferrule{\ctxzero \nfc\rel \ctxone \\ \reset \tmzero \rel \reset \tmone}
  {\inctx \ctxzero {\shift \vark \tmzero} \nf\rel \inctx \ctxone {\shift \vark
      \tmone}}
  \and
  \inferrule{\rctxzero \nfc\rel \rctxone \\ \valzero \nfv\rel \valone}
  {\inctx \rctxzero {\app \varx \valzero} \nf\rel \inctx \rctxone {\app \varx
      \valone}}
\end{mathpar}

\caption{Extension of a relation to contexts and normal forms}%
\label{f:ext-rnf}

\end{figure}

The relation $\nfv\rel$ treats uniformly the different kinds of values
by applying them to a fresh variable $\varx$. As originally pointed
out by Lassen~\cite{Lassen:LICS05}, this is necessary for the
bisimilarity to be sound w.r.t. $\eta$-expansion; otherwise it would
distinguish $\eta$-equivalent terms such as $\lam y {\app \varx y}$
and $\varx$. However, unlike Lassen, we do not use a special
application operator to get rid of administrative $\beta$-redexes when
possible, as it is not necessary in our framework. The definition of
$\nfv\rel$ also easily scales to some other kinds of values: for
example, we consider \textshift as a value in a previous
work~\cite{Biernacki-al:MFPS17} with the same definition.

A control-stuck term $\inctx \ctxzero {\shift \vark {\tmzero}}$ can be executed
if it is plugged into a pure evaluation context surrounded by a \textreset; by
doing so, we obtain a term of the form
$\reset{\subst {\tmzero} \vark {\lam \varx {\reset {\inctx {\ctxzero'}
        \varx}}}}$ for some context $\ctxzero'$. The resulting term is within a
\textreset; similarly, when~$\nf\rel$ compares
$\inctx \ctxzero {\shift \vark {\tmzero}}$ and
$\inctx \ctxone {\shift \vark {\tmone}}$, it relates the \textshift bodies
$\tmzero$ and $\tmone$ within an enclosing \textreset. The pure contexts
$\ctxzero$ and $\ctxone$ are also tested by simply plugging a fresh variable
into them. Comparing $\tmzero'$ and $\tmone'$ without a surrounding \textreset
would be too discriminating, as it would distinguish equivalent terms such as
$\shift \vark {\reset \tm}$ and $\shift \vark \tm$ (axiom
\AXshiftreset). Without \resetId{}, we would have to relate $\reset \tm$ and
$\tm$, which are not equivalent in general (take $\tm=\shift {\vark'} \val$ for
some $\val$), while our definition requires $\reset {\reset \tm}$ and
$\reset \tm$ to be related (which holds for all $\tm$; see
Example~\ref{ex:reset-reset-nf}).

Two open stuck terms $\inctx \rctxzero {\app \varx \valzero}$ and
$\inctx \rctxone {\app \varx \valone}$ are related by $\nf\rel$ if the values
$\valzero$ and~$\valone$ as well as the contexts $\rctxzero$ and $\rctxone$ are
related. We have to be careful when defining bisimilarity on (possibly non pure)
evaluation contexts. We cannot simply compare $\rctxzero$ and~$\rctxone$ by
executing $\inctx \rctxzero y$ and $\inctx \rctxone y$ for a fresh $y$. Such a
definition would equate the contexts~$\mtctx$ and $\reset \mtctx$, which in turn
would relate the terms $\app \varx \val$ and $\reset {\app \varx \val}$, which
are distinguished by the context
$\app {\lamp \varx \mtctx}{\lam y {\shift \vark \Omega}}$. A context containing
a \resetId{} enclosing the hole should be related only to contexts with the same
property. However, we do not want to precisely count the number of delimiters
around the hole; doing so would distinguish $\reset \mtctx$ and
$\reset {\reset \mtctx}$, and, therefore, it would discriminate the contextually
equivalent terms $\reset{\app \varx \val}$ and
$\reset{\reset {\app \varx \val}}$. Hence, we check with $\nfc\rel$
(Figure~\ref{f:ext-rnf}) that if one of the contexts contains a
\resetId{} surrounding the hole, then so does the other; then we compare the
contexts beyond the first enclosing delimiter by simply evaluating them using a
fresh variable.
As a result, it rightfully distinguishes $\mtctx$ and
$\reset \mtctx$, but it relates $\reset \mtctx$ and $\reset {\reset \mtctx}$.

With these auxiliary relations, we define normal-form bisimilarity
using the notion of diacritical progress of
Section~\ref{ss:diacritical}. However, here we do not introduce an
underlying LTS, but instead we refer directly to the reduction
semantics of the calculus, which we find advantageous when working
with open terms. One can check that the notion of progress defined
below satisfies the conditions mentioned in Remark~\ref{r:no-lts}, and
therefore we can still rely on the theory presented in
Section~\ref{ss:diacritical}.

\begin{defi}%
  \label{d:nf-bisim}
  A relation $\rel$ on open terms diacritically progresses to $\rels$, $\relt$
  written $\rel \pprogress \rels, \relt$, if $\rel \mathop\subseteq \rels$,
  $\rels \mathop\subseteq \relt$, and $\tmzero \rel \tmone$ implies:
  \begin{itemize}
  \item if $\tmzero \redcbv \tmzero'$, then there exists $\tmone'$ such that
    $\tmone \clocbv \tmone'$ and $\tmzero' \relt \tmone'$;
  \item if $\tmzero$ is a value, then there exists $\valone$ such that
    $\tmone \evalcbv \valone$, and $\tmzero \nfv\rels \valone$;
  \item if $\tmzero$ is a normal form but not a value, then there exist
    $\tmone'$ such that $\tmone \evalcbv \tmone'$ and
    $\tmzero \nf\relt \tmone'$;
  \item the converse of the above conditions on $\tmone$.
  \end{itemize}
  A normal-form bisimulation is a relation $\rel$ such that
  $\rel \pprogress \rel, \rel$. Normal-form bisimilarity~$\nfbisim$ is the
  largest normal-form bisimulation.
\end{defi}

Testing values is passive, as we want to prevent the use of
bisimulation up to context in that case (if $\valzero \rel \valone$,
then $\valzero \iapp \varx$ and $\valone \iapp \varx$ are
automatically in a bisimulation up to context). The remaining clauses
of the bisimulation are active.

We show how to prove equivalences with normal-form bisimulation with
our running examples.

\begin{exa}[double \textreset]%
  \label{ex:reset-reset-nf}
  We prove that $\reset \tm \nfbisim \reset{\reset \tm}$ by showing that
  $\rel \mathord{\is} \{ (\reset \tm, \reset{\reset \tm}) \mid \tm \in \terms \}
  \mathop\cup \nfbisim$ is a normal-form bisimulation. On top of reduction
  steps, for which we can conclude as in Example~\ref{ex:reset-reset-app}, we
  have to consider the case $\reset \tm = \inctx \rctx {\app \varx \val}$. Then
  by Proposition~\ref{p:eval-reset}, there exists $\rctx'$ such that
  $\rctx = \reset {\rctx'}$. Therefore, we have
  $\reset {\reset \tm} = \reset{\reset{\inctx {\rctx'}{\app \varx
        \val}}}$. We have $\val \nfv\nfbisim \val$, and we have to prove that
  $\reset {\rctx'} \nfc\rel \reset{\reset {\rctx'}}$ holds to conclude. If
  $\rctx'$ is a pure context $\ctx$, then we have to prove
  $\reset{\inctx \ctx y} \rel \reset{\inctx \ctx y}$ and $y \rel \reset y$ for a
  fresh $y$, which are both true because $\nfbisim \mathop{\subseteq} \rel$. If
  $\rctx' = \inctx {\rctx''}{\reset {\ctx}}$, then given a fresh $y$, we have to
  prove $\reset {\inctx {\rctx''} y} \rel \reset{\reset {\inctx {\rctx''} y}}$
  (clear by the definition of $\rel$), and
  $\reset{\inctx \ctx y} \rel \reset{\inctx \ctx y}$ (true because
  $\nfbisim \mathop{\subseteq} \rel$).

  Similarly, if $\reset{\reset \tm} = \inctx \rctx {\app \varx \val}$,
  then we can show that there exists $\rctx'$ such that
  $\rctx = \reset{\reset {\rctx'}}$ and
  $\reset \tm \evalcbv \reset{\inctx {\rctx'}{\app \varx \val}}$, and we can
  conclude as in the previous case. As we can see, the proof is longer that with
  applicative (Example~\ref{ex:reset-reset-app}) or environmental
  (Example~\ref{ex:reset-reset-env}) bisimilarities, just because we have to
  consider open-stuck terms.
\end{exa}

\begin{exa}[Turing's combinator]%
  \label{ex:fixed-point}
  We prove that Turing's combinator $\Turing$ is bisimilar to its variant
  $\Turingshift$ by building the candidate relation~$\rel$
  incrementally, starting from $(\Turing, \Turingshift)$. Evaluating these
  two terms, we obtain
  \begin{align*}
    \Turing & \evalcbv \lam y {\app y {\lamp z {\app {\app {\app \theta
              \theta} y} z}}} \is \valzero \mbox{, and} \\
    \Turingshift & \evalcbv \lam y {\app y {\lamp z {\app {\app {\app {\lamp
                \varx {\reset {\app \theta \varx}}}{\lamp \varx {\reset {\app
                    \theta \varx}}}} y} z}}} \is \valone.
  \end{align*}
  Evaluating $(\app \valzero y, \app \valone y)$ for a fresh $y$, we obtain two
  open-stuck terms, so we add their decomposition to $\rel$. Let
  $\valzero' \is \lam z {\app {\app {\app \theta \theta} y} z}$ and
  $\valone' \is \lam z {\app {\app {\app {\lamp \varx {\reset {\app \theta
              \varx}}}{\lamp \varx {\reset {\app \theta \varx}}}} y} z}$; then
  we add $(\app {\valzero'} z, \app {\valone'} z)$ and $(z, z)$ for a
  fresh~$z$ to $\rel$. Evaluating $\app{\valzero'} z$ and
  $\app{\valone'} z$, we obtain respectively $\app {\app y {\valzero'}} z$
  and $\app {\app y {\valone'}} z$; to relate these two open stuck terms, we
  just need to add $(\app \varx z, \app \varx z)$ (for a fresh $\varx$) to
  $\rel$, since we already have $\valzero' \nfv\rel \valone '$. The constructed
  relation $\rel$ we obtain is a normal-form bisimulation.

  As we can see, the proof is much simpler than with applicative
  (Example~\ref{ex:fixed-point-app}) or environmental
  (Example~\ref{ex:fixed-point-env-utc}) bisimulations, even using the plain
  definition of normal-form bisimulation. We can further simplify the definition
  of the candidate relation thanks to up-to techniques.
\end{exa}

\subsection{Up-to Techniques, soundness, and completeness}%
\label{ss:upto-nf} The already quite tractable equivalence proofs based on
normal-form bisimulation can be further simplified with up-to
techniques. Unlike with environmental bisimilarity, we define smaller
techniques in Figure~\ref{fig:nf-upto} which, when combined together,
correspond to the usual bisimulation up to related contexts. Such a
fine-grained approach allows for a finer classification between strong
and regular up-to techniques.\footnote{We do not do the same with
  environmental bisimilarity, because unlike normal-form bisimilarity,
  it is defined primarily on closed terms, and therefore we do not
  consider, e.g., a bisimulation up to $\lambda$-abstraction with
  environmental bisimilarity.}

\begin{figure}
  \begin{mathpar}
    \inferrule{ }{\tm \utrefl\rel \tm}
    \and
    \inferrule{\tmzero \rel \tmone}{\lam \varx \tmzero \utlam\rel \lam \varx
      \tmone}
    \and
    \inferrule{\tmzero \rel \tmone}{\shift \vark \tmzero \utshift\rel \shift \vark
      \tmone}
    \and
    \inferrule{\tmzero \rel \tmone \\  \valzero \nfv\rel \valone}
    {\subst \tmzero \varx \valzero \utsubst\rel \subst \tmone \varx \valone}
    \and
    \inferrule{\tmzero \clocbv \tmzero' \\  \tmone \clocbv \tmone' \\ \tmzero' \rel \tmone'}
    {\tmzero \utred\rel \tmone}
    \and
    \inferrule{\tmzero \rel \tmone \\ \ctxzero \nfc\rel \ctxone}
    {\inctx \ctxzero \tmzero \utpctx\rel \inctx \ctxone \tmone}
    \and
    \inferrule{\tmzero \rel \tmone \\ \reset\ctxzero \nfc\rel \reset\ctxone}
    {\reset{\inctx \ctxzero \tmzero} \utpctxrst\rel \reset{\inctx \ctxone \tmone}}
    \and
    \inferrule{\tmzero \rel \tmone  \\ \tmzero, \tmone \mbox{ pure} \\ \inctx
      \rctxzero \varx \rel \inctx \rctxone \varx \\ \varx \mbox{ fresh}}
    {\inctx \rctxzero \tmzero \utectxpure\rel \inctx \rctxone \tmone}
\end{mathpar}
\caption{Up-to techniques for normal-form bisimilarity}%
\label{fig:nf-upto}
\end{figure}

The technique $\rawutred$ is the usual bisimulation up-to reduction, $\rawlam$
and $\rawutshift$ allow compatibility {\wrt}$\lambda$-abstraction and
\textshift, while compatibility for variables is a consequence of $\rawrefl$, as
we have $\varx \utrefl\rel \varx$ for all $\varx$. Bisimulation up to
substitution~$\rawutsubst$ is not uncommon for normal-form
bisimilarity~\cite{Lassen:MFPS99,Lassen:LICS06}. The remaining techniques deal
with evaluation contexts and behave the same way as the bisimulation up to
related contexts of Section~\ref{ss:env-upto}: each of them factors out related
contexts, and not simply a common context. We compare contexts using~$\nfc\cdot$
except for $\rawectxpure$, which uses a more naive test, as this technique plugs
contexts with only pure terms (values or delimited terms), which cannot
decompose the contexts. The usual bisimulation up to related contexts can be
obtained by composing the three up-to techniques about contexts.

\begin{lem}
  If $\tmzero \rel \tmone$ and $\rctxzero \nfc\rel \rctxone$ then
  $\inctx \rctxzero \tmzero \mathrel{(\rawpctx \cup
    (\rawectxpure\compo\rawpctxrst))(\rel)} \inctx \rctxone \tmone$.
\end{lem}

We can also derive compatibility {\wrt}application from $\rawrefl$
and $\rawpctx$.

\begin{lem}%
  \label{l:app}
  If $\tmzero \rel \tmzero'$ and $\tmone \rel \tmone'$, then $\app
  \tmzero \tmone \mathrel{(\rawpctx \compo (\rawid \cup \rawpctx
    \compo (\rawid \cup \rawrefl)))(\rel)} \app {\tmzero'}{\tmone'}$.
\end{lem}

\begin{proof}
  Let $\varx$ be a fresh variable; then $\app \varx \mtctx
  \nfc{\utrefl \rel} \app \varx \mtctx$. Combined with $\tmone \rel
  \tmone'$, it implies $\app x \tmone \utpctx {(\rawid \cup
    \rawrefl)(\rel)} \app x {\tmone'}$, i.e., $\app \mtctx \tmone
  \nfc{\utpctx {(\rawid \cup \rawrefl)(\rel)}} \app \mtctx
      {\tmone'}$. This combined with $\tmzero \rel \tmzero'$ using
      $\rawpctx$ gives the required result.
\end{proof}

Finally, compatibility {\wrt}\textreset can be deduced from
$\rawpctxrst$ by taking the empty context. (Defining a dedicated up-to
technique for \textreset would have some merit since it could be
proved strong, unlike $\rawpctxrst$~\cite{Biernacki-al:MFPS17}.)

\begin{thm}%
  \label{t:compatible-nf}
  The set $\setF \is \{ \rawrefl, \rawlam, \rawutshift, \rawutsubst,
  \rawpctx, \rawpctxrst, \rawectxpure, \rawid, \rawutred\}$ is
  diacritically compatible, with $\strong \setF = \setF \setminus \{
  \rawpctx, \rawpctxrst, \rawectxpure \}$.
\end{thm}
\noindent We explain what sets apart $\rawpctx$, $\rawpctxrst$, and
$\rawectxpure$ from the other techniques by sketching the progress proof for
$\rawpctx$.

\proof[Sketch] Let $\rel \pprogress \rel, \rels$, and $\inctx \ctxzero
\tmzero \utpctx\rel \inctx \ctxone \tmone$ such that $\ctxzero
\nfc\rel \ctxone$ and $\tmzero \rel \tmone$.  We proceed by case
analysis on $\ctxzero$ and $\tmzero$. Most cases are straightforward;
the problematic case is when~$\tmzero$ is a variable $x$ and $\ctxzero
= \app \mtctx \val$. Because $\tmzero \rel \tmone$, there exists
$\valone$ such that $\tmone \evalcbv \valone$ and $x \nfv\rel
\valone$. Because $\ctxzero \nfc\rel \ctxone$, we have $\inctx
\ctxzero y \rel \inctx \ctxone y$ for a fresh $y$, and therefore
$\inctx \ctxzero x \utsubst\rel \inctx \ctxone {\valone}$. We can
conclude with $\rawutsubst$, assuming it has been proved before: there
exists~$\tmone'$ such that $\inctx{\ctxone}{\tmone} \clocbv
\inctx{\ctxone}{\valone} \clocbv \tmone'$ and $\app x \val \nf{\of
  {{\fid\setF}^\omega} \rels} \tmone'$.

\noindent If we try to prove $\rawpctx$ strong, we would have $\rel
\pprogress \rels, \relt$ as a hypothesis. In the subcase sketched
above, we would have $x \nfv\rels \valone$ and $\inctx \ctxzero x
\utsubst\rels \inctx \ctxone {\valone}$ instead of $\rel$, and since
there is no progress hypothesis on $\rels$, we could not conclude. The
techniques $\rawectxpure$ and $\rawpctxrst$ have the same problematic
subcase.

Since compatibility {\wrt}the operators of the language can be
deduced from the techniques of Figure~\ref{fig:nf-upto}, we can
conclude that $\nfbisim$ is compatible using
Lemma~\ref{p:properties-compatibility-better}. We can then show that
$\nfbisim$ is sound {\wrt}$\ctxequiv$.

\begin{thm}
  We have $\nfbisim \mathop{\subseteq} \ctxequiv$.
\end{thm}
The following counter-example, inspired by Lassen~\cite{Lassen:LICS05}, shows
that the inclusion is in fact strict; normal-form bisimilarity is not complete.

\begin{prop}%
  \label{p:cex-dupl}
  We have $\reset {\app \varx i} \open\ctxequiv \lamp y {\reset{\app
      \varx i}} \iapp \reset {\app \varx i}$, but these terms are not
  normal-form bisimilar.
\end{prop}

\begin{proof}
  We prove that $\reset {\app \varx i} \open\ctxequiv \lamp y
  {\reset{\app \varx i}} \iapp \reset {\app \varx i}$ holds using
  applicative bisimilarity in
  Proposition~\ref{p:cex-nf-completeness}. They are not normal-form
  bisimilar, because the terms $\reset{z}$ and $\lamp y {\reset{\app
      \varx i}} \iapp \reset z$ (where $z$ is fresh) are not
  bisimilar: the former evaluates to $z$ while the latter evaluates to
  an open-stuck term.
\end{proof}

\noindent Lassen's other counter-example can also be adapted to $\lamshift$: we
can show that $\reset{\app x y} \iapp \Omega$ and $\Omega$ are contextually
equivalent but not normal-form bisimilar.

\begin{rem}
  Following Filinski's simulation of \textshift and \textreset in terms of
  \callccId{} and a single reference cell~\cite{Filinski:POPL94}, one can
  express the terms of the $\lamshift$-calculus in the
  $\lambda\mu\rho$-calculus~\cite{Stoevring-Lassen:POPL07}, a calculus with
  store and a construct similar to \callccId{}. Yet, St{\o}vring and Lassen's
  normal-form bisimilarity is sound and complete with respect to the contextual
  equivalence of $\lambda\mu\rho$~\cite{Stoevring-Lassen:POPL07}, while our
  relation is only sound. It shows that $\lambda\mu\rho$ is more expressive and
  can distinguish more terms than $\lamshift$, mostly because of the state
  construct. For example, the encodings of the two terms of
  Proposition~\ref{p:cex-dupl} in $\lambda\mu\rho$ would not be contextually
  equivalent in $\lambda\mu\rho$, since substituting for $x$ a value that, e.g.,
  increments a value of some reference cell, would lead to two different states
  that can be easily distinguished observationally.
\end{rem}

We show how up-to techniques can simplify the definitions of candidate relations
on several examples, starting with the double \textreset one.

\begin{exa}[double \textreset] To relate $\reset \tm$ and $\reset{\reset \tm}$,
  we can avoid the case analysis of Example~\ref{ex:reset-reset-nf} by noticing
  that
  $\reset \tm \mathrel{\of {\rawectxpure \compo (\rawid \cup \rawrefl)} \rel}
  \reset {\reset \tm}$ holds with
  $\rel {\mathord{\is}} \{ (\varx, \reset \varx) \}$, since $\reset \tm$ is
  pure. We then conclude by showing that $\rel$ is a bisimulation (eventually up
  to something), which is straightforward.
\end{exa}

 Bisimulation up to related contexts is easier to use with normal-form
bisimilarity compared to environmental bisimilarity, as we can plug
any related terms into any contexts (thanks to $\rawlam$ and
$\rawutshift$), without the restriction of Section~\ref{ss:env-upto}
that non-value terms are limited to evaluation contexts. As a result,
the equivalence proof for Turing's combinator and its variant can be
greatly simplified, as we can see below.

\begin{exa}[Turing's combinator]%
  \label{ex:fixed-point-nf-utc}
  Let
  \[
    \valzero  \is \lam y {\app y {\lamp z {\app {\app \Turing y} z}}} \mbox{ and }
    \valone  \is \lam y {\app y {\lamp z {\app {\app {\app {\lamp \varx {\reset
                  {\app \theta \varx}}}{\lamp \varx {\reset {\app \theta \varx}}}}
            y} z}}}.
  \]
  The relation
  \[ \rel \is \{(\Turing, \Turingshift), (\valzero \iapp y, \valone \iapp y),
    (\Turing, \app{\lamp \varx {\reset{\app \theta \varx}}}{\lamp \varx
      {\reset{\app \theta \varx}}}) \mmid y \mbox{ fresh} \}
  \] is a bisimulation up to context up to reduction. Indeed, we remind that
  $\Turing \evalcbv \valzero$ and $\Turingshift \evalcbv \valone$.  Reducing
  $\valzero \iapp y$ and $\valone \iapp y$ for a fresh $y$, we get terms sharing
  the common context $\app y {\lamp z {\app {\app \mtctx y} z}}$, and the two
  terms filling the holes (respectively $\Turing$ and
  $\app{\lamp \varx {\reset{\app \theta \varx}}}{\lamp \varx {\reset{\app \theta
        \varx}}}$) are in $\rel$. We can conclude using bisimulation up to
  related contexts, as we use it after an active reduction step. The terms
  $\Turing$ and
  $\app{\lamp \varx {\reset{\app \theta \varx}}}{\lamp \varx {\reset{\app \theta
        \varx}}}$ also reduce respectively to $\valzero$ and $\valone$, so we
  can conclude in the same way.
\end{exa}

\noindent As an extra example, we prove a variant of the \AXbetaomega axiom; the
axiom itself is proved in Section~\ref{ss:axioms-nf}.

\begin{exa}%
  \label{ex:small-step}
  If $\varx \notin \fv \ctx$, then
  $\reset{\app {\lamp \varx {\reset {\inctx \ctx \varx}}} \tm} \nfbisim \reset
  {\inctx \ctx \tm}$. Indeed, if
  \[ \rel {\mathord{\is}} \{(\reset{\app {\lamp \varx {\reset {\inctx \ctx
            \varx}}} y}, \reset {\inctx \ctx y}) \mid y \mbox{ fresh}\}.\] then
  $\reset{\app {\lamp \varx {\reset {\inctx \ctx \varx}}} \tm} \mathrel{\of
    {\rawpctxrst \compo (\rawid \cup \rawrefl)} \rel} \reset {\inctx \ctx
    \tm}$. Furthermore,
  $\reset{\app {\lamp \varx {\reset {\inctx \ctx \varx}}} y} \redcbv
  \reset{\reset {\inctx \ctx y}}$, and we prove in
  Example~\ref{ex:reset-reset-nf} that
  $\reset{\reset {\inctx \ctx y}} \nfbisim \reset{\inctx \ctx y}$; therefore
  $\rel {\mathord\subseteq} \utred\nfbisim$, and we can conclude from here.
\end{exa}

\subsection{Refined Normal-Form Bisimilarity}%
\label{ss:refined}

The normal-form bisimulation of Definition~\ref{d:nf-bisim} is too
discriminating with control-stuck terms, as we can see with the following
example.

\begin{prop}%
  \label{p:cex-stuck}
  We have $\shift \vark i \ctxequiv \app{(\shift \vark i)} \Omega$, but these
  terms are not normal-form bisimilar.
\end{prop}

\begin{proof}
  We can easily prove that
  $\shift \vark i \ctxequiv \app{(\shift \vark i)} \Omega$ holds with
  applicative bisimilarity or Definition~\ref{d:rnf-bisim}. They are not
  normal-form bisimilar, since the contexts $\mtctx$ and $\apctx \mtctx \Omega$
  are not related by $\nf\nfbisim$ ($\varx$ converges while $\app \varx \Omega$
  diverges).
\end{proof}

When comparing two control-stuck terms $\inctx \ctxzero {\shift \vark
  \tmzero}$ and $\inctx \ctxone {\shift \vark \tmone}$, normal-form
bisimilarity considers the contexts $\ctxzero$, $\ctxone$ and the
\textshift bodies $\tmzero$, $\tmone$ separately, while they are
combined if the control-stuck terms are put under a \textreset and the
capture goes through. We propose a more refined definition of
normal-form bisimulation which tests stuck terms by simulating the
capture, while taking into account the fact that a context bigger
than~$\ctxzero$ and~$\ctxone$ can be captured. We do so by introducing
a \emph{context variable} to represent the context beyond $\ctxzero$
and $\ctxone$. We let $\cvar$ range over a set of context
variables. We introduce such a variable when simulating a capture,
where the context is always captured with its \textreset. To simulate
this, we suppose that $\cvar$ stands for a pure context surrounded by
a delimiter. As a result, the definition of $\nf\rel$ on control-stuck
terms becomes as in Figure~\ref{f:refined-upto}.

\begin{rem}
We could try to use a regular variable $\vark'$ to play the role of a
context variable, and define the extension~$\nf\rel$ on control-stuck
terms as follows:

\begin{mathpar}
  \inferrule{\reset {\subst {\tmzero} \vark {\lam \varx
        {\reset {\app {\vark'}{\inctx \ctxzero \varx}}}}} \rel \reset
    {\subst{\tmone} \vark {\lam \varx {\reset{\app {\vark'}{\inctx \ctxone
              \varx}}}}} \\ \vark', \varx \textrm{
      fresh}}
  {\inctx \ctxzero {\shift \vark \tmzero} \nf\rel \inctx \ctxone {\shift \vark
      \tmone}}
\end{mathpar}

\noindent
However, such variables are substituted with contexts and not with
values, and so they have to be treated separately from regular
variables.
\end{rem}

\begin{figure}
  \begin{mathpar}
  % \inferrule{\inctx {\ctx_0} \varx \rel \inctx {\ctx_1} \varx
  %   \\ \varx \mbox{ fresh}}
  % {\ctx_0 \nfc\rel \ctx_1} \and
  \inferrule{\inctx {\dctx_0} \varx \rel \inctx {\dctx_1}
      \varx
    \\ \inctx {\rctx_0} \varx \rel
    \inctx {\rctx_1} \varx \\ \varx \mbox{ fresh}}
  {\inctx {\rctx_0} {\dctx_0}  \nfc\rel \inctx{\rctx_1}{\dctx_1}}
  % \and
  % \inferrule{ \valzero \iapp \varx \rel \valone \iapp \varx \\ \varx \textrm{
  %     fresh}}
  % {\valzero \nfv\rel \valone}
  \and
  \inferrule{\reset {\subst {\tmzero} \vark {\lam \varx
        {\cvctx \cvar {\inctx \ctxzero \varx}}}} \rel \reset
    {\subst{\tmone} \vark {\lam \varx {\cvctx \cvar {\inctx \ctxone
              \varx}}}} \\ \cvar, \varx \textrm{
      fresh}}
  {\inctx \ctxzero {\shift \vark \tmzero} \nf\rel \inctx \ctxone {\shift \vark
      \tmone}}
  % \and
  % \inferrule{\rctxzero \nfc\rel \rctxone \\ \valzero \nfv\rel \valone}
  % {\inctx \rctxzero {\app \varx \valzero} \nf\rel \inctx \rctxone {\app \varx
  %     \valone}}
\end{mathpar}

  \vspace{1em}
\textbf{Up-to techniques specific to refined bisimilarity}

   \begin{mathpar}
    % \inferrule{ }{\tm \utrefl\rel \tm}
    % \and
    % \inferrule{\tmzero \rel \tmone}{\lam \varx \tmzero \utlam\rel \lam \varx
    %   \tmone}
    % \and
    % \inferrule{\tmzero \rel \tmone}{\shift \vark \tmzero \utshift\rel \shift \vark
    %   \tmone}
    % \and
    \inferrule{\tmzero \rel \tmone}{\cvctx \cvar \tmzero \utcvar\rel \cvctx \cvar
      \tmone}
    \and
    % \inferrule{\tmzero \rel \tmone \\  \valzero \nfv\rel \valone}
    % {\subst \tmzero \varx \valzero \utsubst\rel \subst \tmone \varx \valone}
    % \and
    \inferrule{\tmzero \rel \tmone \\  \dctxzero \nfc\rel \dctxone}
    {\subst \tmzero \cvar \dctxzero \utcsubst\rel \subst \tmone \cvar \dctxone}
    % \and
    % \inferrule{\tmzero \clocbv \tmzero' \\  \tmone \clocbv \tmone' \\ \tmzero' \rel \tmone'}
    % {\tmzero \utred\rel \tmone}
    % \and
    % \inferrule{\tmzero \rel \tmone  \\ \inctx
    %   \ctxzero \varx \rel \inctx \ctxone \varx \\ \varx \mbox{ fresh}}
    % {\inctx \ctxzero \tmzero \utpctx\rel \inctx \ctxone \tmone}
    % \and
    % \inferrule{\tmzero \rel \tmone  \\ \tmzero, \tmone \mbox{ pure} \\ \inctx
    %   \rctxzero \varx \rel \inctx \rctxone \varx \\ \varx \mbox{ fresh}}
    % {\inctx \rctxzero \tmzero \utectxpure\rel \inctx \rctxone \tmone}
\end{mathpar}
\caption{Extension to normal forms and up-to techniques for the
  refined bisimilarity}%
\label{f:refined-upto}
\end{figure}

Formally we extend the syntax of terms and evaluation contexts
$\rctx$, and we introduce a new kind of \emph{delimited contexts}
ranged over by $\dctx$.
\begin{grammar}
 & \textrm{Terms:}& \tm & \bnfdef \ldots \bnfor \cvctx \cvar \tm\\
 & \textrm{Evaluation contexts:} \quad & \rctx & \bnfdef \ldots \bnfor \cvctx \cvar \rctx \\
 & \textrm{Delimited contexts:} \quad & \dctx & \bnfdef \reset \ctx \bnfor \cvctx \cvar \ctx
\end{grammar}
We write $\subst \tm \cvar \dctx$ for the context substitution of $\cvar$ by
$\dctx$ in $\tm$, defined so that
$\subst {(\cvctx \cvar \tm)} \cvar \dctx \is \inctx \dctx {\subst \tm \cvar
  \dctx}$,
$\subst {(\cvctx {\cvar'} \tm)} \cvar \dctx \is \cvctx {\cvar'}{\subst
  \tm \cvar \dctx}$ if $\cvar' \neq \cvar$, and the substitution is propagated
recursively on subterms in the other cases. The capture reduction rule is
changed to take delimited contexts into account.
\[
\inctx \rctx {\inctx \dctx {\shift \vark \tm}}
  \redcbv  \inctx \rctx{\reset{\subst \tm \vark
      {\lam \varx {\inctx \dctx \varx}}}} \mbox{ with } \varx \notin \fv{\dctx}
 \quad \RRshift
\]

Given a relation $\rel$ on extended open terms, we keep the
definitions of $\nfv\rel$, $\nf\rel$ on open-stuck terms,
and~$\nfc\rel$ on pure contexts as in Figure~\ref{f:ext-rnf}, and we
change $\nf\rel$ on control-stuck terms and $\nfc\rel$ on any contexts
as in Figure~\ref{f:refined-upto}. The latter change is to account for
delimited contexts. The extended calculus also features a new kind of
normal forms, of the shape $\inctx \rctx {\cvctx \cvar \val}$, called
\emph{context-stuck terms}. They are similar to open-stuck terms but
are tested differently, as we can see in the definition of progress.

\begin{defi}%
  \label{d:rnf-bisim}
  A relation $\rel$ on extended open terms diacritically progresses to $\rels$,
  $\relt$ written $\rel \pprogressr \rels, \relt$, if
  $\rel \mathop\subseteq \rels$, $\rels \mathop\subseteq \relt$, and
  $\tmzero \rel \tmone$ implies:
  \begin{itemize}
  \item if $\tmzero \redcbv \tmzero'$, then there exists $\tmone'$ such that
    $\tmone \clocbv \tmone'$ and $\tmzero' \relt \tmone'$;
  \item if $\tmzero$ is a value, then there exists $\valone$ such that
    $\tmone \evalcbv \valone$, and $\tmzero \nfv\rels \valone$;
  \item if $\tmzero = \inctx \rctxzero {\cvctx \cvar \valzero}$, then there
    exists $\rctxone$, $\valone$ such that $\tmone \evalcbv \inctx \rctxone
    {\cvctx \cvar \valone}$, $\inctx \rctxzero {\reset \mtctx} \nfc\relt \inctx \rctxone
    {\reset \mtctx}$, and $\valzero \nfv\rels \valone$;
  \item if $\tmzero$ is an open-stuck or control-stuck term, then
    there exists $\tmone'$ such that $\tmone \evalcbv \tmone'$ and
    $\tmzero \nf\relt \tmone'$;
  \item the converse of the above conditions on $\tmone$.
  \end{itemize}
  A refined normal-form bisimulation is a relation $\rel$ such that
  $\rel \pprogressr \rel, \rel$. Refined normal-form
  bisimilarity~$\rbisim$ is the largest refined normal-form
  bisimulation.
\end{defi}

\noindent The clause for context-stuck terms relates the terms $\inctx
\rctxzero {\cvctx \cvar \valzero}$ and $\inctx \rctxone {\cvctx \cvar
  \valone}$ by comparing the contexts $\inctx \rctxzero {\reset
  \mtctx}$ and $\inctx \rctxone {\reset \mtctx}$ because $\cvar$
implicitly includes a \textreset. This essentially amounts to equate
$\inctx \rctxzero \varx$ and $\inctx \rctxone \varx$ for a fresh
$\varx$. In contrast with open-stuck terms, we relate the contexts
with $\relt$ but the values with $\rels$, thus forbidding the use of
regular up-to techniques when comparing values. Our goal is to prevent
the application of the new $\rawutcsubst$ technique in that case; we
explain why after Theorem~\ref{t:compatible-rnf}.

To compare refined bisimilarity to the other relations on \lamshift,
we translate the terms of the extended calculus back to
\lamshift. Given an injective mapping $f$ from context variables to
regular variables, we define the translation $\tr f \cdot$ on extended
terms so that $\tr f {\cvctx \cvar \tm} \is \reset{\of f \cvar \iapp
  \tr f \tm}$ and so that it is recursively applied to subterms in the
other cases. The translation is defined on contexts in a similar
way. It is easy to see that if $\tm$ is a plain \lamshift-term, then
$\tr f \tm = \tm$ for all~$f$, and that reduction is preserved by the
translation.

\begin{prop}
  For all $\tm$, $\tm'$, and $f$, $\tm \clocbv \tm'$ iff
  $\tr f \tm \clocbv \tr f {\tm'}$.
\end{prop}

\noindent We can relate $\nfbisim$ and $\rbisim$ thanks to the translation.
\begin{prop}%
  \label{p:bisim-subset-rbisim}
  For all $\tmzero$, $\tmone$, and $f$ such that the image of $f$ does not
  intersect $\fv \tmzero$ and $\fv \tmone$, if
  $\tr f \tmzero \nfbisim \tr f \tmone$, then $\tmzero \rbisim \tmone$.
\end{prop}

\noindent The condition on $f$ allows for the distinction between the
evaluations to context-stuck and open-stuck terms. If
$\tr f \tmzero \evalcbv \inctx \rctx {\app \varx \val}$ for some $\rctx$,
$\varx$ and $\val$, then either $\varx \in \fv \tmzero$,
$\tmzero \evalcbv \inctx {\rctx'}{\app \varx {\val'}}$,
$\tr f {\rctx'} = \rctx$, and $\tr f {\val'} = \val$, or $\varx = \of f \cvar$
for some $\cvar$, $\tmzero \evalcbv \inctx {\rctx'}{\cvctx \cvar {\val'}}$,
$\tr f {\inctx {\rctx'}{\reset \mtctx}} = \rctx$, and $\tr f {\val'} = \val$.

\begin{proof}
  We prove that
  $\rel \mathord{\is}\{ (\tmzero, \tmone) \mid \tr f \tmzero \nfbisim \tr f
  \tmone \}$ is a refined bisimulation. What needs to be checked are
  context-stuck terms and control-stuck terms. If
  $\tmzero = \inctx \rctxzero {\cvctx \cvar \valzero}$, then
  $\tr f \tmzero = \inctx {\tr f \rctxzero}{\reset{\app {\of f \cvar}{\tr f
        \valzero}}}$, and there exists~$\rctxone$,~$\valone$ such that
  $\tr f \tmone \evalcbv \inctx {\tr f \rctxone}{\reset{\app {\of f \cvar}{\tr f
        \valone}}}$ with
  $\tr f \tmzero \nf\nfbisim \inctx {\tr f \rctxone}{\reset{\app {\of f
        \cvar}{\tr f \valone}}}$, i.e.,
  $\tr f {\inctx \rctxzero{\reset \mtctx}} \nfc\nfbisim \tr f {\inctx \rctxone{\reset
      \mtctx}}$ and $\tr f \valzero \nfv\nfbisim \tr f \valone$. Therefore we have
  $\tmone \evalcbv \inctx \rctxone {\cvctx \cvar \valone}$, and the clause for
  context-stuck terms is verified.

  If $\tmzero = \inctx \ctxzero {\shift \vark {\tmzero'}}$, then $\tr f \tmzero =
  \inctx {\tr f \ctxzero}{\shift \vark {\tr f {\tmzero'}}}$, and there exists
  $\ctxone$, $\tmone'$ such that $\tr f \tmone \evalcbv \inctx {\tr f \ctxone}{\shift
    \vark {\tr f {\tmone'}}}$, $\tr f \ctxzero \nfc\nfbisim \tr f \ctxone$, and $\reset
  {\tr f {\tmzero'}} \nfbisim \reset {\tr f {\tmone'}}$. But $\nfbisim$ is
  compatible and substitutive, therefore we have
  \[\reset
    {\subst {\tr f {\tmzero'}} \vark {\lam \varx {\reset {\app {\of f \cvar}{\inctx
              {\tr f \ctxzero} \varx}}}}} \nfbisim \reset {\subst{\tr f {\tmone'}} \vark
      {\lam \varx {\reset{\app {\of f \cvar}{\inctx {\tr f \ctxone} \varx}}}}}
  \]
  for some fresh $\varx$ and $\cvar$. Consequently, we have
  $\reset {\subst {\tmzero'} \vark {\lam \varx {\cvctx
            \cvar {\inctx \ctxzero \varx}}}} \rel \reset {\subst{\tmone'}
    \vark {\lam \varx {\cvctx \cvar {\inctx \ctxone \varx}}}}$, as
  wished.
\end{proof}

A direct consequence of Proposition~\ref{p:bisim-subset-rbisim} is that
$\nfbisim \mathop\subset \rbisim$. The inclusion is strict, because~$\rbisim$
relates the terms of Proposition~\ref{p:cex-stuck}, while $\nfbisim$ does not.

\subsubsection*{Up-to techniques and soundness} The up-to techniques for refined
bisimilarity are the same as for normal-form bisimilarity
(Figure~\ref{fig:nf-upto}), except that we add techniques specific to context
variables $\rawutcsubst$ and $\rawutcvar$ (defined in
Figure~\ref{f:refined-upto}), and we remove $\rawpctxrst$, as it can be directly
expressed in terms of the two new techniques. In fact, $\rawutcsubst$ is a bit
more powerful than $\rawpctxrst$, as several copies of the same context can be
abstracted away with $\rawutcsubst$ against only one for $\rawpctxrst$. For the
$\rawectxpure$ technique, pure terms now include terms of the form
$\cvctx \cvar \tm$ in addition to values and delimited terms $\reset \tm$.

\begin{thm}%
  \label{t:compatible-rnf}
  The set
  $\setF \is \{ \rawrefl, \rawlam, \rawutshift, \rawutcvar, \rawutsubst,
  \rawutcsubst, \rawpctx, \rawectxpure, \rawid, \rawutred\}$ is diacritically
  compatible, with
  $\strong \setF = \setF \setminus \{ \rawutcsubst, \rawpctx, \rawectxpure \}$.
\end{thm}
Unsurprisingly, the technique $\rawutcsubst$ is not strong as it behaves like
$\rawpctxrst$. In particular, it exhibits the same problematic subcase as the
one presented after Theorem~\ref{t:compatible-nf}, by taking
$\tmzero \is \cvctx \cvar \varx$ and $\dctxzero = \reset{\app \mtctx
  \val}$. Then from $\tmzero \rel \tmone$ and $\rel \pprogress \rel, \rels$, we
know there exist~$\rctxone$ and~$\valone$ such that
$\tmone \evalcbv \inctx \rctxone {\cvctx \cvar \valone}$,
$\mtctx \nfc\rels \rctxone$, and $\varx \nfv \rel \valone$. From there, we can
conclude as in Section~\ref{ss:upto-nf}, using $\rawutsubst$; more details are
given in the appendix. To conclude in that case, it is important to have
$\varx \nfv\rel \valone$ and not $\varx \nfv\rels \valone$, justifying why the
test for values is passive for context-stuck terms.

From Theorem~\ref{t:compatible-rnf} and
Proposition~\ref{p:properties-compatibility-better}, we deduce that $\rbisim$ is
compatible, which we then use to show that $\rbisim$ is sound {\wrt}$\ctxequiv$
in the following sense.

\begin{thm}
  For all $(\tmzero, \tmone) \in \terms^2$, if $\tmzero \rbisim \tmone$, then
  $\tmzero \ctxequiv \tmone$.
\end{thm}
\noindent As we restrict $\tmzero$ and $\tmone$ to plain $\lamshift$ terms, we
do not need to work up to the translation. The relation $\rbisim$ is not
complete because it still does not relate the terms of
Proposition~\ref{p:cex-dupl}. We would like to stress that even though $\rbisim$
equates more contextually equivalent terms than $\nfbisim$, the latter is still
useful, since it leads to very simple proofs of equivalence, as we can see with
the examples of Sections~\ref{ss:upto-nf} and~\ref{ss:axioms-nf}. Therefore,
$\rbisim$ does not disqualify $\nfbisim$ as a proof technique. In fact, they can
be used together, as in the next example.

\begin{exa}%
  \label{ex:refined}
  If $\vark' \notin \fv \ctx \cup \fv \tm$ and $\varx \notin \fv \ctx$, then
  $\inctx \ctx {\shift \vark \tm} \rbisim \shift {\vark'}{\subst \tm \vark {\lam
      \varx {\reset {\app {\vark'}{\inctx \ctx \varx}}}}}$.  The two terms are
  control stuck, therefore we have to prove that
  $\reset{\subst \tm \vark {\lam \varx {\cvctx \cvar {\inctx \ctx
            \varx}}}} \rbisim \reset{\subst \tm \vark {\lam \varx {\reset {\app
          {\lamp y {\cvctx \cvar y}}{\inctx \ctx \varx}}}}}$ holds for
  a fresh $\cvar$. Let $f$ be an injective mapping verifying the conditions of
  Proposition~\ref{p:bisim-subset-rbisim}. We know that
  $\reset {\app {\of f \cvar}{\inctx \ctx \varx}} \nfbisim \reset{\app {\lamp y
      {\reset {\app {\of f \cvar} y}} }{\inctx \ctx \varx}}$ holds by Example~\ref{ex:small-step}, so we have
  $\cvctx \cvar {\inctx \ctx \varx} \rbisim \reset{\app {\lamp y
      {\cvctx \cvar y} }{\inctx \ctx \varx}}$ by Proposition~\ref{p:bisim-subset-rbisim}. We can then conclude using $\rawrefl$, $\rawlam$,
  and $\rawutsubst$.

  Proving this result using only the regular normal-form bisimulation would
  require us to equate $\inctx \ctx y$ and $y$ (where $y$ is fresh), which is
  not true in general (take $\ctx= \vctx {\lamp z \Omega} \mtctx$).
\end{exa}

\begin{figure}

\begin{mathpar}
  \inferrule{\inctx {\dctx_0} \varx \rel \inctx {\dctx_1}
      \varx
    \\ \inctx {\rctx_0} \varx \rel
    \inctx {\rctx_1} \varx \\ \varx \mbox{ fresh}}
  {\inctx {\rctx_0} {\dctx_0}  \nfc\rel \inctx{\rctx_1}{\dctx_1}}
  \and
  \inferrule{\cvctx \cvar {\valzero \iapp \varx} \rel \cvctx
      \cvar {\valone \iapp \varx} \\ \varx, \cvar \textrm{
      fresh}}
  {\valzero \nfv\rel \valone}
  \and
  \inferrule{\rctxzero \nfc\rel \rctxone \\ \valzero \nfv\rel \valone}
  {\inctx \rctxzero {\app \varx \valzero} \nf\rel \inctx \rctxone {\app \varx
      \valone}}
\end{mathpar}

\vspace{1em}
\textbf{Up-to techniques}

  \begin{mathpar}
    \inferrule{ }{\prg \utrefl\rel \prg}
    \and
    \inferrule{\prgzero \rel \prgone}{\cvctx \cvar {\lam \varx \prgzero}
      \utlam\rel \cvctx \cvar {\lam \varx \prgone}}
    \and
    \inferrule{\prgzero \rel \prgone}{\cvctx \cvar {\shift \vark
          \prgzero} \utshift\rel \cvctx \cvar {\shift \vark \prgone}}
    \and
    \inferrule{\prgzero \rel \prgone}{\cvctx \cvar \prgzero \utcvar\rel
      \cvctx \cvar \prgone}
    \and
    \inferrule{\prgzero \rel \prgone \\  \valzero \nfv\rel \valone}
    {\subst \prgzero \varx \valzero \utsubst\rel \subst \prgone \varx \valone}
    \and
    \inferrule{\prgzero \rel \prgone \\  \dctxzero \nfc\rel \dctxone}
    {\subst \prgzero \cvar \dctxzero \utcsubst\rel \subst \prgone \cvar \dctxone}
    \and
    \inferrule{\prgzero \clocbv \prgzero' \\  \prgone \clocbv \prgone' \\ \prgzero' \rel \prgone'}
    {\prgzero \utred\rel \prgone}
    \and
    \inferrule{\prgzero \rel \prgone  \\ \inctx
      \rctxzero \varx \rel \inctx \rctxone \varx \\ \inctx
      \rctxzero \varx, \inctx \rctxone \varx \mbox{ pure} \\  \varx \mbox{ fresh}}
    {\inctx \rctxzero \prgzero \utectxpure\rel \inctx \rctxone
        \prgone}
\end{mathpar}
  \caption{Extension to normal forms and up-to techniques for the
    original semantics}%
  \label{f:nf-original}
\end{figure}

\subsection{Normal-Form Bisimulation for the Original Semantics}%
\label{ss:nf-original}

Any sound bisimilarity for the relaxed semantics, such as $\nfbisim$ or
$\rbisim$, is also sound for the original semantics. We define in this section a
bisimilarity which, while being not complete {\wrt}$\ctxequivp$, still relates
more terms in the original semantics than~$\nfbisim$ or $\rbisim$. We follow the
same principle as in Section~\ref{ss:env-LTS-original}, and define a
bisimilarity which primarily compares pure terms. We then extend it to any terms
by introducing a context variable which stands for a potential evaluation
context, as with refined bisimilarity.

Formally, we work on the extended calculus of Section~\ref{ss:refined}, and we
let $\prg$ range over pure terms, which are now of three possible shapes.
\begin{grammar}
  & \textrm{Pure terms:} \quad & \prg & \bnfdef \val \bnfor \reset \tm \bnfor  \cvctx
  \cvar \tm
\end{grammar}
We update the definition of $\nfv\cdot$, $\nfc\cdot$, and $\nf\cdot$ in
Figure~\ref{f:nf-original}. Because we work on pure terms, the control-stuck
terms case has been removed; similarly, the evaluation contexts in the
context-stuck and open-stuck terms cases are delimited, so the pure context case
of $\nfc\cdot$ is no longer useful. The definition of $\nfv\cdot$ has been
changed so that we compare pure terms in its premise.

\begin{defi}%
  \label{d:onf-bisim}
  A relation $\rel$ on extended pure open terms diacritically progresses to $\rels$,
  $\relt$ written $\rel \pprogresso \rels, \relt$, if
  $\rel \mathop\subseteq \rels$, $\rels \mathop\subseteq \relt$, and
  $\prgzero \rel \prgone$ implies:
  \begin{itemize}
  \item if $\prgzero \redcbv \prgzero'$, then there exists $\prgone'$ such that
    $\prgone \clocbv \prgone'$ and $\prgzero' \relt \prgone'$;
  \item if $\prgzero$ is a value, then there exists $\valone$ such that
    $\prgone \evalcbv \valone$, and $\prgzero \nfv\rels \valone$;
  \item if $\prgzero = \inctx \rctxzero {\cvctx \cvar \valzero}$, then there
    exists $\rctxone$, $\valone$ such that $\prgone \evalcbv \inctx \rctxone
    {\cvctx \cvar \valone}$, $\inctx \rctxzero {\reset \mtctx} \nfc\relt \inctx \rctxone
    {\reset \mtctx}$, and $\valzero \nfv\rels \valone$;
  \item if $\prgzero$ is an open-stuck term, then there exist $\prgone'$ such
    that $\prgone \evalcbv \prgone'$ and $\prgzero \nf\relt \prgone'$;
  \item the converse of the above conditions on $\prgone$.
  \end{itemize}
  A pure normal-form bisimulation is a relation $\rel$ such that
  $\rel \pprogresso \rel, \rel$. Pure normal-form bisimilarity~$\onfbisim$ is
  the largest pure normal-form bisimulation.
\end{defi}
\noindent Again, testing values in the context-stuck terms case is passive, to
prevent $\rawutcsubst$ to be used here; otherwise, from
$\cvctx \cvar \valzero \rel \cvctx \cvar \valone$, we could relate
$\cvctx {\cvar'}{\valzero \iapp \varx}$ and
$\cvctx {\cvar'}{\valone \iapp \varx}$ directly for any $\valzero$
and~$\valone$.

We extend $\onfbisim$ to all terms as follows: $\tmzero \onfbisim \tmone$ if
$\cvctx \cvar \tmzero \onfbisim \cvctx \cvar \tmone$ for a fresh $\cvar$. Pure
bisimilarity relates more terms in the original semantics than the normal-form
bisimilarities of the relaxed semantics.

\begin{prop}%
  \label{p:nfb-in-onfb}
  We have $\nfbisim \mathop\subsetneq \onfbisim$ and $\rbisim
  \mathop\subsetneq \onfbisim$.
\end{prop}

\begin{proof}
  Because $\nfbisim \mathop\subseteq \rbisim$, it is enough to show that
  $\rbisim \mathop\subseteq \onfbisim$. Let $\tmzero \rbisim \tmone$; because
  $\rbisim$ is compatible, we have
  $\cvctx \cvar \tmzero \rbisim \cvctx \cvar \tmone$. We then
  prove that $\rbisim$ is a pure bisimulation. On pure terms, the tests of the
  two notions of bisimulation differ only on values: we have
  $\valzero \iapp \varx \rbisim \valone \iapp \varx$, and we need $\cvctx \cvar
  {\valzero \iapp \varx} \rbisim \cvctx \cvar {\valone \iapp \varx}$. We can
  easily conclude using again the fact that $\rbisim$ is compatible.
\end{proof}
\noindent The inclusions are strict, as we show in Proposition~\ref{p:skkt-nf} that
$\onfbisim$ verifies the \AXshiftelim axiom while the others two do not.

The up-to techniques for $\onfbisim$ are defined in Figure~\ref{f:nf-original};
they are essentially the same as for refined bisimilarity with some minor
adjustments to ensures that we relate pure terms in the premises as well as in
the conclusion. We also remove the now useless $\rawpctx$ technique.

\begin{thm}%
  \label{t:compatible-onf}
  The set
  $\setF \is \{ \rawrefl, \rawlam, \rawutshift, \rawutreset,
  \rawutsubst, \rawutcsubst, \rawectxpure, \rawid, \rawutred\}$ is diacritically
  compatible, with
  $\strong \setF = \setF \setminus \{ \rawutcsubst, \rawectxpure \}$.
\end{thm}

\noindent We deduce that $\onfbisim$ is compatible on pure terms. For
compatibility {\wrt}any terms, let $\tmzero \onfbisim \tmone$; then by
definition, $\cvctx \cvar \tmzero \onfbisim \cvctx \cvar \tmone$ for a fresh
$\cvar$. We can then deduce compatibility {\wrt}evaluation contexts
(application and \textreset) with $\rawutcsubst$. For the remaining constructs,
we need separate proofs that $\lam \varx \tmzero \onfbisim \lam \varx \tmone$
and $\shift \vark \tmzero \onfbisim \shift \vark \tmone$, but these are
straightforward. Consequently, $\onfbisim$ is compatible on all terms, and we
can show it is sound {\wrt}$\ctxequivp$.

\begin{thm}
  For all $(\tmzero, \tmone) \in \terms^2$, if $\tmzero \onfbisim \tmone$, then
  $\tmzero \ctxequivp \tmone$.
\end{thm}

\begin{exa}%
  \label{ex:s-tail-nf}
  We prove Example~\ref{ex:stail-env} again with pure normal-form bisimulation:
  if $\vark \notin \fv\tmone$, then
  $\app{\lamp \varx{\shift \vark \tmzero}} \tmone \onfbisim \shift\vark {\appp
    {\lamp \varx \tmzero} \tmone}$. We want to relate
  $\cvctx \cvar {(\app{\lamp \varx{\shift \vark \tmzero}} \tmone)}$ with
  $\cvctx \cvar {\shift\vark {\appp {\lamp \varx \tmzero} \tmone}}$ for
  a fresh $\cvar$. But
  $\cvctx \cvar {\shift\vark {\appp {\lamp \varx \tmzero} \tmone}}
  \redcbv \reset{\app {\lamp \varx {\subst \tmzero \vark {\lam y {\cvctx
              \cvar y}}}} \tmone}$. Let
  \[ \rel {\mathord{\is}} \{(\cvctx {\cvar}{(\app{\lamp \varx{\shift
            \vark \tmzero}} z)}, \reset{\app {\lamp \varx {\subst \tmzero \vark
          {\lam y {\cvctx \cvar y}}}} z}) \mid z \mbox{ fresh} \}; \]
  then
  $\cvctx \cvar {(\app{\lamp \varx{\shift \vark \tmzero}} \tmone)}
  \mathrel{\of {\rawutcsubst \compo (\rawid \cup \rawrefl)} \rel} \reset{\app
    {\lamp \varx {\subst \tmzero \vark {\lam y {\cvctx \cvar y}}}}
    \tmone}$. The relation $\rel$ is a bisimulation up to $\rawutred$ and
  $\rawrefl$, since we have
  $\cvctx \cvar {(\app{\lamp \varx{\shift \vark \tmzero}} z)} \redcbv^2
  \reset{\subst{\subst \tmzero \vark {\lam y {\cvctx \cvar y}}} \varx
    z}$ and
  $\reset{\app {\lamp \varx {\subst \tmzero \vark {\lam y {\cvctx \cvar
              y}}}} z} \redcbv \reset{\subst{\subst \tmzero \vark {\lam y
        {\cvctx \cvar y}}} \varx z}$; we obtain two identical terms.
\end{exa}

\subsection{Proving the Axioms}%
\label{ss:axioms-nf}

We provide further examples by proving the axioms with normal-form
bisimilarities. As usual, the \AXbeta, \AXresetshift, \AXresetval, and \AXbeta
axioms are consequences of the fact that reduction is included in the
bisimilarity.

\begin{prop}%
  \label{l:redcbv-in-nf}
  If $\tm \redcbv \tm'$, then $\tm \nfbisim \tm'$.
\end{prop}

\begin{proof}
  The relation $\{(\tm, \tm') \mmid \tm \redcbv \tm'\}$ is a normal-form
  bisimulation up to $\rawrefl$.
\end{proof}

\begin{prop}[\AXshiftreset axiom]
  We have $\shift \vark {\reset \tm} \nfbisim \shift \vark \tm$.
\end{prop}

\begin{proof}
  These terms are stuck, so we have to show that
  $\reset{\reset \tm} \nfbisim \reset \tm$ (proved in
  Example~\ref{ex:reset-reset-nf}) and $\mtctx \nfc\nfbisim \mtctx$ (but
  $\nfbisim$ is reflexive).
\end{proof}

\begin{prop}[\AXresetlift axiom]%
  \label{p:resetlift-nf}
  We have $\reset {\app {\lamp \varx \tmzero}{\reset \tmone}} \nfbisim \app
  {\lamp \varx {\reset \tmzero}}{\reset \tmone}$.
\end{prop}

\begin{proof}
  If
  $\rel {\mathord{\is}} \{(\reset {\app {\lamp \varx \tmzero} y}, \app {\lamp
    \varx {\reset \tmzero}} y ) \mid y \mbox{ fresh}\}$, then
  $\reset {\app {\lamp \varx \tmzero}{\reset \tmone}} \mathrel{\of
    {\rawectxpure \compo (\rawid \cup \rawrefl)} \rel} \app {\lamp \varx
    {\reset \tmzero}}{\reset \tmone}$, and $\rel$ is a bisimulation up to
  $\rawutred$ and $\rawrefl$, since the two terms in $\rel$ reduces to $\reset
  {\subst \tmzero \varx y}$.
\end{proof}

\begin{prop}[\AXbetaomega axiom]%
  \label{p:omega-nf}
  If $\varx \notin \fv \ctx$, then $\app {\lamp \varx {\inctx \ctx
      \varx}} \tm \nfbisim \inctx \ctx \tm$.
\end{prop}

\begin{proof}
  If
  $\rel {\mathord{\is}} \{ (\app {\lamp \varx {\inctx \ctx \varx}} y, \inctx
  \ctx y) \mid y \mbox{ fresh} \}$, then
  $\app {\lamp \varx {\inctx \ctx \varx}} \tm \mathrel{\of {\rawpctx \compo
      (\rawid \cup \rawrefl)} \rel} \inctx \ctx \tm$, and~$\rel$ is a
  bisimulation up to $\rawutred$ and $\rawrefl$, since the two terms in $\rel$
  reduces to $\inctx \ctx y$.
\end{proof}

\begin{prop}[\AXshiftelim axiom]%
  \label{p:skkt-nf}
  If $\vark \notin \fv \tm$, then $\tm \onfbisim \shift \vark
  {\app \vark \tm}$.
\end{prop}

\begin{proof}
  We must relate $\reset{\cvctx \cvar {\shift \vark {\app \vark \tm}}}$ and
  $\reset {\cvctx \cvar \tm}$ for a fresh $\cvar$, but
  $\reset{\cvctx \cvar {\shift \vark {\app \vark \tm}}} \redcbv \reset{\app
    {\lamp \varx {\reset{\cvctx \cvar \varx}}} \tm}$, but we know that
  $\reset {\app y \tm} \nfbisim \reset{\app {\lamp \varx {\reset {\app y \varx}}}
    \tm}$ holds for all $y$ (Example~\ref{ex:small-step}), so we can conclude
  with $\tr f \cdot$ and Proposition~\ref{p:nfb-in-onfb}.
\end{proof}
\noindent Consequently, $\onfbisim$ is complete {\wrt}$\cpsequiv$, and
$\cpsequiv$ can be used as a proof technique for~$\onfbisim$.

\subsection{Conclusion}

We propose several normal-form bisimilarities for the two semantics of
$\lamshift$. For the relaxed semantics, we define normal-form and refined
bisimilarities which differ in how they handle control-stuck terms; the former
is easier to use but relates less contextually equivalent terms than the
latter. Refined bisimilarity is defined on an extended calculus, where context
variables represent unknown delimited contexts, the same way regular variables
stand for unknown values. We follow the same idea for the original semantics,
where the bisimilarity is defined on pure terms, and extended to any terms
thanks to context variables.

Normal-form bisimulation is already tractable enough that we can prove complex
equivalences with its plain definition (see
Example~\ref{ex:fixed-point}). Proofs can be further simplified thanks to up-to
techniques. Bisimulation up to related contexts is simpler to use than with
environmental bisimilarity, as any term can be plugged in any context. As a
result, the equivalence proof for Turing's combinator is simpler with
normal-form than with environmental bisimilarity (compare
Example~\ref{ex:fixed-point-nf-utc} and Example~\ref{ex:fixed-point-env-utc}).
The downside of normal-form bisimilarity is that it is not complete
{\wrt}contextual equivalence, and fails to relate terms that can be trivially
related with applicative bisimilarity, as witnessed by the terms of
Proposition~\ref{p:cex-dupl}.

% Local variables:
% mode: latex
% TeX-master: "journal.tex"
% End:

\section{Extensions}%
\label{s:extensions}
In this section, we discuss how our results are affected if we
consider other semantics for~$\lamshift$, or if we study other
delimited-control operators, giving directions for future work in the process.

\subsection{Local Reduction Rules}%
\label{ss:local-rules}

In the semantics of Section~\ref{s:calc}, contexts are captured in one
reduction step. Another usual way of computing capture is to use local
reduction rules, where the context is consumed piece by
piece~\cite{Felleisen:POPL88}. Formally, we introduce \emph{elementary
  contexts}, defined as follows:
\begin{grammar}
  & \textrm{Elementary contexts:} \quad & \ectx & \bnfdef \vctx \val \mtctx \bnfor \apctx
  \mtctx \tm
\end{grammar}

\noindent
The reduction rule \RRshift is then replaced with the next two rules.
\[
\begin{array}{rlll}
   \inctx \rctx {\inctx \ectx {\shift \vark \tm}} &
  \hspace{-0.5em}\redcbv\hspace{-0.5em} & \inctx \rctx{\shift
    {\vark'}{\subst \tm \vark {\lam \varx {\reset {\app
            {\vark'}{\inctx \ectx \varx}}}}}} \mbox{ with }
  \varx,\vark' \notin \fv \ectx \cup \fv{\tm} & \RRshifte
  \\[2mm]
  \inctx
  \rctx {\reset {\shift \vark \tm}} &
  \hspace{-0.5em}\redcbv\hspace{-0.5em} & \inctx \rctx {\reset{\subst
      \tm \vark {\lam \varx \varx}}} & \RRshiftid
\end{array}
\]

As we can see in rule \RRshifte, the capture of an elementary context
does not require a \resetId{}, and it leaves the operator \shiftId{}
in place to continue the capture process. The process stops when a
\resetId{} is encountered, in which case the rule~\RRshiftid applies:
the \shiftId{} operator is removed, and its variable $\vark$ is
replaced with the function representing the delimited empty context.
%identity.

With local reduction rules, control stuck terms are of the form
$\shift \vark \tm$ (without any surrounding context). This has major
consequences on the definition of normal-form bisimulations, as it
brings the regular definition (Definition~\ref{d:nf-bisim}) and the
refined one (Definition~\ref{d:rnf-bisim}) closer together, e.g., the
terms of Proposition~\ref{p:cex-stuck} can be proved contextually
equivalent with the regular definition (when phrased in terms of the
local reduction rules).

%with local rules, two stuck terms $\shift \vark \tmzero$ and $\shift
%\vark \tmone$ are normal-form bisimilar if $\reset{\tmzero}$ and
%$\reset{\tmone}$ are normal-form bisimilar.
The resulting bisimulation proofs are arguably more difficult than
with the semantics of Section~\ref{s:calc}, as we can see with the
next example.
\begin{exa}[\AXbetaomega axiom] Assume we want to prove that
  $\app {\lamp \varx {\inctx \ctx \varx}} \tm \nfbisim \inctx \ctx \tm$
  ($\varx \notin \fv \ctx$) with local rules. If $\tm$ is a control stuck term
  $\shift \vark {\tm'}$, we have to relate
  $\reset{\subst {\tm'} \vark {\lam y {\reset {\app {\app {\vark'}{\lamp \varx
              {\inctx \ctx \varx}}} y}}}}$ ($y$, $\vark'$ fresh) with
  $\reset{\tm' \vect{\subs}}$, where $\vect \subs$ are the substitutions we
  obtain as a result of the progressive capture of $\ctx$ by
  $\shift \vark {\tm'}$. We do not need sequences of substitutions with the
  semantics of Section~\ref{s:calc}.
\end{exa}

The theory for applicative and environmental bisimulations is not affected by
using local rules; in particular, we still have to compare control-stuck terms
by putting them in a pure (multi-hole) context. However, a proof using a
small-step bisimulation of any kind becomes tedious with local rules, as they
introduce a lot of redexes (first to capture a whole pure context, and then to
reduce all the produced $\beta$-redexes), and a reduction of each redex has to
be matched in a small-step relation. We, therefore, believe that the reduction
rules of Section~\ref{s:calc} are better suited to proving the equivalence of
two $\lamshift$ terms.

\subsection{Call-by-Name Reduction Semantics}%
\label{ss:cbn}

In call-by-name, arguments are not reduced to values before
$\beta$-reduction takes place. Such a semantics can be achieved by
changing the syntax of (pure) evaluation contexts as follows:
\begin{grammar}
  &\textrm{CBN pure contexts:}& \ctx & \bnfdef  \mtctx \bnfor \apctx \ctx
  \tm \\
  &\textrm{CBN evaluation contexts:}\quad & \rctx & \bnfdef \mtctx \bnfor \apctx \rctx
  \tm \bnfor \resetctx \rctx
\end{grammar}
and by turning the $\beta$-reduction rule into
\[
\begin{array}{rlll}
  \quad \inctx \rctx {\app {\lamp \varx \tmzero} \tmone} & \redcbn & \inctx
  \rctx {\subst \tmzero \varx \tmone} & \quad \RRbetan
\end{array}
\]
The rules $\RRshift$ and $\RRreset$ are the same as in call-by-value,
but their meanings change because of the new syntax for call-by-name
contexts. We still distinguish the relaxed semantics (without
outermost enclosing \resetId{}) from the original semantics.

The results of this paper can be adapted to call-by-name by
transforming values used as arguments into arbitrary terms, for
example when comparing $\lambda$-abstractions with applicative
bisimilarity, or when building testing terms from the environment in
environmental bisimilarity. We can also relate the bisimilarities to
the call-by-name CPS equivalence, which has been axiomatized by
Kameyama and Tanaka~\cite{Kameyama-Tanaka:PPDP10}. The axioms for
call-by-name are the same or simpler than in call-by-value: the axioms
\AXresetshift, \AXresetval, and \AXshiftreset can be proved in
call-by-name using bisimulations with the same proofs as in
call-by-value. The call-by-value axioms \AXbeta, \AXbetaomega,
and~\AXresetlift are replaced by a single axiom for call-by-name
$\beta$-reduction
\[
\app{\lamp \varx \tmzero} \tmone =_{\textrm{\tiny KT}} \subst \tmzero \varx
\tmone,
\]
which is straightforward to prove since the three bisimilarities
contain reduction. Finally, the axiom \AXshiftelim still holds only
for the original semantics.

\subsection{CPS-based Equivalences}%
\label{ss:cps-bisim}

It is possible to go beyond CPS equivalence and use the CPS definition
of \shiftId{} and \resetId{} to define behavioral equivalences in
terms of it: $\tmzero$ and $\tmone$ are bisimilar in $\lamshift$ if
their translations $\cps \tmzero$ and $\cps \tmone$ are bisimilar in
the plain $\lambda$-calculus. As an example, we can define CPS
applicative bisimilarity~$\appbisimcps$ as follows: given two closed
terms~$\tmzero$ and~$\tmone$ of $\lamshift$, we have $\tmzero
\appbisimcps \tmone$ if $\cps\tmzero$ and $\cps\tmone$ are applicative
bisimilar in the call-by-value
$\lambda$-calculus~\cite{Abramsky-Ong:IaC93}. We compare here this
equivalence to the contextual equivalence~$\ctxequivp$ for the
original semantics, since the CPS of Figure~\ref{f:cps-translation} is
valid for that semantics only.

Even if $\appbisimcps$ is sound {\wrt}$\ctxequivp$, we show it is not
complete. A CPS translated term is of the form $\lam {\vark_1\vark_2} \tm$,
where $\vark_1$ and $\vark_2$ stand for, respectively, the continuation and the
metacontinuation of the term, which are $\lambda$-abstractions of a special
shape. But applicative bisimilarity in $\lambda$-calculus compares terms with
any $\lambda$-abstraction, not just a continuation or metacontinuation, making
$\appbisimcps$ over-discriminating compared to $\ctxequivp$. Indeed, let
$\valzero \is \lam \varx \Omega$ and
$\valone \is \lam \varx {\reset {\varx \iapp i} \iapp \Omega}$. We have
$\valzero \ctxequivp \valone$, roughly because~$\valzero$ diverges as soon as it
is applied to a value $\val$, and so does $\valone$, either because
$\reset {\app \val i}$ diverges, or because $\reset {\app \val i}$ converges and
$\Omega$ then diverges (more formally, the relation
$\{(\lam{x}{\Omega},\lam{x}{\app{\reset{\app{x}{i}}}{\Omega}})\} \cup \{
(\Omega, \app{\reset{t}}{\Omega}) \mid t \in \cterms\} \cup \{(\Omega,
\app{\val}{\Omega}) \mid \val \in \cvalues\}$ is an applicative bisimulation,
included in $\ctxequiv$ and, therefore, in~$\ctxequivp$). The CPS translation of
these terms, after some administrative reductions, yields
\begin{align*}
 \cps \valzero &= \lam{\vark_1\vark_2}{\vark_1 \iapp {\lamp x {\cps \Omega}}
                 \iapp \vark_2}, \mbox{ and} \\
  \cps\valone &= \lam{\vark_1\vark_2}{\vark_1 \iapp {\lamp {x\vark'_1\vark'_2}{x \iapp {\lamp y
        {\cps y}} \iapp \gamma \iapp \lamp z {\cps \Omega \iapp \val' \iapp \vark_2}}} \iapp
  \vark_2}
\end{align*}
where $\gamma$ is defined in Figure~\ref{f:cps-translation} and
$\val'$ is some value, the precise definition of which is not
important. If $\val \is \lam {z\vark_2}{z \iapp
  \lamp{x'\vark''_1\vark''_2} i \iapp i \iapp i}$, then $\cps \valzero
\iapp \val \iapp i \rtclo\redcbv \cps\Omega$ and $\cps \valone \iapp
\val \iapp i \rtclo\redcbv i$; the diverging part in $\cps \valone$,
namely $\lam z {\cps \Omega \iapp \val' \iapp \vark_2}$, is thrown
away by $\val$ instead of being eventually applied, as it should be if
$\val$ was a continuation (the term $\lam{x'\vark''_1\vark''_2}{i}$ is
not in the 2-layer CPS).

A possible way to get completeness for $\appbisimcps$ could be to
restrict the target language of the CPS translation to a CPS calculus,
\ie a subcalculus where the grammar of terms enforces the correct
shape of arguments passed as values, continuations, or
metacontinuations (as in,
e.g.,~\cite{Kameyama-Hasegawa:ICFP03}). However, even with
completeness, we believe it is more tractable to work in direct style
with the relations we define in this paper, than on CPS translations
of terms: as we can see with $\valzero$ and $\valone$ above,
translating even relatively simple source terms leads to voluminous
terms in CPS\@. Besides, $\appbisimcps$ compares all translated terms
with a continuation (which corresponds to a context $\ctx$) and a
metacontinuation (which corresponds to a metacontext $\rctx$), while
bisimilarities in direct style need at most a context~$\ctx$ to
compare stuck terms.

Nonetheless, we believe that studying fully the relationship between CPS-based
behavioral equivalences and direct-style equivalences is an interesting future
work. We would like to consider other CPS translations, including a CPS
translation for the relaxed semantics~\cite{Materzok-Biernacki:ICFP11}, or the
1-layer CPS translation for the original
semantics~\cite{Danvy-Filinski:LFP90}. We would also like to know if it is
possible to obtain a CPS-based soundness proof for normal-form bisimilarity, as
in $\lambda$-calculus~\cite{Lassen:LICS05}, to have a complete picture of the
interactions between CPS and behavioral equivalences.

\subsection{The \texorpdfstring{$\lammutphat$}{lambda-mu-tp}-Calculus}%
\label{ss:lammutphat}

The $\lambda\mu$-calculus~\cite{Parigot:LPAR92} contains a $\mu$-construct
%operator which
that can be seen as an abortive control operator. In this calculus, we
evaluate named terms of the form $[\alpha] \tm$, and the names
$\alpha$ are used as placeholders for evaluation contexts. Roughly, a
$\mu$ term $\mu\alpha.[\beta]\tm$ is able to capture its whole (named)
evaluation context $[\gamma] \ctx$, and substitutes $\alpha$ with
$[\gamma] \ctx$ in $[\beta]\tm$. Context substitution is the same as
the one presented in Section~\ref{ss:refined}. In particular, it is
capture-free, \eg in $[\beta]\subst \tm \alpha {[\gamma] \ctx}$, the
free names of $[\gamma] \ctx$ (such as $\gamma$) cannot be bound by
the $\mu$ constructs in $\tm$.

The $\lammutphat$-calculus~\cite{Herbelin-Ghilezan:POPL08} extends the
$\lambda\mu$-calculus by adding a special name \tphat which can be
dynamically bound during a context substitution. Besides, the
$\mu$-operator no longer captures the whole context, but only up to
the nearest enclosing $\mu$-binding of $\tphat$. As a result, a
$\mu$-binding of $\tphat$ can be seen as a delimiter, and in fact, the
$\lammutphat$-calculus simulates
$\lamshift$~\cite{Herbelin-Ghilezan:POPL08}. In particular, their CPS
equivalences coincide. However, defining bisimilarities in \lammutphat
may lead to relations similar to the $\lambda\mu$-calculus
ones~\cite{Biernacki-Lenglet:MFPS14} because of names. Indeed, we have
to compare named values~$[\alpha] \val$ in~\lammutphat, which requires
substituting $\alpha$ with some named context, as in the
$\lambda\mu$-calculus~\cite{Biernacki-Lenglet:MFPS14}. Similarly,
control stuck terms are of the form $[\alpha]\inctx \rctx {\mu
  \tphat.[\beta] \val}$, and a way to relate them would be by
replacing $\beta$ with a context~$[\gamma] \ctx$ or $[\tphat]
\ctx$. It would be interesting to compare the behavioral theories of
$\lamshift$ and $\lammutphat$ to see if the encoding of the former
into the latter is fully abstract (\ie preserves contextual
equivalence).

\subsection{Typed Setting}%
\label{ss:typed}

A type system affects the semantics of a language by ruling out ill-typed terms,
and thus restricts the possible behaviors compared to the untyped
calculus. Applicative~\cite{Gordon:TCS99,Gordon-Rees:POPL96},
normal-form~\cite{Lassen-Levy:CSL07,Lassen-Levy:LICS08}, and
environmental~\cite{Sumii-Pierce:JACM07,Sumii:CSL09} bisimilarities have been
defined for various calculi and type systems. The type systems for \shiftId{}
and \resetId{}~\cite{Danvy-Filinski:DIKU89, Asai-Kameyama:APLAS07} assign types
not only to terms, but also to contexts. Pure contexts $\ctx$ are given types of
the form $\tpctxDF \tpa \tpb$, where $\tpa$ is the type of the hole and $\tpb$
is the answer type, and evaluation contexts (also called metacontexts) $\rctx$
are assigned types of the form $\tpmetactxDF \tpa$, where $\tpa$ is the type of
the hole. A typing judgment $\typejDF \tpenv \tpb \tm \tpa \tpc$ roughly means
that under the typing context $\tpenv$, the term $\tm$ can be plugged into a
pure context~$\ctx$ of type $\tpctxDF \tpa \tpb$ and a metacontext~$\rctx$ of
type $\tpmetactxDF \tpc$, producing a well-typed term
$\inctx \rctx{\reset{\inctx \ctx \tm}}$. In general, the evaluation of $\tm$ may
capture the surrounding context of type $\tpctxDF \tpa \tpb$ to produce a value
of type $\tpc$, with $\tpb \neq \tpc$. Function types also contain extra
information about the contexts the terms are plugged into: a term of type
$\arrowtpDF \tpa \tpb \tpc \tpd$ can be applied to an argument of type $\tpa$
within a pure context of type $\tpctxDF \tpb \tpc$ and a metacontext of type
$\tpmetactxDF \tpd$.

The complexity of the type systems for \shiftId{} and \resetId{}
(compared to, \eg plain $\lambda$-calculus) may have some consequences
on the definition of a typed bisimilarity for the language. In
particular, we wonder how the extra type annotations for pure contexts
and metacontexts should be factored in the bisimilarities. It seems
natural to include types for the pure contexts for control stuck
terms, since pure contexts already occur in the definitions of
applicative and environmental bisimilarities in that case; it is not
clear if and how the types for the metacontexts should be
mentioned. The study of a typed $\lamshift$ can be interesting also to
see how the types modify the equivalences between terms.
%, in particular we wonder if the equivalences proved in this paper
%still hold in the presence of types.
We leave this as a future work.

A related and unexplored topic is defining logical relations to
characterize contextual equivalence for typed calculi with delimited
continuations. So far, Asai introduced logical relation to prove the
correctness of a partial evaluator for \shiftId{} and
\resetId{}~\cite{Asai:TFP05}, whereas Biernacka et al.\ in a series of
articles proposed logical predicates for proving termination of
evaluation in several calculi of delimited
control~\cite{Biernacka-Biernacki:PPDP09,Biernacka-al:PPDP11,Biernacka-al:COS13}. We
expect such logical relations to exploit the notion of context and
metacontext and, therefore, to be
biorthogonal~\cite{Krivine:APAL94,Pitts-Stark:HOOTS98}. Biorthogonal
and step-indexed Kripke logical relations have been proposed for an
ML-like language with \callccId{} by Dreyer et
al.~\cite{Dreyer-al:JFP12} and adapting this approach to a similar
language based on Asai and Kameyama's polymorphic type system for
\shiftId{} and \resetId{}~\cite{Asai-Kameyama:APLAS07} presents itself
as an interesting topic of future research. An alternative to
step-indexed Kripke logical relations that also have been shown to
account for abortive continuations are parametric
bisimulations~\cite{Hur-al:TR14}, built on relation transition systems
of Hur et al.~\cite{Hur-al:POPL12}. Whether such hybrids of logical
relations and bisimulations can effectively support reasoning about
delimited continuations is an open question.

\newpage
\subsection{Other Delimited-Control Operators}%
\label{ss:others}

\subsubsection*{CPS hierarchy}
The operators \shiftId{} and \resetId{} are just an instance of a more
general construct called the \emph{CPS
  hierarchy}~\cite{Danvy-Filinski:LFP90}. As explained in
Section~\ref{ss:cps-equivalence}, \shiftId{} and \resetId{} have been
originally defined by a translation into CPS\@. When iterated, the CPS
translation leads to a hierarchy of continuations, in which it is
possible to define a hierarchy of control operators \shiftnId{i} and
\resetnId{i} ($i \geq 1$) that generalizes \shiftId{} and \resetId{},
and that makes possible to separate computational effects that
should exist independently in a program. For example, in order to
collect the solutions found by a backtracking algorithm implemented
with \shiftnId{1} and \resetnId{1}, one has to employ \shiftnId{2} and
\resetnId{2}, so that there is no interference between searching and
emitting the results of the search. The CPS hierarchy was also
envisaged to account for nested computations in hierarchical
structures~\cite{Biernacka-al:LMCS05}.

In the hierarchy, a \shiftId{} operator of level $i$ captures the
context up to the first enclosing \resetnId{i} or \resetnId{j} with $j
> i$. So for example in $\reset{\inctx \ctxone {\reset{\inctx \ctxzero
      {\shifti 1 \vark \tm}}_2}}_1$, the $\rawshift_1$ captures only
$\ctxzero$, not $\ctxone$. We believe the results of this paper
generalize to the CPS hierarchy without issues. The notions of pure
context and control stuck term now depend on the hierarchy level: a
pure context of level $i$ does not contain a \resetnId{j} (for $j \geq
i$) encompassing its hole, and can be captured by an operator
\shiftnId{i}. A control stuck term of level $i$ is an operator
\shiftnId{i} in a pure context of level $i$. The definitions of
bisimulations have to be generalized to deal with control stuck terms
of level~$i$ the same way we treat stuck terms of level~1. For
example, two control stuck terms of level $i$ are applicative
bisimilar if they are bisimilar when put in an arbitrary level $i$
pure context surrounded by a \resetnId{i}. The proofs for $i=1$ should
carry through to any $i$.

\subsubsection*{Operator \shiftzId{}}
The operator \shiftzId{} ($\rawshift_0$) allows a term to capture a
pure context with its enclosing
delimiter~\cite{Danvy-Filinski:LFP90}. The capture reduction rule for
this operator is thus as follows:
\[
\inctx \rctx{\reset{\inctx \ctx { \shiftz \vark \tm}}} \redcbv \inctx \rctx {\subst
  \tm \vark {\lam \varx {\reset {\inctx \ctx \varx}}}}, \mbox{ with } \varx \notin \fv
\ctx
\]
Note that there is no \resetId{} around $\tm$ in $\inctx \rctx {\subst
  \tm \vark {\lam \varx {\reset {\inctx \ctx \varx}}}}$. Consequently,
a term is able to directly decompose an evaluation context $\rctx$
into pure contexts through successive captures with $\rawshift_0$;
this is not possible in $\lamshift$.

The definitions of bisimilarities of this paper should extend to a
calculus with \shiftzId{} as far as the relaxed semantics is
concerned. Since a term is able to access the context beyond the first
enclosing \resetId{}, contextual equivalence is more discriminating
with \shiftzId{} than in~$\lamshift$. For example, $\reset{\reset
  \tm}$ is no longer equivalent to $\reset \tm$, as we can see by
taking $\tm = \shiftz \vark {\shiftz \vark \Omega}$.

For the original semantics (that in the case of \shiftzId{} assumes a
persistent top-level \resetId{}), the definitions have to take into
account the fact that a delimited term $\reset \tm$ may evaluate to a
control stuck term (like, \eg $\reset{\shiftz \vark {\shiftz \vark
    \tm}}$ for any $\tm$) and that, therefore, it is not sufficient to
compare values with values only. For instance, in order to validate
the following equation taken from the axiomatization of
\shiftzId{}~\cite{Materzok:CSL13}:
\[
\shiftz \vark {\reset {\app {\lamp \varx {\shiftz {\vark'}{\app \vark \varx}}}
    \tm}} =_{\textrm{\tiny M}} \tm, \mbox{ with } \vark \notin \fv \tm
\]
%% Consequently, a sound and complete bisimilarity for the original
%% semantics
we would have to be able to compare normal forms of different
kinds, which can be achieved
%% as for the original semantics of \shiftId{}, i.e.,
by putting the normal forms in a context $\reset \ctx$ for any $\ctx$.

%% Consequently, a sound bisimilarity for the original semantics cannot
%% compare values only with values: it should compare two normal forms
%% (of the same kind or not) by putting them in a context $\reset \ctx$
%% for any $\ctx$.

%% But then the resulting bisimilarity is almost as
%% complex as contextual equivalence, and we wonder how useful it can be.

\subsubsection*{Operators \controlId{} and \promptId{}}
The \controlId{} operator ($\rawcontrol$) captures a pure context up
to the first enclosing \promptId{} ($\rawprompt$), but the captured
context does not include the
delimiter~\cite{Felleisen:POPL88}. Formally, the capture reduction
rule is as follows:
\[
\inctx \rctx {\prompt {\inctx \ctx {\control \vark \tm}}} \redcbv \inctx \rctx
{\prompt {\subst \tm \vark {\lam \varx {\inctx \ctx \varx}}}}, \mbox{ with }
\varx \notin \fv \ctx
\]
Unlike with \shiftId{} and \resetId{} where continuation composition
is static, with \controlId{} and \promptId{} it is dynamic, in the
sense that the extent of control operations in the captured context
comprises the context of the resumption of the captured
context~\cite{Biernacki-al:SCP06}.
%% Unlike with \shiftId{} and \resetId{}, successive controls operators may capture
%% a bigger and bigger context. E.g., in a term similar to
%% Example~\ref{e:reduction}, the second control would capture $\app i
%% {\apctx {(\vctx i \mtctx)} {\appp \omega \omega}}$ instead of $\apctx
%% {(\vctx i \mtctx)} {\appp \omega \omega}$.
A \controlzId{} variant also exists~\cite{Shan:HOSC07}, where the
delimiter is captured with the context but not kept: as a result, no
delimiter is present in the right-hand side of the capture reduction
rule.
%% In contrast with \shiftId{} and \resetId{}, control and prompt (as well as
%% control$_0$ and prompt$_0$) are not defined using CPS, and therefore
%% CPS equivalence cannot be defined for these calculi.

The theory of this paper should extend to \controlId{} and \promptId{}
with minor changes. However, studying this calculus would still be
interesting to pinpoint the differences between the equivalences of
\shiftId{}/\resetId{} and \controlId{}/\promptId{}. For example,
$\prompt{\prompt \tm}$ is equivalent to $\prompt \tm$, the same way
$\reset{\reset \tm}$ is equivalent to $\reset \tm$. In fact, we
conjecture the axioms can still be proved equivalent if we replace
\shiftId{} and \resetId{} with \controlId{} and \promptId{} (with the
same restriction for \AXshiftelim). In contrast, $\tmzero \is
\app{(\shift {\vark_1}{\app{\app {\vark_1}{\lamp \varx {\shift{\vark_2
          } \tm}}} \Omega})} \val$ (where $\vark_1,\vark_2, \varx
\notin \fv \tm$) is equivalent to $\shift \vark \Omega$ (because
$\tmzero \lts\ctx \clocbv \reset{\app {\reset{\inctx \ctx {\shift
        {\vark_2} \tm}}} \Omega}$, and this term always diverges), but
the term $\tmzero' \is \app{(\control {\vark_1}{\app{\app
      {\vark_1}{\lamp \varx {\control{\vark_2 } \tm}}} \Omega})} \val$
is equivalent to $\prompt{\tm}$ (because $\tmzero' \lts\ctx \clocbv
\prompt{\app {\inctx \ctx {\shift {\vark_2} \tm}} \Omega} \redcbv
\prompt \tm$). Maybe we can find (general enough) laws which hold with
\controlId{} and \promptId{} but not with \shiftId{} and \resetId{},
and conversely.

\subsubsection*{Multiple prompts}
In languages with (named) multiple
prompts~\cite{Gunter-al:FPCA95,Dybvig-al:JFP06,Downen-Ariola:ESOP12}
control delimiters (prompts) as well as control operators are tagged
with names, so that the control operator captures the evaluation
context up to the dynamically nearest delimiter with the matching
name. In a calculus with tagged \shiftId{} ($\rawshiftt{a}$) and \resetId{}
($\rawresett{a}$) the operational semantics of \shiftId{} is given by the
following rule:
\[
\inctx {\rctx'}{\resett{a}{\inctx \rctx { \shiftt{a}{\vark}{\tm}}}}
\redcbv \inctx {\rctx'} {\resett{a}{\subst \tm \vark {\lam \varx
      {\resett{a}{\inctx \rctx \varx}}}}}, \mbox{ with } a \notin
\prompts{\rctx} \mbox{ and }\varx \notin \fv \ctx
\]
where $\prompts{\rctx}$ is the set of the prompts guarding the hole of
$\rctx$. Such calculi resemble the CPS hierarchy, already considered
in this section, however there are differences in their semantics. In
contrast to the CPS hierarchy, where evaluation contexts form a
hierarchy\footnote{In the original semantics, the evaluation context
  of level $i+1$ is a list (a stack, really) of evaluation contexts of
  level $i$ separated by control delimiters of level $i$ (contexts of
  level $1$ are just the standard CBV evaluation contexts.)  and the
  number of context layers is fixed~\cite{Biernacka-al:LMCS05}.} and
the extent of control operations of level $i$ is limited by control
delimiters of any level $j \geq i$, in the calculus with multiple
prompts the evaluation context is a list of the standard CBV
evaluation contexts separated by named prompts and the control
operations reach across any prompts up to a matching one. Moreover,
the salient and unique feature of such calculi is dynamic name
generation that allows one, e.g., to eliminate unwanted interactions
between the control operations used to implement some control
structure (e.g., coroutines) and the control operations of the code
that uses the control structure.

Even without dynamic name generation, which gives an additional
expressive power to such calculi, calculi with multiple prompts
generalize, e.g., simple exceptions~\cite{Gunter-al:FPCA95} and the
\catchId{}/\throwId{} constructs~\cite{Crolard:JFP99}. The results of
this article can be seamlessly adapted to these calculi and most, if
not all, of the presented techniques should carry over without
surprises.

However, when dynamic name generation is included in the calculus, comparing two
terms becomes more difficult, as prompts with the same purpose can be generated
with different names. With environmental bisimilarity, we can use environments
to remember the relationships between generated prompts. We do so
in~\cite{Aristizabal-al:LMCS17} and define sound and complete environmental
bisimilarities and their up-to techniques for a calculus with dynamically
generated prompt names. Resource generation makes the definition of a sound
applicative bisimilarity difficult for such a calculus, as argued
in~\cite{Koutavas-al:MFPS11}.

\subsection{Other Constructs}%
\label{ss:constructs}

Here, we briefly discuss what happens when $\lamshift$ is extended
with constructs that can be found in usual programming languages.

\subsubsection*{Constants} While adding constants (such as numerals, booleans,
\ldots) to the language does not raise any issue for
applicative~\cite{Gordon:TCS99} and environmental bisimilarities,
defining a satisfactory normal-form bisimilarity in the presence of
constants raises some difficulties~\cite{Stoevring-Lassen:POPL07}:
e.g., it is not clear how to define a normal-form bisimulation which
equates $x+y$ and $y+x$. Relying on encodings of constants into plain
$\lambda$-calculus is not enough, as these encodings usually do not
respect the properties of the constants, like, for example,
commutativity of~$+$. These problems are orthogonal to the presence of
control operators though.

\subsubsection*{Store} Bisimilarities for languages with store are usually of
the environmental
kind~\cite{Sangiorgi-al:TOPLAS11,Koutavas-Wand:POPL06,Sumii:CSL09},
and~\cite{Koutavas-al:MFPS11} argues that the usual form of
applicative bisimilarity is not sound in the presence of store.
St{\o}vring and Lassen define a sound and complete normal-form
bisimilarity for~$\lambda\mu\rho$~\cite{Stoevring-Lassen:POPL07}, a
calculus with store and an abortive control construct inspired by
Parigot's $\lambda\mu$~\cite{Parigot:LPAR92}. Their work largely
relies on the fact that in $\lambda\mu$-calculus (and in
$\lambda\mu\rho$ as well), terms are of the form~$[a]\tm$, where the
name $a$ acts as a placeholder for an evaluation context. These names
are also essential to be able to define a sound and complete
applicative bisimilarity for
$\lambda\mu$~\cite{Biernacki-Lenglet:MFPS14}. Developing a sound
behavioral theory of $\lamshift$ extended with higher-order store,
potentially taking advantage of context variables, is of interest as a
future work.

%% While developing a sound behavioral theory of $\lamshift$ extended
%% with higher-order store is of interest as a future work, we cannot
%% hope to obtain complete normal-form bisimilarity in this case. The
%% reason is that in a calculus with the \shiftId{} operator, a
%% computation once delimited remains delimited until it terminates, and
%% the presence of store does not influence this property, so, e.g.,
%% terms $\app{\reset{\app{x}{y}}}{\Omega}$ and $\Omega$ would be
%% contextually equivalent in the extended calculus, but not normal-form
%% bisimilar (one is open-stuck and the other diverges).

%% We do not have such names in $\lamshift$, so it is not clear that
%% adding store to the calculus would lead to a sound and complete
%% normal-form bisimilarity as in $\lambda\mu\rho$. We might have to rely
%% on context variables, as in Sections~\ref{ss:refined}
%% and~\ref{ss:nf-original}. We leave this study as a future work.

\subsubsection*{Exceptions} Like for store, Koutavas et
al.~\cite{Koutavas-al:MFPS11} give examples showing that applicative
bisimilarity is not sound for a calculus with exceptions, and
environmental bisimilarity should instead be used. Studying an
extension of $\lamshift$ with exceptions would be interesting to
compare the encoding of exceptions using \shiftId{} and
\resetId{}~\cite{Filinski:POPL94} to the native constructs. We leave
this as a future work.

% Local variables:
% mode: latex
% TeX-master: "journal.tex"
% End:

\section{Conclusion}%
\label{s:conclusion}

In our study of the behavioral theory of a calculus with \shiftId{} and
\resetId{}, we consider two semantics: the original one, where terms are
executed within an outermost \resetId{}, and the relaxed one, where this
requirement is lifted. For each, we define a contextual equivalence
(respectively $\ctxequivp$ and $\ctxequiv$), that we try to characterize with
different kinds of bisimilarities (normal-form $\nfbisim$, $\rbisim$,
$\onfbisim$, applicative~$\appbisim$, and environmental $\envbisim$,
$\envbisimp$). We also compare our relations to CPS equivalence~$\cpsequiv$, a
relation which equates terms with $\beta\eta$-equivalent CPS translations. We
summarize in Figure~\ref{f:conclusion} the relationships between these
relations.

\begin{figure}
  \[
  \begin{array}{rrcccl}
    \mbox{relaxed semantics: } & \nfbisim \mathop\subsetneq & \hspace{-0.75em} \rbisim &
                                                                                         \hspace{-0.75em} \mathop\subsetneq
    & \hspace{-0.75em} \ctxequiv &  \hspace{-0.75em} \mathop= \appbisim  \mathop=
                                   \envbisim \\
                               && \hspace{-0.75em} \rotatebox[origin=c]{-90}{$\subsetneq$} &
                                                                                       &
                                                                                         \hspace{-0.75em} \rotatebox[origin=c]{-90}{$\subsetneq$} \\
    \mbox{original semantics: } & \cpsequiv \mathop\subsetneq & \hspace{-0.75em}
                                                                \onfbisim & \hspace{-0.75em} \mathop\subsetneq
    & \hspace{-0.75em} \ctxequivp & \hspace{-0.75em} \mathop= \envbisimp \\
  \end{array}
  \]

  \caption{Relationships between the equivalences of $\lamshift$}%
\label{f:conclusion}
\end{figure}

Normal-form bisimulation arguably leads to the simplest equivalence proofs in
most cases; by essence, the lack of quantifications on testing entities in its
definition leads to simpler proof obligations. More importantly, it benefits
from up-to techniques which manipulate contexts, which is of prime importance in
a calculus where context capture is part of the semantics. Environmental
bisimilarity also allows for such context manipulating up-to techniques, albeit
in a less general form. As a result, proving that Turing's fixed-point
combinator is bisimilar to its \shiftId{}/\resetId{} variant can be done using
up-to techniques for normal-form bisimilarity with only three pairs
(Example~\ref{ex:fixed-point-nf-utc}), while it requires an inductively defined
candidate relation with environmental bisimulation
(Example~\ref{ex:fixed-point-env-utc}). The lack of such a powerful up-to
technique for applicative bisimulation makes it clearly less tractable than the
other styles, as witnessed by the proofs for Turing's combinator
(Example~\ref{ex:fixed-point-app}), or for the~\AXbetaomega axiom
(Proposition~\ref{p:omega-app} vs Propositions~\ref{p:omega-nf}
and~\ref{p:omega-env}).

However, normal-form bisimulation cannot be used to prove all equivalences,
since its corresponding bisimilarity is not complete. It can be too
discriminating to relate very simple terms, like those in
Propositions~\ref{p:cex-stuck} and~\ref{p:cex-dupl}, even though refined
normal-form bisimulation (Section~\ref{ss:refined}) can help. In contrast,
applicative and environmental bisimilarities are complete, and can be used as
alternatives when normal-form bisimulation fails.

To summarize, to prove that two given terms are equivalent, we would suggest to
first try normal-form bisimulation, and if it fails, try next environmental
bisimulation. Applicative bisimulation should be used only in the simplest
cases, such as terms similar to those of
Proposition~\ref{p:cex-nf-completeness}. The relations for the relaxed semantics
can also be used as proof techniques for the original semantics, except in cases
similar to the \AXshiftelim axiom, where only the equivalences dedicated to the
original semantics can be used.

% Local variables:
% mode: latex
% TeX-master: "journal.tex"
% End:

\bibliographystyle{abbrv}
\bibliography{journal}

\iftoggle{withoutApp}{}{
\appendix
\renewcommand*{\thesection}{\Alph{section}}

\section{Proofs Sketches for Normal-Form Bisimilarity}%
\label{a:soundness-nf}

We only sketch the progress proofs for the refined bisimilarity and for the
original semantics as they are very similar to the proofs for the relaxed
semantics~\cite{Biernacki-al:MFPS17}. We start with the proof sketches for the
original semantics, which exhibits the most differences.

\begin{lem}%
  \label{l:progress-strong}
  For all $f \in \strong \setF$, $f \sevolve {\widehat{\mathsf{strong}(\setF)}}^\omega,
    {\fid\setF}^\omega$.
\end{lem}

\begin{proof}
  Let $\rel \pprogresso \rels, \relt$. The $\rawrefl$ case is
  straightforward.

  Let $\prgzero \utred\rel \prgone$ with $\prgzero \clocbv \prgzero'$,
  $\prgone \clocbv \prgone'$, and $\prgzero' \rel \prgone'$. If
  $\prgzero = \prgzero'$ (in particular, if $\prgzero$ is a normal form), then
  any test on $\prgzero$ is matched by $\prgone$: $\prgone \clocbv \prgone'$
  and $\prgzero \rel \prgone'$, so we can conclude with the progress hypothesis
  on $\rel$. If $\prgzero \neq \prgzero'$, then $\prgzero \redcbv \prgzero''$
  for some $\prgzero''$. Then $\prgzero'' \clocbv \prgzero'$, therefore we have
  $\prgzero'' \utred \rel \prgone$, which implies $\prgzero'' \utred\relt
  \prgone$ (because $\rel \mathop\subseteq \relt$ by definition of progress).

  Let
  $\cvctx \cvar {\lam \varx \prgzero} \utlam \rel \cvctx \cvar {\lam \varx
    \prgone}$ with $\prgzero \rel \prgone$. The terms are context-stuck, and we
  have $\reset \mtctx \utrefl\relt \reset \mtctx$, and
  $\cvctx {\cvar'}{\lamp x \prgzero \iapp y} \utred{\utcvar{\utsubst\relt}}
  \cvctx {\cvar'}{\lamp x \prgone \iapp y}$, for any fresh $\cvar'$ and $y$. The
  case
  $\cvctx \cvar {\shift \vark \prgzero} \utlam \rel \cvctx \cvar {\shift \vark
    \prgone}$ with $\prgzero \rel \prgone$ is similar.

  Let $\cvctx \cvar \prgzero \utcvar \rel \cvctx \cvar \prgone$ with
  $\prgzero \rel \prgone$. Either $\prgzero \redcbv \prgzero'$ and we progress to
  $\utcvar\relt$, or $\prgzero$ is a normal form. Then $\cvctx \cvar \prgzero$
  is also a normal form, and the result is easy to verify for each of them.

  Let $\subst \prgzero \varx \valzero \utsubst\rel \subst \prgone \varx \valone$
  with $\prgzero \rel \prgone$ and $\valzero \nfv\rel \valone$. The interesting
  case is when $\prgzero = \inctx \rctxzero {\app \varx {w_0}}$. Because
  $\prgzero$ is pure, in fact $\rctxzero = \inctx {\rctxzero'}{\reset\ctxzero}$
  for some $\rctxzero'$ and $\ctxzero$. Therefore,
  $\prgone \evalcbv \inctx {\rctxone'}{\reset {\inctx \ctxone {\app \varx
            {w_1}}}}$, with
  $\inctx{\rctxzero'}{\reset{\ctxzero}} \nfc\relt \inctx{\rctxone'}{\reset
      \ctxone}$ and $w_0 \nfv\relt w_1$. We have
  $\cvctx \cvar {\valzero \iapp y} \rel \cvctx \cvar {\valone \iapp y}$ for
  fresh $\cvar$ and $y$; we distinguish two cases. If
  $\cvctx \cvar {\valzero \iapp y} \redcbv \prgzero''$, then there exists
  $\prgone''$ such that $\cvctx \cvar {\valone \iapp y} \clocbv \prgone''$ and
  $\prgzero'' \relt \prgone''$. Then
  $\subst \prgzero \varx \valzero \redcbv \inctx {\subst {\rctxzero'} \varx
    \valzero}{\subst {\subst
      {\prgzero''} y {\subst {w_0} \varx \valzero}} \cvar {\reset {\subst \ctxzero \varx \valzero}}}$ and
  $\subst \prgone \varx \valone \clocbv \inctx {\subst {\rctxone'} \varx \valone}{\subst {\subst
      {\prgone''} y {\subst {w_1} \varx \valone}} \cvar {\reset {\subst \ctxone
        \varx \valone}}}$; the resulting terms are in
  $\relt$ up to $\rawutcsubst$, $\rawutsubst$, and $\rawectxpure$. Otherwise,
  $\cvctx \cvar {\valzero \iapp y}$ is an open-stuck term; then there exists
  $\prgone''$ such that $\cvctx \cvar {\valone \iapp y} \evalcbv \prgone''$ and
  $\cvctx \cvar {\valzero \iapp y} \nf\relt \prgone''$. We can conclude as in
  the first case.
\end{proof}

\begin{lem}
  $\rawectxpure \evolve {\widehat{\mathsf{strong}(\setF)}}^\omega \compo
  {\fid\setF} \compo
  {\widehat{\mathsf{strong}(\setF)}}^\omega,{\fid\setF}^\omega$
\end{lem}

\begin{proof}[Sketch]
  Let $\rel \pprogresso \rel, \relt$. Let
  $\inctx \rctxzero \prgzero \utectxpure\rel \inctx \rctxone \prgone$ with
  $\prgzero \rel \prgone$ and
  $\inctx \rctxzero \varx \rel \inctx \rctxone \varx$ for a fresh $\varx$. The
  cases $\prgzero \redcbv \prgzero'$, and $\prgzero$ is a context-stuck or an
  open-stuck term are easy to check. If $\prgzero = \valzero$, then there exists
  $\valone$ such that $\prgone \evalcbv \valone$ and
  $\valzero \nfv\rel \valone$. Then
  $\inctx \rctxone \prgone \clocbv \inctx \rctxone \valone$, and
  $\inctx \rctxzero \valzero = \subst {\inctx \rctxzero \varx} \varx \valzero
  \utsubst\rel \subst {\inctx \rctxone \varx} \varx \valone = \inctx \rctxone
  \valone$, so we can conclude with  Lemma~\ref{l:progress-strong}.
\end{proof}

\begin{lem}
  $\rawutcsubst \evolve {\widehat{\mathsf{strong}(\setF)}}^\omega \compo
  {\fid\setF} \compo
  {\widehat{\mathsf{strong}(\setF)}}^\omega,{\fid\setF}^\omega$
\end{lem}

\begin{proof}[Sketch]
  Let $\rel \pprogresso \rel, \relt$. Let
  $\subst \prgzero \cvar {\dctxzero} \utcsubst\rel \subst \prgone \cvar
  {\dctxone}$ with $\prgzero \rel \prgone$ and
  $\inctx \dctxzero \varx \rel \inctx \dctxone \varx$ for a fresh $\varx$. The
  interesting case is when
  $\prgzero = \inctx \rctxzero {\cvctx \cvar \valzero}$. There exists
  $\rctxone$, $\valone$ such that
  $\prgone \evalcbv \inctx \rctxone {\cvctx \cvar \valone}$,
  $\inctx \rctxzero y \relt \inctx \rctxone y$ for a fresh $y$, and
  $\valzero \nfv\rel \valone$. The possible reductions of
  $\subst \prgzero \cvar \dctxzero$ comes from $\inctx \dctxzero \valzero$, but
  we have $\inctx \dctxzero \valzero \utsubst\rel \inctx \dctxone
  \valone$. Suppose $\inctx \dctxzero \valzero \redcbv \prgzero'$. By
  Lemma~\ref{l:progress-strong}, there exists $\prgone'$ such that
  $\inctx \dctxone \valone \redcbv \prgone'$ and $\prgzero' \mathrel {\of
    {{\fid\setF}^\omega} \relt} \prgone'$. Then $\subst \prgzero \cvar \dctxzero
  \redcbv \subst {\inctx \rctxzero {\prgzero'}} \cvar \dctxzero$, $\subst \prgone \cvar \dctxone
  \clocbv \subst {\inctx \rctxone {\prgone'}} \cvar \dctxone$, and the resulting
  terms are in $\of{{\fid\setF}^\omega} \relt$ up to $\rawectxpure$ and $\rawutcsubst$,
  i.e., in $\of {{\fid\setF}^\omega} \relt$, as wished.
\end{proof}

For the relaxed semantics, we discuss only the interesting cases where
control-stuck terms can be produced, which are $\rawutshift$ and $\rawpctx$.

\begin{lem}
  $\rawutshift \sevolve {\widehat{\mathsf{strong}(\setF)}}^\omega,
  {\fid\setF}^\omega$.
\end{lem}

\begin{proof}
  Let $\rel \pprogressr \rels, \relt$ and $\tmzero \rel \tmone$. For fresh
  $\cvar$ and $\varx$, we have to relate
  $\subst \tmzero \vark {\lam \varx {\cvctx \cvar \varx}}$ and
  $\subst \tmone \vark {\lam \varx {\cvctx \cvar \varx}}$, which is direct with
  $\rawrefl$ and $\rawutsubst$.
\end{proof}

\begin{lem}
  $\rawpctx \evolve {\widehat{\mathsf{strong}(\setF)}}^\omega \compo {\fid\setF}
  \compo {\widehat{\mathsf{strong}(\setF)}}^\omega,{\fid\setF}^\omega$
\end{lem}

\begin{proof}[Sketch]
  Let $\rel \pprogressr \rel, \relt$. Let
  $\inctx \ctxzero \tmzero \utectxpure\rel \inctx \ctxone \tmone$ with
  $\tmzero \rel \tmone$,
  $\tmzero = \inctx {\ctxzero'}{\shift \vark {\tmzero'}}$, and
  $\inctx \ctxzero \varx \rel \inctx \ctxone \varx$ for a fresh $\varx$. Let
  $\cvar$, $\cvar'$ be fresh context variables. There exist $\ctxone'$ and
  $\tmone'$ such that
  $\tmone \evalcbv \inctx {\ctxone'}{\shift \vark {\tmone'}}$ and
  $\reset {\subst {\tmzero'} \vark {\lam y {\cvctx {\cvar'}{\inctx {\ctxzero'}
          y}}}} \relt \reset {\subst {\tmone'} \vark {\lam y {\cvctx
        {\cvar'}{\inctx {\ctxone'} y}}}}$. Then
  $\subst {\reset {\subst {\tmzero'} \vark {\lam y {\cvctx {\cvar'}{\inctx
            {\ctxzero'} y}}}}} {\cvar'} {\cvctx \cvar \ctxzero} \utcsubst\relt
  \subst {\reset {\subst {\tmone'} \vark {\lam y {\cvctx {\cvar'}{\inctx
          {\ctxone'} y}}}}} {\cvar'} {\cvctx \cvar \ctxone}$, from which we get
  $\reset {\subst {\tmzero'} \vark {\lam y {\cvctx {\cvar}{\inctx \ctxzero
          {\inctx {\ctxzero'} y}}}}} \relt \reset {\subst {\tmone'} \vark {\lam
      y {\cvctx {\cvar}{\inctx \ctxone {\inctx {\ctxone'} y}}}}}$, as wished.
\end{proof}

% Local variables:
% mode: latex
% TeX-master: "journal.tex"
% End:

}

\end{document}